%% file: ThesisGlaucia.tex
\documentclass[oneside,openany,12pt,a4paper]{book}
\usepackage{amsmath,amsthm,amsfonts,amssymb,amscd}
\usepackage[utf8]{inputenc}
\usepackage{indentfirst}
\usepackage{graphicx}
\usepackage{color}
\usepackage[T1]{fontenc}
\usepackage{mathpazo}
\usepackage[a4paper]{geometry}
\usepackage[USenglish]{babel}
\usepackage[colorlinks=true,menucolor=black]{hyperref}
\usepackage{hyperref}
\usepackage{mathtools}
\usepackage{bm}
\usepackage{ulem}
\usepackage{enumerate}
\usepackage{sidecap}
\usepackage{amsmath}
\newcommand{\de}[1]{\left(#1\right)}
\newcommand{\De}[1]{\left[#1\right]}
\newcommand{\DE}[1]{\left\{#1\right\}}
\newcommand{\ie}{\textit{i.e. }}
\newcommand{\eg}{\textit{e.g. }}

\newcommand{\bra}[1]{\left\langle #1 \right|}
\newcommand{\ket}[1]{\left| #1 \right\rangle}
\newcommand{\braket}[2]{\left\langle #1 \middle| #2 \right\rangle}
\newcommand{\ketbra}[2]{\left|#1\middle\rangle\middle\langle#2\right|}
\newcommand{\proj}[1]{|#1\rangle\!\langle #1 |} 
\newcommand{\I}{\mathbb{1}}
\newcommand{\1}{\mathbf{1}}
\newcommand{\Hi}{\mathcal{H}}
\newcommand{\half}{\mbox{$\textstyle \frac{1}{2}$}}

\def\C{{\mathbb{C}}}
\def\R{{\mathbb{R}}}
\def\N{{\mathbb{N}}}
\newcommand{\mean}[1]{\left\langle#1\right\rangle}
\newcommand{\corre}[1]{\langle#1\rangle}
\def\norm#1{ {|\hspace{-.022in}|#1|\hspace{-.022in}|} }

\def\NORM#1{ {\left|\hspace{-.022in}\left| #1 \right|\hspace{-.022in}\right|} }

\newcommand{\boxp}{\vec{P}(a,b|x,y)}

\newtheorem{theorem}{Theorem}[section]
\newtheorem{definition}{Definition}[section]
\newtheorem{proposition}{Proposition}[section]
\newtheorem{lemma}{Lemma}[section]
\newtheorem{corollary}{Corollary}[section]

\newtheorem{Atheorem}{Theorem}[chapter]
\newtheorem{Adefinition}{Definition}[chapter]

\makeatletter
\newenvironment{subdefinition}[1]{%
  \def\subdefinitioncounter{#1}%
  \refstepcounter{#1}%
  \protected@edef\theparentnumber{\csname the#1\endcsname}%
  \setcounter{parentnumber}{\value{#1}}%
  \setcounter{#1}{0}%
  \expandafter\def\csname the#1\endcsname{\theparentnumber.\alph{#1}}%
  \ignorespaces
}{%
  \setcounter{\subdefinitioncounter}{\value{parentnumber}}%
  \ignorespacesafterend
}
\makeatother
\newcounter{parentnumber}

\mathchardef\ordinarycolon\mathcode`\:
\mathcode`\:=\string"8000
\def\vcentcolon{\mathrel{\mathop\ordinarycolon}}
\begingroup \catcode`\:=\active
  \lowercase{\endgroup
  \let :\vcentcolon
  }
  
 \newcommand{\defeq}{\vcentcolon=}
\newcommand{\eqdef}{=\vcentcolon}

\newcommand*\xbar[1]{%
  \hbox{%
    \vbox{%
      \hrule height 0.2pt 
      \kern0.5ex
      \hbox{%
        \kern-0.1em
        \ensuremath{#1}%
        \kern-0.1em
      }%
    }%
  }%
}

\newcommand{\xor}{{\sc{xor}}}
\newcommand{\bi}{\mathcal{B}}
\DeclareMathOperator{\rank}{rank}

\newcommand{\Tr}{\operatorname{Tr}}

\newcommand{\st}{\;\; \text{s.t.}}

\newcommand{\etal}{\textit{et al}.}

\newcommand{\inp}{Q}
\newcommand{\out}{\mathcal{O}}
\newcommand{\oA}{\mathbf{A}}
\newcommand{\oB}{\mathbf{B}}
\newcommand{\gr}{\mathcal{G}}

\newcommand{\pr}{\mathbf{P}}

\begin{document}
\input{title}
\include{quote}
\frontmatter
\include{acknowledge}

\include{abstract}
\include{papers}

{\hypersetup{linkcolor=black}
\tableofcontents
}

\renewcommand\emph{\textbf}
\include{intro}
\mainmatter

\include{NL}
\include{complexity}
\include{games}

\include{graph}
\include{xorbound}

\include{xordbound}

\include{multiplayerbound}
\include{conclu}

\include{MQ}
\include{Absstatebound}

\include{apendix}

\normalem 
\backmatter
\bibliographystyle{myalphaabbr2}
{\small{
\bibliography{biblio.bib}}}

\end{document}

%% file: title.tex
\begin{titlepage}

\begin{center}
{\Huge{\textbf{Bounds on quantum nonlocality}}}\\[7em]


{\Large{PhD Thesis}}\\[9em]

{\large{Gláucia Murta Guimarães}}\\[12em]
\end{center}

\begin{flushright}
\begin{minipage}{6cm}
Tese apresentada ao Programa de Pós-Graduação em Física da 
Universidade Federal de Minas Gerais, como requisito parcial para a
obtenção do título de Doutora em Física.
\end{minipage} 
\end{flushright}

\vfill
\begin{center}
{\large{February 2016}}
\end{center}
\end{titlepage}

%% file: quote.tex
\thispagestyle{empty}
\phantom{q}
\vfill

\begin{flushright}
\begin{minipage}{6.5cm}
\textit{``Que beleza é conhecer o desencanto\\
e ver tudo bem mais claro no escuro''}
\\[5pt]
\rightline{{\rm --- Tim Maia}}
\end{minipage} 
\end{flushright}

%% file: acknowledge.tex
\chapter*{Acknowledgements}


Em um país onde educação superior é um provilégio para poucos, começo agradecendo aos meus pais, Nilma e Francisco,
que sempre se empenharam em encontrar os meios  para garantir que eu tivesse todas as oportunidades para chegar até aqui.

Secondly, I thank my supervisor Marcelo Terra Cunha for almost 6 years of supervision. I am grateful for all the  
stimulating discussions we had during this time and for all you taught me about research and life.
For your guidance through my first research steps, for encouragement during important decisions
and for
always being ready to support me in the difficult (academic and personal) moments\footnote{E claro, pelos açaís, pelos cafés, e 
por me apresentar várias músicas que hoje são parte
da minha trilha sonora favorita.}.

Many thanks to Fernando Brandão (again) for the first course in quantum information, and for playing a fundamental rule in the international opportunities 
that I had during the PhD.

I am grateful to Micha\l{} Horodecki for receiving me for a one-year `sandwich' PhD in KCIK. 
This time was extremely valuable for my development as a researcher
and crucial for many of the results presented in this thesis.

From the great researchers I interacted with during my PhD I would like to specially thank:
Daniel Cavalcanti, for being a role model since my Master's;
Adán Cabello, for participating in my first steps 
in science and for all enthusiastic 
discussions;
Marcin Paw\l{}owski, for always transmitting such a great excitement about science, for great trip adventures and funny discussions, 
and for the opportunity to 
come back to KCIK for a few more months\footnote{And of course, for making me watch the Indiana Jones 
trilogy with Polish dubs and English subtitles.}; 
Karol Horodecki, for a really nice collaboration from which I learned a lot; 
and Pawe\l{} Horodecki, for many enlightening discussions.

A lot of the work presented in this thesis would not have been possible if it was not for 
the careful 
guidance of Ravishankar Ramanathan. I am grateful for all I have learned from you during this time. Much of the researcher I
became was shaped by our collaboration.

Dzi\k{e}kuj\k{e} bardzo to all the people from KCIK for such nice atmosphere at work.
In special I thank all the inhabitants of room `sto dwa' for all the physics and specially the non-physics discussions. 
Pankaj Joshi, thanks for all the ``sweet food" you made me try, and Pawe\l{} Mazurek, thanks for the dance!
And to Ania and Czarek, thanks for making me feel home 10.000 km away.
This year in Poland was an amazing time from which I carry many good memories.

A todos os integrantes do Departamento de Física da UFMG, muito obrigada por serem minha segunda casa nos últimos 10 anos. 
À PG Física por todo apoio administrativo e acadêmico. À Shirley por estar sempre pronta para nos assistir e por tornar a Biblioteca
do Departamento de Física um lugar de apoio para a pesquisa e ensino 
realizados no departamento. 

Às mulheres do departamento de Física, obrigada por serem exemplo e ins-piração.

Agradeço aos integrantes do corredor do doutorado e aos vizinhos da astro-física pela ótima convivência, e em especial aos que contribuíram para que as edições do `Bar da Gláucia' fossem um sucesso!
Aos meus colegas de Sala,  Mangos, Mychel e Rapaiz, muito obrigada pela melhor sala de todas (e me desculpem por todas as vezes que troquei vocês
pelo ar condicionado). Em especial agradeço ao Mangos, pelo ombro amigo em muitos momentos difíceis.

Aos Terráqueos contemporâneos: Bárbara Amaral\footnote{Exemplo de irmã mais velha.}, 
Cristhiano Duarte\footnote{Certamente contribuiu para me tornar uma
pessoa menos pura. Que horror!}, Na-tália Móller\footnote{Fico feliz de você ser minha primeira co-autora!}, Leonardo Guerini\footnote{Léo, 
valeu por todas as
nossas conversas não-locais sobre a vida :)}, Gabriel Fagundes\footnote{Gabriel! Nunca dá pra trás num evento, e ainda traz a caixa de som.}, 
Marcello Nery\footnote{Marcellooow, ainda tô esperando você se redimir por
 não ir na(s) minha(s) festa(s).}, Jessica Bavaresco\footnote{Viu, obrigada pela amizade, pelas saídas em BH e pela super força na reta final!},  
Tassius Maciel\footnote{Fonte das histórias mais trolls que eu já ouvi.} e José Roberto Pereira Júnior\footnote{Grande filósofo.},
e os contemporâneos de mestrado Mateus Araújo\footnote{Mateus! É sempre muito massa discutir com você!} e Marco Túlio 
Quintino\footnote{Mais importante que todas as nossas conversar de física, obrigada por sempre tomar conta de mim :)},
muito obrigada por fazerem parte dessa etapa. E por compartilharem tantas discussões, almoços, dúvidas, 
festas, reuniões, cafés, Paratys \ldots

Aos membros do EnLight, professores, pós-docs e alunos, muito obrigada por todo o conhecimento compartilhado durante esses anos.
Em especial, agra-deço ao Carlos Parra por me lembrar ocasionalmente, durante a escrita desta Tese, de manter minha sanidade mental.
E ao Pierre-Louis de Assis,
por ocasionalmente\footnote{Na verdade, frequentemente.} me fazer perdê-la com suas interrupções inconvenientes (das  quais já sinto saudades).
Ao Dudu (Eduardo Mascarenhas), obrigada por cuidar de mim nos momentos difíceis e por ser um exemplo como pesquisador.

I thank Marcus Huber, Fabien Clivaz, and Atul Mantri for the very nice collaborations initiated during my PhD, 
from which I certainly profit a lot.

I am grateful to all the quantum friends I have made during this time. It is always pleasant to meet you somewhere in the world.
In special I thank Alexia Salavrakos, Joe Bowles, and Flavien Hirsch for so many special moments.

À minha irmã Bizy, agradeço por sempre me apoiar. Por ser exemplo e ins-piração. E por me alimentar durante a escrita desta Tese.

This Thesis was significantly improved due to the careful reading of my supervisor Marcelo Terra Cunha and
Jessica Bavaresco. I also thank Mateus Araújo, Marco Túlio Quintino, and Hakob Avetisyan for comments and feedback in earlier versions.
I thank the referees of this Thesis: Daniel Cavalcanti, Reinaldo Oliveira Vianna, Fernando de Melo, Raphael Campos Drumond, and Andreas Winter
for very nice discussions and feedbacks.
I owe a special thanks to Jessica Bavaresco and Thiago Maciel\footnote{Tchê, obrigada também por sempre estar disponível para tirar minhas dúvidas numéricas 
(e foram inúmeras) ao longo desses anos.} for the technical support, making 
it possible to have an international committee
in my PhD defense.

Finally, I acknowledge CNPq for my first year scholarship and I am greatful to Fundação de Amparo à Pesquisa do Estado de Minas Gerais (FAPEMIG) for 
the remaining three years of scholarship, for the sandwich PhD program which opened so many doors for me, and for financial support to many conferences.
I also acknowledge support from NCN grant 2013/08/M/ST2/00626,  Polish  MNiSW  Ideas-Plus
Grant IdP2011000361, ERC Advanced Grant QOLAPS
and  National  Science  Centre  project  Maestro  DEC-
2011/02/A/ST2/00305.

%% file: abstract.tex
\chapter*{Resumo}

Não-localidade é um dos aspectos mais intrigantes da teoria quântica, que revela que a natureza é 
intrinsecamente diferente da nossa visão clássica do mundo.
Um dos principais objetivos no estudo de não-localidade é determinar a máxima violação obtida por correlações quânticas em 
um cenário de Bell. Entretanto, dada uma desigualdade de Bell, nenhum algoritmo geral é conhecido para calcular esse máximo.
Como um passo intermediário, o desenvolvimento de cotas eficientemente computáveis para o valor quântico de desigualdades de Bell
tem tido um papel importante para o desenvolvimento da área.
Nessa tese, apresentamos nossas contribuições explorando cotas eficientemente computáveis, baseada na norma de certas matrizes, para o
valor quântico de uma classe particular de desigualdades de Bell: os jogos lineares.
Na primeira parte introduzimos os pré-requisitos necessários para os resultados principais: Conceitos e resultados das teorias de 
otimização e complexidade de computação, com foco em problemas de não-localidade; O formalismo de jogos não-locais como um caso particular de desigualdades
de Bell; E a abordagem de grafos para não-localidade.
Na segunda parte apresentamos nossos resultados principais sobre a caracterização de condições necessárias e suficientes para um jogo {\xor} não
ter vantagem quântica, e provamos uma cota eficientemente computável para o valor quântico de jogos lineares. Os principais 
resultados apresentados aqui são: (i) Determinação da capacidade de Shannon para uma nova família de grafos; (ii) Generalização, para funções com $d$ possíveis valores,
do princípio de não-vantagem
em computação não-local; (iii) Um método sistemático de gerar testemunha de 
emaranhamento genuíno independente de dispositivo  para sistemas tripartidos.

\newpage
\chapter*{Abstract}

Nonlocality is one of the most intriguing aspects of quantum theory which reveals that nature is intrinsically different than our classical view of 
the world. One of the main goals in the study of quantum nonlocality is to determine the maximum violation achieved by quantum correlations in a Bell 
scenario. However, given a Bell inequality, there is no general algorithm to perform this task. As an intermediate step, the development of efficiently computable bounds 
has played an important role for the advance of the field.
In this thesis we present our contributions  exploring efficiently computable bounds, based on a norm of some matrices, to the quantum
value of a particular class o Bell inequalities: the linear games.
In the first part of the thesis we introduce the necessary background to follow the main results: Concepts and results of optimization and computational
complexity theories, focusing on nonlocality problems; The framework of nonlocal games as a particular class of Bell inequalities; And the graph-theoretic
approach to nonlocality. In the second part we present our main results concerning the characterization of necessary and sufficient conditions for an XOR
game to have no quantum advantage, 
and we prove an efficiently computable upper bound to the quantum value of linear games.
The main outcomes of the research presented in this thesis are: (i) The determination of the Shannon capacity for a new family of graphs; (ii)
A larger alphabet generalization of the principle of no-advantage for nonlocal computation;
(iii) And a systematic way to design device-independent witnesses of genuine multipartite entanglement for tripartite systems.

%% file: papers.tex
\chapter*{List of papers}
\addcontentsline{toc}{chapter}{List of papers}

The content of this Thesis is based on results developed in the following papers:

 \begin{enumerate}
 \item \textit{Characterizing the Performance of {\xor} Games and the Shannon Capacity of Graphs}\\
   R. Ramanathan, A. Kay, \textbf{G. Murta} and P. Horodecki\\
 \href{http://link.aps.org/doi/10.1103/PhysRevLett.113.240401}{\textbf{Phys. Rev. Lett., \textbf{113}, 240401, (2014)}}.
  \item \textit{Generalized {\xor} games with $d$ outcomes and the task of nonlocal computation}\\
  R. Ramanathan, R. Augusiak, and \textbf{G. Murta}\\
 \href{http://link.aps.org/doi/10.1103/PhysRevA.93.022333}{\textbf{Phys. Rev. A, 92, 022333 (2016)}}.
  \item \textit{Quantum bounds on multiplayer linear games and device-independent witness of genuine tripartite entanglement }\\
  \textbf{G. Murta}, R. Ramanathan, N. M\'oller, and M. Terra Cunha\\
  \href{http://link.aps.org/doi/10.1103/PhysRevA.93.022305}{\textbf{Phys. Rev. A, \textbf{93}, 022305, (2016)}}.\\
\end{enumerate}

The author also contributed to the work:
\begin{itemize}
 \item \textit{Bounds on quantum nonlocality via partial transposition}\\
 K. Horodecki and \textbf{G. Murta}\\
  \href{http://link.aps.org/doi/10.1103/PhysRevA.92.010301}{\textbf{Phys. Rev. A, \textbf{92}, 010301(R), (2015)}}.
\end{itemize}
A summary of the results developed in this work is presented in Appendix \ref{chapterstatebound}.

%% file: intro.tex
\chapter*{Prologue}\label{intro} 
\addcontentsline{toc}{chapter}{Prologue}
\markboth{Prologue}{}

One of the most intriguing aspects of quantum theory is the fact that it is intrinsically probabilistic. This probabilistic character
led Einstein, Podolsky, and Rosen, in the remarkable EPR paper of 1935 \cite{EPR}, to question whether quantum theory was an incomplete theory, and 
therefore this probabilistic character would emerge from  our lack of knowledge of some variables.
These questionings were answered  in a negative way by Bell in 1964 \cite{Bell}.
With a mathematical formulation of the EPR paradox, Bell showed that if we were able to complete quantum mechanics in the way proposed by EPR then we
should not observe some phenomenon (the violation of a Bell inequality) that we actually do! 
The work of Bell does not imply that quantum theory is the ultimate theory, however
no such refinement as the one pursued by EPR can exist.

Even worse than this probabilistic character, what is really intriguing about quantum theory
is the fact that, up 
to the moment, there is no set of physical principles that fully characterizes it.
If we consider special relativity, this theory
has some surprising predictions that goes against our daily life experiences. 
However as weird as they seem,
all these predictions 
can be derived from two physical principles: (i)  The laws 
of physics are the same in all inertial
reference frames; (ii) The speed of light in vacuum 
is $c$ in all inertial 
 reference frames. Once we accept these principles
(and I do not claim this is an easy task!) there is no mystery, and 
we are able to 
explain all the phenomena that arise from the theory.

Quantum theory is  very well established by a 
bunch of mathematical axioms that tells us how to predict the statistics of the results of experiments. 
However we still do not have many clues on which are the 
physical principles behind this purely mathematical formulation.
In his famous quotation, Feynman (in the prestigious `Messenger Lectures' at Cornell University \cite{Feynman}) said
\begin{quote}
\textit{``There was a time when the newspaper said that only twelve men understood the theory of relativity. I do not believe there ever was
such a time. There might have been a time when only one man did, because he was the only guy who caught on, before he wrote his paper. But after 
people read the paper a lot of people understood the theory of relativity in some way or other, certainly more than twelve. On the other hand, I think 
I can safely say that nobody understands quantum mechanics.''} 
\end{quote}

This lack of principles receives a clear formulation in the study of nonlocality. 
When defining the sets of local and no-signaling correlations, we have clear mathematical constraints that delimit them, and, 
additionally, these constraints  have a physical
(information theoretic) interpretation. For example, the no-signaling principle states, in an information theoretic language, 
that if Alice and Bob
do not communicate no information can be obtained about the other party by analyzing only the local statistics. 
The additional constraints imposed to the set of local correlations 
also have a physical interpretation. 
However, when it concerns the set of quantum correlations all that we know is that the probability distributions
can be described by positive operator valued measures applied to a trace-one positive operator that acts on a Hilbert 
space $\Hi$. Which, definitely,  does not sound very  physical! 
And this is why Feynman says that \textit{``nobody understands quantum mechanics''}.


Almost a century has passed since the questionings of EPR and we still do not have a satisfactory description of quantum theory in terms
of physical axioms.
However, in the mean time
we have developed technologies based on quantum effects and explored in many different ways the novelties brought by quantum theory. 
In the quantum 
nonlocality domain people found a way to explore Bell inequality violations in order to develop secure
cryptographic protocols that do not 
rely in any assumption about the specific description of the system, but rather only on the statistics of the results of experiments,
the called device-independent
paradigm.
And besides the manipulation of quantum effects, we have also achieved some understanding on the consequences and limitations 
due to the mathematical formulation of the theory.

So at this point I should apologize and  warn the reader that unfortunately the following pages will not make you 
\textit{understand quantum mechanics}. However, if you keep going you might have a glance on the subject of quantum nonlocality, which
highlights one of the weirdest aspects of quantum theory in a very clear and simple scenario: where 
Alice and Bob, space-like separated, perform 
local measurements on their systems, and the only thing that matters is the statistics of their inputs and outputs.
This simple scenario opens space for a rich discussion 
of the fundamental aspects of quantum theory. 
The analysis of the performance of Alice and Bob in some particular tasks when they have access to quantum resources or not gives
us a framework to explore the extent and limitations of the theory.
This thesis is devoted to the study of the task of evaluating the quantum value of a Bell expression.
We will discuss the difficulty of this problem putting it into the 
language of computational complexity and optimization theories. And we will present our contributions concerning bounds on the quantum
value of a particular class of Bell expressions: the linear games.
At the end of the day, I hope the reader \textit{enjoy it}!\vspace{-0.5em}

\subsection*{Outline\footnote{This Thesis was revised in March/2017 and Journal references were updated.}}

In Part I we introduce the necessary background to 
follow the results presented here. 
 In Chapter \ref{chapterNL} we present a brief introduction to nonlocality stating some concepts and
general results.
 Chapter \ref{chaptercomput} introduces optimization and computational complexity theories. 
 Chapter \ref{chaptergames} presents the 
framework of nonlocal games, which can be seen as a particular class of Bell expressions, focusing on linear games which are the main subject of
study of this thesis. 
 In Chapter \ref{chaptergraphNL} we introduce the graph-theoretic approach to nonlocality, showing how some 
graph invariants are related to the classical, quantum and no-signaling values of Bell expressions. 

Part II is devoted to the results developed by the author, together with collaborators, during the last four years.\vspace{-0.5em}
\begin{itemize}
 \item In Chapter \ref{chapterxor} we focus on {\xor} games. We present a necessary and sufficient condition for an
{\xor} game to have no quantum advantage and, exploring this result, we are able to determine the Shannon capacity of a broad new 
family of graphs.\vspace{-0.5em}
\item  In Chapter \ref{chaptergamesd} we present an efficiently computable upper bound to the quantum value of linear games. We
explore it re-deriving a recently discovered bound to the 
CHSH-$d$ game. We also show that these bounds can exclude the 
existence of some no-signaling boxes that would lead to the trivialization of communication complexity. As the main outcome of the 
introduced bound, we derive a larger alphabet generalization of the principle of no-advantage for nonlocal computation.\vspace{-0.5em}
\item In Chapter \ref{chapternplayer} we extend the previous bound to $n$-player linear games. We also derive an upper bound to the quantum 
value of a multipartite version 
of the CHSH-$d$ game and we extend the result concerning no-quantum realization of no-signaling boxes that would lead to the 
trivialization of  communication complexity in a multipartite scenario.
Finally, we present a systematic way to derive device-independent witnesses of genuine multipartite entanglement for tripartite systems.
\end{itemize}

\vfill
\begin{quote}
\textit{``So do not take the lecture too seriously, feeling that you really have to understand in terms of some model what 
I am going to describe, 
but just relax and enjoy it.''} (Feynman \cite{Feynman})
\end{quote}

%% file: NL.tex
\part{Preliminaries}

\chapter{Nonlocality}\label{chapterNL}

Let us analyze the following story:

\begin{center}
\begin{minipage}{12cm}
\textit{Alice and Bob went abroad for their PhD studies and now they are flatmates.
After some months living together Bob noticed a strange behavior of Alice: every time
Bob wakes up looking forward to tell Alice the news from his hometown, she coincidently 
wakes up particularly grumpy, even though in general she is a very easy going and talkative person. 
When Bob realizes that this grumpy behavior of Alice is recurrent, but only happens in the specific days he has some news to tell,
he tries to find out what could be the cause of this strange correlation. \\
He is sure that this cannot be caused by himself, since they meet every evening when they get back home, and everything is fine
before they go to their respective rooms until the next day. Moreover, this situation happens in random days but coincidently
every time Bob wants to tell news during the breakfast.\\
After a long analysis he finally finds out the explanation for this correlation: the phone call to his family the evening before.
Every time he made a phone call, the Internet of the house stopped working.
And this was happening because their wireless router was settled to the same  frequency as the one used by their wireless phone.
Alice, on the other hand, checks her computer simulations at home every evening (accessing her working computer remotely). Because the internet 
fails to work for the hours Bob spend in the phone, she is only able to finish her work very late at night, which causes a big grump!\\
By adjusting the router's frequency, the problem was solved and Alice and Bob lived happily ever after...}
\end{minipage}
\end{center}
\vspace{1em}

At first the correlation between Alice and Bob may sound very strange, 
however when we become aware of the previously ``hidden'' fact that the frequency of their wireless router was interfering with the frequency of their wireless phone,
causing all the trouble,  
everything looks pretty natural.

That is the idea of a \emph{local hidden variable model}: To find an explanation for correlations in terms of some common cause (the term local will become
clear soon). However, as we will see, there exist correlations in nature which cannot be explained by a local hidden variable model, 
these correlations are then called \emph{nonlocal correlations}.
Nonlocal correlations are one of the most intriguing aspects of nature. And besides
their foundational interest, these correlations have also shown to be very useful in cryptographic and information processing tasks as, for example, device-independent 
randomness amplification and expansion \cite{ColbeckRennerAmplify, randexpansion}, device-independent quantum key distribution \cite{Ekert1991,PAB09,MPA11,VV14}, 
and reduction of communication complexity
\cite{vanDam, noisyPRxCC, NLCommCompl}. 

In the study of nonlocality we consider the following scenario: Alice and Bob are far away from each other\footnote{In
technical words, we want Alice and Bob to be space-like separated.} and they are going to observe things that 
happen around them, \ie they are going to ask questions to their systems (or in a more scientific language, they are going to perform experiments/measurements in their respective laboratories). The set 
of possible questions that Alice can ask to her system is denoted $\inp_A$ and the set of possible questions that Bob can ask
is denoted $\inp_B$. 
The sets of possible answers (outcomes) to these questions\footnote{For simplicity, here we focus on 
the case where each experiment that Alice and Bob perform has the same set of outputs, but this need not to be the case. 
Nevertheless, most of the 
results are straightforward generalized to the asymmetric case.}  are denoted respectively $\out_A$ and $\out_B$. 
An example of a question that can be asked is \textit{`Is it raining now?'},
which has two possible outcomes: \textit{`yes'} or \textit{`no'}. They can 
also throw a dice and observe the upper face, which has six possible outcomes: 1, 2, 3, 4, 5 and 6. 

We also consider that Alice and Bob can ask their systems only one question at a time (\ie Alice is not allowed to check if it is raining and throw a dice at the same 
time\footnote{Of course in a classical world there is no restriction in performing this task. One can perfectly go outdoors and trow a dice 
obtaining at the same time a number and the answer about the weather. 
However we cannot assert this for a quantum system, and there exist pairs of questions such an experiment to determine the output of one 
disturbs the output of the other.}). 
The motivation for this restriction is that, when we consider quantum theory, we might deal with incompatible observables, as
for example a measurement of spin in the $\hat{x}$-direction and a measurement of spin in the $\hat{z}$-direction, 
 hence we have questions that cannot be asked together.

Our goal is to analyze the joint probability distribution that Alice 
performs the experiment  $x\in \inp_A$ and obtains the outcome $a \in \out_A$ and Bob observes $y\in \inp_B$ and obtains the outcome $b \in \out_B$:
\begin{align}
 P(a,b|x,y).
\end{align}

Since we are only concerned with the statistics of the outputs  given the inputs, 
and nothing else matters for us, we can model any such experiment 
as a black
box (see Figure \ref{figNL}): which has some buttons as inputs (the possible questions) and a set of light bulbs as outputs
(the possible answers to the question).  This is called a \emph{device-independent} scenario, where we do not make any assumption
over the internal mechanisms of the devices used for the experiment.

\begin{figure}[h]
\begin{center}
 \includegraphics[scale=0.8]{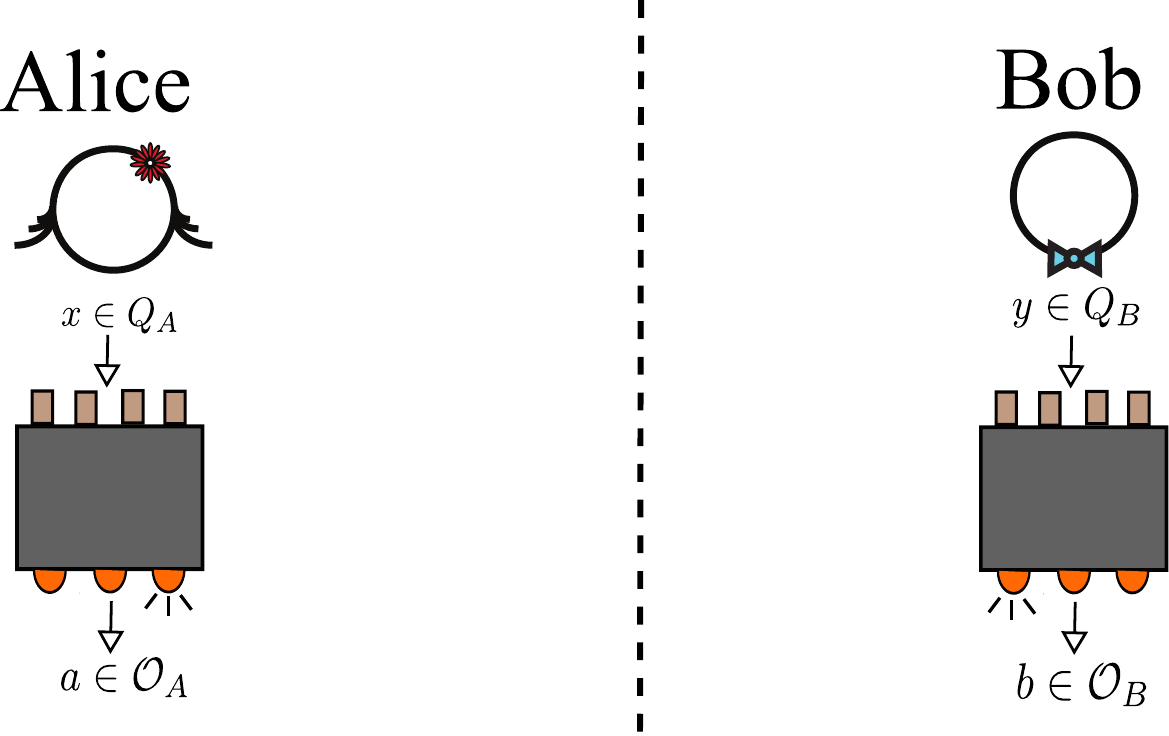}
 \caption{\textbf{Nonlocality scenario:} Alice and Bob are far apart and they are going to perform experiments on their respective laboratories. 
 Their experiments can  be described as black boxes:
 the upper buttons are the possible inputs and the lower light bulbs are the possible outputs.}\label{figNL}
 \end{center}
 \end{figure}
 
A \emph{box} $\vec{P}(a,b|x,y)$ is a vector specifying all the joint probability distributions of a particular scenario\footnote{For the particular
case where Alice and Bob have two possible inputs with two outputs, $a,b,x,y\in \DE{0,1}$, $\vec{P}(a,b|x,y)$ is the sixteen-component vector:
\begin{align*}
 \vec{P}(a,b|x,y)=(P(00|00),P(01|00),P(10|00),P(11|00),P(00|01),\ldots,P(11|11)).
\end{align*}
}. In the 
following sections we are going to analyze which properties the set of boxes  $\vec{P}(a,b|x,y)$ satisfies.\vspace{1em}

In this Chapter we are going to give a brief overview of the concepts and main results in the study of nonlocality. For an introduction to 
nonlocality, with detailed proofs of many results, the reader is referred to Ref. \cite{dissertaGlaucia} (only in Portuguese) or Ref. \cite{dissertaMT}. 
A nice review, from 2014, contains the references for many important results on the field of nonlocality \cite{reviewNL}.

\section{Local correlations}\label{secLocal}

In the study of nonlocality we are considering a scenario where Alice and Bob are far away from each other during the course of their experiments.
We also assume that their choices of which experiment they are going to perform are made when they are already far apart. 
Our classical 
intuition leads us to expect that, whichever correlations
they observe, they have to be explained by a common cause that does not depend on which experiment they chose to perform (since this choice 
was made when they were far apart).
The boxes that capture this classical intuition are called \emph{local} boxes.

\begin{definition}[Local correlations] \label{deflocalrealism}
 Local correlations are the ones that can be explained by a local hidden variable model, \ie a box $\vec{P}(a,b|x,y)$ is local
 if there exists a variable $\lambda \in \Lambda$, independent of the choice of inputs of Alice and Bob, such that
 \begin{align}\label{eqlocal}
  P(a,b|x,y)=\int_{\Lambda}q(\lambda) p(a|x,\lambda) p(b|y,\lambda) d\lambda
 \end{align}
 where $q(\lambda)$ is a probability distribution.
\end{definition}

The definition of a local box states that all the correlations observed by Alice and Bob in their experiments are due to the lack
of knowledge of some hidden variable $\lambda \in \Lambda$.
Note that in Definition \ref{deflocalrealism} we do not make any assumption over the nature of the variable $\lambda$, it can be a continuous variable, 
it can be a set of variables and so on \ldots The only assumption is that $\lambda$ is not correlated with the choices
of inputs of Alice and Bob\footnote{This assumption is also referred as \emph{free will}, as, if Alice and Bob can freely make their choices, this 
assumption will be satisfied. Since here we do not want to make any metaphysical discussion, let us 
just assume that this independence (between the hidden variable and the choice of inputs) holds, no matter what justifies it.} (\emph{measurement independence}). 
Another important assumption in Definition \ref{deflocalrealism} is that,
conditioning on all variables that could have a causal relation with a particular event, the probability of this event is independent of any other variable. This is often refereed to as \emph{local causality}. Therefore
 for each value of $\lambda$ the local probability distribution on Alice's outcome 
is independent of Bob's experiment, \ie $p(a|b,x,y,\lambda)=p(a|x,\lambda)$, and the same holds for Bob's local distribution.
This captures
the interpretation that the hidden variable $\lambda$ would be a common cause in the past that is responsible for generating the correlations.





There are other equivalent ways to formulate Definition \ref{deflocalrealism}, see, for example, Ref. \cite{Fine}.
Moreover, Eq. \eqref{eqlocal} can be derived from a slightly different set of assumptions. But it is important to have in mind that 
there is always a set of assumptions (measurement independence and local causality in the previous discussion) present in the definition of local correlations, and the violation of a Bell inequality do not tell us which particular assumption is being invalidated.
For a 
very nice discussion on the assumptions implicit in the definition of \textit{locality} we refer the reader to Ref. \cite{Mateusblog}\footnote{This nonlocal
reference was
added in the revised version of the thesis.}. 
Therefore whenever  we use the term \emph{nonlocal} in this thesis we refer to the impossibility of writing a joint probability distribution in the form of
Eq. \eqref{eqlocal}.

Definition \ref{deflocalrealism} reflects the intuition we learn from our daily life experience and it also 
expresses the predictions of classical theories (as classical mechanics and special 
relativity), which were the theories that prevailed before the advent of quantum theory.

\subsection*{The local polytope}

The set of all boxes $\boxp$ that can be written in the form \eqref{eqlocal} is called the \emph{local set} of correlations, $\mathcal{L}$. For any
scenario that we consider, \ie for any finite set of inputs $\inp_A$ and $\inp_B$ and any finite set of outputs $\out_A$ and $\out_B$, the local set is a \emph{polytope}.
A polytope is a convex set with a finite number of extremal points. The extremal points of the local polytope are the deterministic local boxes:
\begin{align}\label{eqdeterministic}
 P(a,b|x,y)=D(a|x)D(b|y),
\end{align}
where $\vec{D}(a|x), D(a|x)  \in \DE{0,1}$, is a deterministic probability distribution, and analogously
for $\vec{D}(b|y)$. 
One can actually derive that 
every box $\boxp$ that satisfies Eq. \eqref{eqlocal} can be written as a convex combination of deterministic probability distributions, 
Eq. \eqref{eqdeterministic}\footnote{For this reason Definition \ref{deflocalrealism} is also called \emph{local realism} or \emph{local determinism}.}.

\begin{proposition}[Local polytope] 
The local polytope is the convex hull of the deterministic local boxes:
\begin{align}
 \mathcal{L}:=\DE{\boxp \mid P(a,b|x,y)=\sum_{i}c_i D_i(a|x)D_i(b|y)}
\end{align}
where $c_i\geq 0$, $\sum_i c_i=1$, and $i$ runs over all possible deterministic boxes of the scenario.
 \end{proposition}
 
For a scenario with $|\inp_A|=|\inp_B|=m$ and  $|\out_A|=|\out_B|=d$, the local set is a convex polytope of dimension\footnote{The dimension of the 
polytope is determined by taking into account the normalization of the probability distributions, and the no-signaling condition that we 
are going to specify soon.} $m^2(d-1)^2+2m(d-1)$
with $d^{2m}$ vertices \cite{PironioPhD}.

A convex polytope is fully characterized by its vertices, but an equivalent characterization is given  by its facets\footnote{This is the 
\emph{Main Theorem for polytopes}, see  Theorem 1.1 in Ref. \cite{Lecpolytopes}.}.
The facets are hyperplanes that delimit the set. 
The nontrivial facets\footnote{The trivial ones are the positivity condition of the probabilities: $P(a,b|x,y)\geq 0$ $\forall\; a,b,x,y$.} of
the local polytope are called \emph{tight Bell inequalities} \cite{Bell}. 
A \emph{Bell inequality} is a condition that is necessarily satisfied by all local
correlations. It can be a tight condition and correspond to a facet of the local polytope, or else it may correspond to faces of the local polytope
with lower dimension, or may even not touch the polytope.

In the scenario where Alice and Bob each can chose one between two possible inputs $\inp_A=\inp_B=\DE{0,1}$ and each
input has two possible outputs $\out_A=\out_B=\DE{0,1}$, the local polytope has as its unique nontrivial facet (up to relabeling of inputs and outputs)
the notorious CHSH inequality \cite{CHSH}.

Consider the following expression:
\begin{align}\label{eqCHSH}
 \mathcal{S}^{CHSH}=\mean{A_0,B_0}+\mean{A_1,B_0}+\mean{A_0,B_1}-\mean{A_1,B_1},
\end{align}
where $\mean{A_x,B_y}=P(a=b|x,y)-P(a\neq b|x,y)$. 
A substitution of Eq. \eqref{eqlocal} into the RHS of Eq. \eqref{eqCHSH} shows that
\begin{align}\label{eqCHSHL}
 \mathcal{S}^{CHSH} \leq 2 
\end{align}
for any local box.

The CHSH inequality \eqref{eqCHSHL} is the simplest and most well explored of all the Bell inequalities. It was introduced in 1969 by Clauser, Horne, Shimony and Holt \cite{CHSH}.
After the work of Bell \cite{Bell}, which finally opened the possibility to formalize in a
mathematical way the concepts of local realism first discussed by Einstein, Podolsky and Rosen \cite{EPR},
the CHSH inequality was proposed
as a condition that could be  experimentally tested. 
Inequality  \eqref{eqCHSHL} was used in the first experiment that closed the locality loophole \cite{expAspect}, performed
 by Aspect's group,
and also in the recent groundbreaking loophole-free Bell experiment by Hensen \etal \cite{bellHensen}. A variant of the CHSH inequality (the CH-Eberhard
inequality\footnote{The CH-Eberhard inequality is a reformulation of the CHSH inequality which is more suitable for taking into account detection 
efficiencies.} \cite{CH, Eberhard}) was used in the subsequent
experiments by Giustina \etal \cite{bellGiustina} and Shalm \etal \cite{bellShalm}.
These last three experiments have finally ruled out local realism in nature\footnote{Up to some 
stronger loopholes, as the super-determinism, that by definition cannot scientifically be ruled out.}.


\section{No-signalling correlations}\label{secNS}
We may be less picky and not seek a local hidden variable model
to explain our correlations, but we want to keep some minimal assumptions
about the possible boxes: 
If Alice and Bob are far away from each other and do not communicate during their experiments, 
it is reasonable to expect that Bob can get no information about what happens in Alice's laboratory and vice-versa. This is the no-signaling principle
and the most general boxes we are going to deal with are the ones that at least satisfy this constraint.
More formally, the no-signaling 
principle states that the marginals of the local experiments do not depend on the other part's experiment.

\begin{definition}[No-signaling condition]
 A box $\boxp$ is no-signaling iff
 \begin{subequations}\label{eqNS}
 \begin{align}
  \sum_b P(a,b|x,y) &=\sum_b P(a,b|x,y') \; \;\;\forall \; y,y' \in \inp_B\;,\; \forall x \in \inp_A\;, \;\forall a \in \out_A,\\
   \sum_a P(a,b|x,y) &=\sum_a P(a,b|x',y)\; \;\;\forall \; x,x' \in \inp_A\;,\; \forall y \in \inp_B\;,\; \forall b \in \out_B.
 \end{align}
 \end{subequations}
\end{definition}

Note that the no-signaling condition implies that local marginal probabilities are well defined:
 \begin{subequations}
 \begin{align}
  P(a|x):=\sum_b P(a,b|x,y)\;\; \forall y, \\
  P(b|y):= \sum_a P(a,b|x,y)\;\; \forall x.
 \end{align}
  \end{subequations}
  
The boxes that satisfy the no-signaling condition \eqref{eqNS} form the set of no-signaling correlations $\mathcal{NS}$.  $\mathcal{NS}$
is also a convex polytope (as only linear constraints were made to the probability distributions), which contains the classical polytope. 
This can be easily seen by checking
that local boxes \eqref{eqlocal} satisfy the no-signaling
condition \eqref{eqNS}, hence
\begin{align}
 \mathcal{L}\subseteq \mathcal{NS}.
\end{align}
Later we are going to see (Section \ref{secCHSH}) that, in general, this inclusion can be a strict relation.

The no-signaling polytope is much easier to characterize than the local polytope since  $\mathcal{NS}$ is fully characterized
by Eqs. \eqref{eqNS} (and by the trivial conditions of positivity and normalization of the probability distributions) which are linear
constraints that can be easily checked. 
Although $\mathcal{L}$ is also delimited by linear constraints, the tight Bell inequalities 
are not easy to derive and we are left with a description in terms of the deterministic points, which is an integer quadratic problem (see Section 
\ref{secopt}).  

\section{Quantum correlations}\label{sec.QuantumCorr}
Quantum correlations are boxes $\boxp$ that can be described as quantum local measurements being performed in a shared quantum state
(see Appendix \ref{Aquantum} for an overview of concepts and definitions in quantum theory).

\begin{definition}[Quantum correlation]\label{defQcorr}
 A box $\boxp$ is quantum if there exist a quantum state $\rho \in D(\Hi_A \otimes \Hi_B)$ and 
 local POVMs $\DE{M_x^a}_a$ and $\{M_y^b\}_b$ acting on $\Hi_A$ and $\Hi_B$ respectively, such that
\begin{align}\label{eqQcorr}
 P(a,b|x,y)=\Tr \de{M_x^a \otimes M_y^b \;\rho},
\end{align}
for arbitrary Hilbert spaces $\Hi_A$ and $\Hi_B$.
\end{definition}

The set of all boxes $\boxp$ that admit a description as Eq. \eqref{eqQcorr} is the set of quantum correlations $\mathcal{Q}$.
Note that in Definition \ref{defQcorr} we do not put any restriction on the dimension of the system.

The set of quantum correlations contains the local polytope $\mathcal{L}$. This is expected by the fact that quantum theory is a generalization 
of classical theory, hence $ \mathcal{L}\subseteq \mathcal{Q}$.
Some facts concerning the relation of $\mathcal{Q}$ and $\mathcal{L}$ are: 
\begin{itemize}
\item Local measurements in \emph{separable} quantum states only generate correlations in $\mathcal{L}$.
\item If the local measurements of one of the parties are \emph{jointly measurable}\footnote{Two sets of POVMs
$\DE{E^i}_{i=1}^m$ and $\DE{F^j}_{j=1}^n$ are 
jointly measurable if there exists a third POVM $\DE{G^{i,j}}_{i,j=1}^{m,n}$ such that 
\begin{align*}
 \sum_j \Tr  G^{i,j}\rho=\Tr  E^{i}\rho\;\;\;\; \text{and}\;\;\;\;  \sum_i \Tr G^{i,j}\rho=\Tr F^{j}\rho
\end{align*}
for every quantum state $\rho$. This means that the statistics of the original measurements can be obtained by the marginals of the 
statistics for the POVM $\DE{G^{i,j}}_{i,j=1}^{m,n}$.} then the correlations generated are in $\mathcal{L}$.
\end{itemize}

Therefore, in order to observe correlations beyond the classical polytope one necessarily needs entanglement and not joint measurability.
Whether these conditions are sufficient to generate nonlocal correlations is a fruitful field of research.
 In the standard Bell scenario, it is known
that some entangled quantum states can only generate classical correlations \cite{Werner, Barrett} (a systematic method to 
check whether entangled states admit a local model  
was recently derived in Refs. \cite{LeoLocal, FlavienLocal}).
Partial results concerning joint measurability can be found in Refs. \cite{WPF09,MTJMlocal}.
 More general scenarios were introduced in the study of nonlocality: The \emph{hidden nonlocality scenario} \cite{hiddenPopescu,hiddenZHHH}
 where Alice and Bob are allowed to  pre-process one copy of their state by a local filtering operation 
 before starting the Bell test;
 The \emph{many-copy scenario} where many copies
 of a state are shared between Alice and Bob \cite{PalazuelosActivation, DaniActivation}; And the \emph{network scenario} \cite{DaniNet1,DaniNet2} where copies of a bipartite
 quantum state $\rho$ are distributed in a network of arbitrary shape and number of parties.
   These general scenarios were shown to be more powerful than
 the standard one \cite{genuinehidden,DaniNet1,DaniNet2} and even 
 the phenomena of super-activation of nonlocality was exhibited \cite{PalazuelosActivation, DaniActivation}.
 However, whether nonlocality, entanglement and not joint measurability are equivalent
in these general scenarios remains an open problem.

In Ref. \cite{GMstatebound} we show that the value achieved
by a quantum  state in a Bell scenario is bounded by a term related to its distinguishability
from the set of separable states by means of a restricted
class of operations.
We also propose quantifiers for the nonlocality of a quantum state in the asymptotic scenarios where many copies
and filter operations are allowed, and we show
that these quantities can be bounded by the relative entropy of entanglement of the state (or the partially transposed state, in the case of PPT states).
A summary of the results of Ref. \cite{GMstatebound} is presented in Appendix \ref{chapterstatebound}.

\vspace{1em}

Concerning the relation between $\mathcal{Q}$ and $\mathcal{NS}$, we can straightforwardly verify that quantum correlations 
satisfy the no-signaling condition \eqref{eqNS}:
\begin{align}
 \sum_b  P(a,b|x,y)&= \Tr \de{M_x^a \otimes \de{\sum_b M_y^b} \rho}\nonumber\\
 &= \Tr \de{M_x^a \otimes \I\; \rho}\\
 &= \Tr \de{M_x^a \;\rho_A}\nonumber\\
 &\eqdef P(a|x),\nonumber
\end{align}
where $\rho_A$ is the reduced state of Alice (as defined in \eqref{eqredstate}), and analogously for Bob's marginal.

In summary, we have
\begin{align}
 \mathcal{L}\subseteq \mathcal{Q} \subseteq \mathcal{NS},
\end{align}
and we are going to see in the next Section that all these inclusions can be strict in a general Bell scenario.

Even though the quantum set lies in between two polytopes, in general $\mathcal{Q}$ is not a polytope. The characterization of the quantum set 
of correlations is the main open problem in the field of nonlocality, and it is not even known for the simplest scenario of two inputs
and two outputs\footnote{A partial result characterizes the border of the quantum set in the simplest two-input two-output scenario in the correlation 
representation \cite{chshMasanes}, \ie when we consider only the correlators $\mean{A_xB_y}=P(a=b|x,y)-P(a\neq b|x,y)$ instead of the probabilities $P(a,b|x,y)$.}.
We know that $\mathcal{Q}$ is a convex set\footnote{It is not hard to show that the convex combination of two quantum boxes 
can be 
expressed as a quantum box with measurements and state in a Hilbert space of higher dimension.},
but it is not known if this set is closed\footnote{A set $X$ is closed if every converging sequence of points in $X$ converges to a point of $X$.}. 
An alternative way to define the quantum set of correlations is to impose commutativity of every 
measurement of Alice with every measurement of  Bob, in place of the tensor product structure. The set of correlations generated by these assumptions
is denoted $\mathcal{Q}'$. It is clear that $\mathcal{Q} \subseteq \mathcal{Q'}$, and for finite dimensional Hilbert spaces we have equivalence,
but whether or not these two sets\footnote{Actually Tsirelson's statement is concerned with the equivalence of the closure of the sets.} are equivalent for the infinite dimensional case is known as the
Tsirelson's problem\footnote{See Tsirelson's comments on the problem in: \url{http://www.tau.ac.il/~tsirel/Research/bellopalg/main.html}.} \cite{TsirelsonProblem} (this problem is equivalent 
to a long standing open problem in $\C^*$-algebra, called the Connes' embedding conjecture, see \cite{Conneembedding}). 
An infinite hierarchy of well characterized sets that converges to  the set $\mathcal{Q'}$, 
called \emph{NPA hierarchy},
was introduced by Navascués, Pironio and Acín in Ref. \cite{NPA} (see more in  Section \ref{secNPA}). This constitutes one of the most powerful tools
to deal with problems in the field of quantum nonlocality.

When we are dealing with a particular nonlocality scenario and given a particular Bell expression, as for example $\mathcal{S}^{CHSH}$, 
we might be interested in knowing which value can be achieved if Alice and Bob have access to quantum boxes (as we will see later, they can reach 
$\mathcal{S}^{CHSH}=2\sqrt{2}>2$). Due to the lack of characterization of the quantum set of correlations, this is in general a very hard problem.
More than that: it is not even known whether the quantum value of a Bell inequality is computable in general,
since there is a priori no restriction on the dimension of the Hilbert space for the quantum state and measurements. Only for some
particular instances it is
possible to compute the value exactly or to find efficient approximations. The NPA hierarchy \cite{NPA} is typically
used to get upper bounds on the quantum bound of Bell expressions. However the quality of approximation achieved by these bounds remains unknown
and the number of parameters to be optimized in each level of the hierarchy increases exponentially.
Lower bounds are usually obtained by the called \textit{see-saw} iterative method, where we fix the dimension of the system and recursively 
optimize over a small set of the variables (the quantum state or one of the party's measurements) fixing
the value of the other variables as obtained in the previous step (see Ref. \cite{refseesaw}). 
Each step of the see-saw is an SDP and can be efficiently solved, however this procedure is not guaranteed to converge not 
even to the global maximum of the fixed dimension.
Hence a central problem of great importance in nonlocality theory is to find easily computable and good bounds to handle general classes of Bell inequalities.
In Chapters \ref{chapterxor}, \ref{chaptergamesd} and \ref{chapternplayer} we present our contributions in this direction.

\section{The CHSH scenario}\label{secCHSH}
We now illustrate the concepts introduced in the previous Sections exploring the simplest scenario that can exhibit nonlocal correlations:
the CHSH scenario \cite{CHSH}.
In the CHSH scenario Alice and Bob each has two possible inputs $\inp_A=\inp_B=\DE{0,1}$ and each
input has two possible outputs $\out_A=\out_B=\DE{0,1}$.

The local polytope $\mathcal{L}$ for this scenario can be characterized by the 16 deterministic local boxes or equivalently by its facets. Up to 
relabel of inputs and outputs the only nontrivial facet of the local polytope is the CHSH inequality
$
 \mathcal{S}^{CHSH}\leq 2.
$
Let us write  
\begin{align}
 \mathcal{S}^{CHSH}_c = 2,
\end{align}
 to denote the maximum value attainable by classical (local) theories for the CHSH expression.
We have already introduced the CHSH expression $\mathcal{S}^{CHSH}$ in Eq. \eqref{eqCHSH} and now we evaluate it for quantum and no-signaling boxes.

In quantum theory, in order to calculate the expected values $\mean{A_xB_y}$, we can associate an observable to the measurements of Alice and 
Bob in the following way
\begin{align}\label{eqdefobservable}
 \mathbf{A}_x\defeq \;M_x^0-M_x^1,\\
 \mathbf{B}_y\defeq \;M_y^0-M_y^1, \nonumber
\end{align}
where $\DE{M_x^0,M_x^1}$ are the POVM elements associated to experiment $x$ performed by Alice, and analogously for $\mathbf{B}_y$.
Hence we have that the correlator $\mean{A_xB_y}$ is equivalent to the expected value of the operator $\mathbf{A}_x \otimes \mathbf{B}_y $:
\begin{align}\label{eqobservable}
 \mean{A_xB_y}\equiv \mean{\mathbf{A}_x \otimes \mathbf{B}_y }=\Tr \De{(\mathbf{A}_x \otimes \mathbf{B}_y) \,\rho}.
\end{align}

Now, consider that Alice and Bob share the maximally entangled singlet state
\begin{align}
 \ket{\psi^-}=\frac{1}{\sqrt{2}}(\ket{01}-\ket{10}),
\end{align}
and they perform the measurements associated with the following observables:
\begin{align}
 \mathbf{A}_0=\sigma_Z \;\;,&\;\;\mathbf{A}_1=\sigma_X \nonumber\\
 \mathbf{B}_0=\frac{1}{\sqrt{2}}\sigma_Z+\frac{1}{\sqrt{2}}\sigma_X\;\;,&\;\;\mathbf{B}_1=\frac{1}{\sqrt{2}}\sigma_Z-\frac{1}{\sqrt{2}}\sigma_X.
\end{align}
A direct calculation gives $ \mathcal{S}^{CHSH}=2\sqrt{2}$
for this experiment.

It was shown by Tsirelson \cite{Tsirelsonbound} that this is actually the maximum value we can achieve with 
quantum correlations, hence we have
\begin{align}  
 \mathcal{S}^{CHSH}_q=2\sqrt{2},
\end{align}
where $ \mathcal{S}^{CHSH}_q$ denotes the maximum value attainable by quantum theory for the CHSH expression.

Now let us consider the following box $\boxp$:
\begin{align}\label{eqprbox}
 P(a,b|x,y)=\begin{cases}
             \frac{1}{2} \;\; \text{if}\;\;a\oplus b=x\cdot y,\\
             0\;\; \text{otherwise}.
            \end{cases}
\end{align}
All the marginals are well defined $P(a|x)=P(b|y)=1/2$ $\forall\; a,b,x,y$ and hence this box is no-signaling. However this box 
is not quantum since one can straightforwardly verify
that the value achieved in the CHSH expression is
$
 \mathcal{S}^{CHSH}=4.
$
This is actually the maximum possible value (note that the expected values $\mean{A_xB_y}$ are numbers in the interval $\De{-1,1}$), therefore
\begin{align}
 \mathcal{S}^{CHSH}_{NS}=4.
\end{align}

The box \eqref{eqprbox} was first introduced in Ref. \cite{TsirelsonPRbox} and it became well known after the work of
 Popescu and Rorlich (and hence denoted PR-box) \cite{prbox}, where they discussed whether
the no-signaling
principle was sufficient to limit the nonlocality of quantum theory, showing that actually no-signaling correlations can go far beyond.

\vspace{1em}
So in the simplest nontrivial scenario we have seen that there exist quantum correlations that can violate the locality assumption, hence they
cannot be explained by a local hidden variable model.
Also we can conclude that the no-signaling principle \eqref{eqNS} is not enough to set the limits 
of quantum theory, as it can give rise
to correlations much more general than the ones restricted by the quantum formalism. In Chapters \ref{chaptergamesd} and \ref{chapternplayer} we are 
going to discuss a bit of the implications of these extremal no-signaling boxes in the scenario of communication complexity.

\section{Multipartite scenarios}\label{secNLn}

In the study of nonlocality we can also consider scenarios with many parties involved, $A_1, \ldots, A_N$, all of them performing experiments
far away from each other. In these scenarios, our objects of study are the multipartite boxes
$\vec{P}(a_1,\ldots,a_N|x_1,\ldots,x_N)$, where $a_i\in \out_{A_i}$ represent the output of part $i$ when she/he performs the experiment $x_i\in \inp_{A_i}$

The \emph{locality} condition is straightforwardly generalized for the case of $N$ parties:

\begin{definition} \label{deflocalN}
 A multipartite box $\vec{P}(a_1,\ldots,a_N|x_1,\ldots,x_N)$ is local if there exists a local hidden variable model that reproduces the correlations,
 \ie if there exists a variable $\lambda \in \Lambda$, independent of the choice of inputs of the parties, such that
 \begin{align}\label{eqlocalN}
  P(a_1,\ldots,a_N|x_1,\ldots,x_N)=\int_{\Lambda}q(\lambda) p(a_1|x_1,\lambda) \ldots p(a_N|x_N,\lambda) d\lambda
 \end{align}
 where $q(\lambda)$ is a probability distribution.
\end{definition}

The no-signaling condition is a bit trickier, since now we want to assure no-signaling among all parties.
\begin{definition}
 A multipartite box $\vec{P}(a_1,\ldots,a_N|x_1,\ldots,x_N)$ is no-signaling if the no-signaling condition is satisfied by all bi-partition of
 the parties. More formally, consider a subset of the parties $S \subset \DE{A_1, \ldots, A_N}$, hence the no-signaling condition states that
 \begin{subequations}\label{eqNSN}
 \begin{align}
  \sum_{\vec{a}_{S^c}} P(\vec{a}_S,\vec{a}_{S^c}|\vec{x}_{S},\vec{x}_{S^c}) &=\sum_{\vec{a}_{S^c}} P(\vec{a}_S,\vec{a}_{S^c}|\vec{x}_{S},\vec{x'}_{S^c}):= P(\vec{a}_S|\vec{x}_{S}),
  \end{align}
 \end{subequations}
for all $\vec{x}_{S^c},\vec{x'}_{S^c} \in {\inp}_{S^c}$, $\vec{x}_{S} \in  {\inp}_{S}$, $\vec{a}_S \in {\out}_S$, and all proper subset $S$ of the parties,
where $\vec{x}_S=(x_{i_1},\ldots,x_{i_k})$ and $ {\inp}_S=\inp_{i_1}\times \ldots \times \inp_{i_k}$, for $A_{i_j}\in S$.
\end{definition}

Concerning multipartite nonlocality, we have now different 
levels of correlations.  Analogously to the case of multipartite entanglement (see Appendix \ref{Amultipartiteentang}), 
where we have the concept of \emph{genuine multipartite entanglement} (GME), for multipartite Bell scenarios
we have the concept of \emph{genuine multipartite nonlocality}.

\begin{definition}
 A $N$-partite box $\vec{P}(a_1,\ldots,a_N|x_1,\ldots,x_N)$ is genuine $N$-partite nonlocal if  it \emph{cannot} be written as
 \begin{align}\label{eqbilocal}
 P(a_1,\ldots,a_N|x_1,\ldots,x_N)=\sum_{i} \tilde{q}(i)  \int_{\Lambda}q_i(\lambda)P(\vec{a}_{S_i}|\vec{x}_{S_i},\lambda)P(\vec{a}_{S_i^c}|\vec{x}_{S_i^c},\lambda)\,d\lambda
 \end{align}
 for any hidden variables $\lambda_i$, where $\tilde{q}(i)$ and $q_i(\lambda_i)$ are probability distributions, and 
 $i$ runs over all proper subset of the parties. Moreover, each $P(\vec{a}_S|\vec{x}_{S_i},\lambda_i)$ are no-signaling probability distributions\footnote{In the first time 
 the concept of 
genuine multipartite nonlocality was introduced \cite{Svetlichny} no assumptions was made about the joint probability distributions. 
Nowadays, different definitions are considered, see \cite{GMNL} for a discussion.}.
\end{definition}

Bi-separable quantum states (see Eq. \eqref{eqbiseparable}) can only generate correlations of the form \eqref{eqbilocal}, hence if an 
$N$-partite quantum state $\rho$ exhibits genuine $N$-partite nonlocality, we can conclude that $\rho$ is genuinely $N$-partite
entangled. However the converse is not true, and there exist genuinely $N$-partite entangled states that do not exhibit genuine multipartite 
nonlocality \cite{RemikGMExNL,JoeGMExNL}.
In Ref. \cite{Svetlichny} Svetlichny proposed a method of detecting genuine multipartite nonlocality, designing a tripartite ``Bell-like'' inequality
that was satisfied for all correlations of the form \eqref{eqbilocal} but could be violated for genuinely nonlocal correlations. These results  
were later generalized for multipartite systems in Refs. \cite{GMNLnbody1,GMNLnbody2}. 
Other references on the subject are \cite{GMNL,GMNL2,ACSA10}.

Multipartite nonlocality is still poorly explored and the characterization of these scenarios is less known than the bipartite case.
In Chapter \ref{chapternplayer} we are going to present bounds for the quantum value of a particular class of multipartite Bell inequalities and, as an
application of these bounds, we present a systematic way to design device independent 
witnesses of genuine tripartite entanglement.

%% file: complexity.tex
\chapter{A glance at Optimization and Complexity theories}\label{chaptercomput}
\chaptermark{Optimization and Complexity} 

In our daily life we are constantly dealing with constrained optimization problems. As for example, 
when we go to the cinema with a group of friends. We want to have the best seats (the more central ones in the upper
rows of the cinema room, so that we do not have to tilt our heads to watch the movie), but at the same time we want to seat all together, so 
this is not always an easy problem to solve. And while we choose among the vacant
seats, taking into account the pros and cons of the available options and choosing the one that will give a higher gain (which can be accounted by the 
number of happy people minus the number of unsatisfied people), we are mentally solving this hard optimization 
problem.

In science the situation is not different and many of the interesting problems can be phrased as an optimization problem.
In Chapter \ref{chapterNL} we have discussed the concept of nonlocality and how linear expressions (Bell expressions) can be designed to differentiate 
classical theories (with a local hidden variable model), to quantum theory, and these ones from no-signaling theories.
Therefore in the study of nonlocality an important question we recurrently ask is: 
\begin{center}
\begin{minipage}{12cm}
\textit{Given a Bell expression, what is the maximum value one can achieve if subjected to local/quantum/no-signaling correlations?}
\end{minipage}
\end{center}
The answers to these optimization problems have fundamental importance, since the gaps between the classical and quantum, and the quantum and no-signaling
optimal values show an intrinsic difference between these theories.
Also these gaps have practical applications for the development of quantum algorithms for information processing tasks.
 
In this chapter we introduce some concepts and theoretical results in optimization and computational complexity theories.

\section{Computability and computational complexity}

In this Section we present some basic concepts of computability and computational complexity. Our goal is only
to give an intuitive idea on the subject. 
 For a formal introduction see \cite{CompIntract, CCBarak} (and also \cite{defPNP} 
for a quick overview).

\subsection*{Uncomputability/Undecidability}

Given an  problem that we want to solve, the first step in computability theory is to try to design an algorithm
which is a systematic way to deal with the problem. For any instance\footnote{An instance is a particular input of the problem. In our
example of the cinema problem, the problem itself is specified by the parameters: number of people and available seats. 
A particular instance of the cinema problem is, for example, four people and the two first rows available.} (input) of the problem,
the algorithm follows a number of specified steps
in order to reach the final answer. 
 A problem is said to be computable if there exists an algorithm that, for every input, returns the (approximately) right answer in a
 finite number of steps.

\begin{definition}[Computability]\label{defcompute}
A  problem (P) is computable if there exists an algorithm such that, for every instance $\mathcal{I}$ and for every $\epsilon >0$, there
exists an integer $N_0=N_0(\mathcal{I},\epsilon)$ such that for $N>N_0$ steps the algorithm returns a value $\epsilon$-close\footnote{By $\epsilon$-close we mean
$|p^*-p_N|\leq \epsilon$, where $p^*$ is the optimal value and $p_N$ is the value obtained after $N$ steps.} to the correct
value.
\end{definition}

In Definition \ref{defcompute} we consider that the problem might have a continuum of possible answers. For problems with a finite set of possible 
solutions, $\epsilon$-closeness is reduced to exact computation. A very important class of problems with a finite set 
of solutions are the \emph{decision problems}. A decision
problem is a  problem with only two possible answers:``YES'' or ``NO''.

A remarkable result is that there exist \emph{uncomputable/undecidable}\footnote{Decidability is the term used for the particular 
case of decision problems.} problems, 
 \ie there exist problems for which it is impossible to
construct a single algorithm that for every input will compute the answer in a finite number of
steps. 

One of the first problems shown to be undecidable was the \emph{halting problem}.
The halting problem is the problem of determining, for a given algorithm $\mathcal{A}$ and input $x$, whether the algorithm stops
(\ie it gives the output in a finite number of steps), 
or if it continues running forever.
The proof of undecidability was presented by Turing \cite{Turing1937} in the same work where he 
introduced the idea of a universal computing machine: the Turing machine (for a nice presentation and discussion of the 
halting problem, see \cite{Penrose}).

\subsection*{Computational complexity}

Computable problems can be classified according to the amount of resources required to solve them. 
And by resources we mean time, memory, energy, and so on. The problems are then classified according to the minimum amount of resources
required by the best possible algorithm that solves it. 

The classes of computational complexity are usually defined in terms of decisions problems. Every optimization problem has a counterpart decision problem
associated to it, for example, instead of asking \textit{`What is the maximum value of a function $f$?'}, we could ask \textit{`Is the maximum value of 
 $f$ greater than $c$?'}. The associated decision problem can be no more difficult than the optimization problem itself (since we could simply solve the 
optimization problem finding the maximum of $f$ and then compare it with $c$), but interestingly many decision problems can be shown to be no easier 
than their corresponding optimization problems \cite{CompIntract}. Therefore the restriction to decision problems does not lose much generality.

Here we are going to consider the classification of the problems in terms of the time required for the solution
of the problem.
Given an input of length $n$, the \emph{time complexity} function of an algorithm, $T(n)$, is the largest amount of time needed by the algorithm to solve
a problem with input size $n$. Usually time complexity is expressed in the `big O notation' which describes the limiting behavior of a function. We say
$T(n)=\mathcal{O}(g(n))$ if there exists
 $n_0$ and a constant $c$ such that
$
  T(n)\leq cg(n)$  for all $ n\geq n_0$.

A problem is considered \emph{easy}, \emph{tractable} or \emph{feasible} if there exists an algorithm that solves the problem using a polynomial in $n$
amount of time. In case there is no such polynomial time algorithm, the problem is
 said to be \emph{hard}, \emph{intractable} or \emph{infeasible}. 

The first complexity class we are going to define is the \emph{class P}, which is the class of problems that can 
be solved by an algorithm with time complexity polynomial in the size of the input.

\begin{definition}[The complexity class P]
 A decision problem $\pr$ belongs to the complexity class P if
 there exists an algorithm $\mathcal{A}$, with time complexity
 $T(n)=\mathcal{O}(p(n))$ (where $p(n)$ is a polynomial in $n$), such that for any instance  $x$ of the problem, $|x|=n$,
 \begin{itemize}
  \item if $\pr(x)=$``YES'' then $\mathcal{A}(x)=$``YES'',
  \item if $\pr(x)=$``NO''  then $\mathcal{A}(x)=$``NO''.
 \end{itemize}
\end{definition}

The class P was introduced by Cobham in 1964 \cite{Cobham64} and suggested to be a reasonable definition of an efficient algorithm.
 A similar suggestion was made by Edmonds in Ref. \cite{Edmonds1965}.
The belief that the class P constitutes the class of efficiently computable
 problems\footnote{Note that a polynomial time algorithm of complexity $n^{100}$ seconds would take many orders of magnitude more than the age of the 
 Universe for an input of size 10, while the exponential algorithm of complexity $2^n$ would take only about 17 minutes for the 
 same input size. However the Cobham–Edmonds thesis is supported by many examples of natural problems and how they scale with the input size.
 Furthermore, polynomial time algorithms involve a deep knowledge of the structure of the problem, in contrast with exponential
 time algorithms which  are usually a mere brute-force search over all possibilities. Here we 
 are just going to assume that the class P is a reasonable definition of efficient (for more discussion on this point, see \cite{CCBarak, CompIntract}). }
 is called the \emph{Cobham–Edmonds thesis}.

Another important class is the \emph{class NP}. NP is the class of decision problems that can be efficiently verified\footnote{Originally the class NP was defined in terms of non-deterministic Turing machines (a very abstract computational model),
and only later it was recognized as the class of problems that can be easily verified (see \cite{defPNP}).}, \ie 
once a proof $y$ is provided together with the input $x$, one can check in polynomial time whether the 
answer is ``YES''.

\begin{definition}[The complexity class NP]
 A decision problem $\pr$ belongs to the complexity class NP if there exists
 an algorithm $\mathcal{V}$,  of time  complexity  $T(n)=\mathcal{O}(p(n))$ (where $p(n)$ is a polynomial in $n$), 
 such that for any instance of the problem $x$, $|x|=n$,
 \begin{itemize}
  \item if $\pr(x)=$``YES'' then there is a proof $y$ such that $\mathcal{V}(x,y)=$``YES'',
  \item if $\pr(x)=$``NO'' then for all proofs $y$  $\mathcal{V}(x,y)=$``NO'',
 \end{itemize}

\end{definition}

It is easy to see that P$\subseteq$ NP, since for a problem in P one can simply ignore the proof and solve the problem in polynomial time.
But whether or not P$=$NP is one of the biggest open problems in computer science.

\subsection*{Complete and hard problems}

Many researches believe P$\neq$NP, based on the fact that some problems in NP seems to be intrinsically more difficult than the problems in P.
However up to now no formal proof in any direction was ever found. An intermediate advance in the classification of NP problems was made by the 
introduction of the concept of a \emph{polynomial time reduction}, which allowed to select the hardest problems of the class NP. These hardest problems 
are the ones for which it is most unlikely to find an efficient algorithm, and in case P$\neq$NP these problems 
definitely belong to the non-intersecting region.

\begin{definition}[Polynomial time reduction]\label{defreduction}
A problem $\pr_1$ is polynomial time reducible to $\pr_2$, if there exists a polynomial time algorithm  $\mathcal{R}$ such that for every input $x$
of problem $\pr_1$
\begin{align}
 \pr_1(x)=\pr_2(\mathcal{R}(x)),
\end{align}
and in this case we say $\pr_1 \preceq_{\mathcal{R}} \pr_2$.
\end{definition}

The idea of a reduction is to map a problem $\pr_1$ into another problem $\pr_2$, such that by solving $\pr_2$ one is able to get the solution of $\pr_1$.
However, since we are concerned with efficiency, a good reduction is one that can be performed in polynomial time. With that in mind, 
by Definition \ref{defreduction}
we see that if problem $\pr_2$ is efficiently solvable then $\pr_1$ is also efficiently solvable\footnote{We simply have to apply
the reduction algorithm $\mathcal{R}$ which takes polynomial time, and then we solve $\pr_2$ which also takes polynomial time.}. And if $\pr_1$ is not efficiently solvable, we can conclude that $\pr_2$ 
cannot be efficiently solvable, otherwise we would have a contradiction. Therefore if $\pr_1 \preceq_{\mathcal{R}} \pr_2$, then 
we can say that $\pr_2$ is at least as hard as $\pr_1$.

\begin{definition}[NP-hard problems]
 A problem $\pr$ is NP-hard if there exists a polynomial time reduction of every problem $\pr' \in$ NP to problem $\pr$:
 \begin{align}
  \pr' \preceq_{\mathcal{R}} \pr\;\; \forall \pr' \in \text{NP}.
 \end{align}
\end{definition}

\begin{definition}[NP-complete problems]
 A problem $\pr$ is NP-complete if $\pr$ is NP-hard and if $\pr \in$ NP.
\end{definition}

The NP-complete problems are the hardest problems  of the NP class, since by finding a polynomial time algorithm for solving an NP-complete problem one 
automatically solves any problem in NP in polynomial time (and then would have proved P$=$NP!).

Note that once we identify an NP-complete problem $\pr_1$, by reducing it to an NP problem $\pr_2$, we automatically prove that $\pr_2$ is also NP-complete.
Hence the concept of reduction opens the possibility of many proofs of hardness in the field of computational complexity.
The first proof of NP-completeness was given by Stephen Cook in Ref. \cite{Cook1971}, where he showed that the SAT problem\footnote{The SAT (satisfiability 
problem) is the problem of determining whether there exists a consistent assignment for the variables of a particular Boolean circuit such that
the whole expression is evaluated as true. For example, the Boolean circuit $(x_1\vee x_2)\wedge x_3$ can be evaluated as true with the 
assignment $x_1=$true, $x_2=$false, and $x_3=$true.} is
NP-complete (known as the Cook-Levin theorem{\cite{Cook1971,Levin1973}). 

In Ref. \cite{Karp21},  Richard Karp uses Cook-Levin theorem in order to show that there is a polynomial reduction from the SAT 
problem to each of 21 combinatorial and graph theoretical computational problems. In particular a $\DE{0,1}$-integer programming and the calculus of the 
independence number of a graph (that we are going to discuss later) are NP-complete problems.


\section{Optimization problems}\label{secopt}

In the previous Section we presented the concepts of computability and uncomputability. Also, we have seen
that the computable problems can be divided in classes of complexity which classify the problems according to how many resources 
are necessary to solve it.
In this section we discuss a bit of the theory of optimization following approaches of Refs. \cite{BentalConvexOpt} and \cite{Boyd}.

Let us consider an optimization problem where we want to minimize a function $f_0$ subjected to some constraints:
\begin{align}\label{eqOptimization}
 (P):\begin{cases}
      \min & f_0(x)\\
      \st & f_i(x)\geq 0 \;,\;i=1,\ldots,m \\
      & h_i(x)=0\;,\;i=1,\ldots,p
     \end{cases}
\end{align}
where
\begin{itemize}
 \item $x=(x_1,\ldots,x_n)\in \R^n:$ is the \emph{optimization variable},
 \item $f_0:$ is an $n$-variable real function called \emph{objective function},
 \item $f_i,\;h_i:$ are $n$-variable real functions called \emph{constraint functions}.
\end{itemize}

The set of points for which the objective and constraint functions are defined is called the \emph{domain}
of the problem $(P)$:
\begin{align}
 \mathcal{D}=\bigcap_{i=0}^{m} \text{dom } f_i \cap \bigcap_{i=1}^{p} \text{dom } h_i.
\end{align}
A point $x \in \mathcal{D}$ is \emph{feasible} if it satisfy all the constraints, \ie $f_i(x)\geq 0 \;,\; i=1,\ldots,m$ and $h_i(x)=0\;,\;i=1,\ldots,p$.
The set
of all feasible points is called the \emph{feasible set} $\mathcal{F}$,
\begin{align}
 \mathcal{F}=\DE{x\mid x\in \mathcal{D}, \;f_i(x)\geq 0\; , \; i=1,\ldots,m\;,\; h_i(x)=0\;,\;i=1,\ldots,p}.
\end{align}

The \emph{optimal value} $p^*$ of problem $(P)$ is the infimum of $f_0$ over the feasible points
\begin{align}
 p^*=\inf_{x\in \mathcal{F}} f_0(x).
\end{align}
If the optimal value is achieved by a feasible point then the problem is said to be \emph{solvable}. However, for some problems the 
optimal value may not be achieved by any feasible point.\vspace{1em}

The simplest optimization problem is the \emph{linear programming} (LP) where the objective and constraints are affine functions\footnote{A function $f$ 
is affine if $f(ax+by) =af(x)+bf(y)$
$\forall$ $a+b=1$.}. 
For an LP the numerical method of \emph{interior-point}\footnote{For details of the interior-point method see Ref. \cite{Boyd}.}, developed in the 1980s,
can solve it efficiently with $\mathcal{O}(nm^2)$ operations, where $n$ is the number of variables and $m$ the number
of inequality constraints. Therefore
LP $\in$ P.

Many advances in numerical methods for solving optimization problems are due to the recognition that the interior-point method
can also be used to solve other \emph{convex optimization} problems efficiently. 
Convex optimization problems are the ones where the objective and constraint functions are 
convex\footnote{A function $f$ is convex if $f(ax+by) \leq af(x)+bf(y)$
$\forall$ $a,b \geq 0$, $a+b=1$.}. 
Hence a convex optimization problem is usually considered a tractable one, whereas non-convex problems are in general hard.
Fortunately many interesting problems in many areas: physics, mathematics, engineering
and so on, can be phrased as a convex optimization problem.

A  particular case of convex optimization problem is the semidefinite programming (SDP). For SDPs, algorithms which utilize 
the method of interior-point are 
well established, therefore, these problems can also be solved efficiently (in polynomial time in the number of variables).
For more general convex problems the numerical methods are not so well established as for LP and SDP, still the interior-point 
methods work well in practice. 

As nicely pointed by Boyd and Vandenberghe \cite{Boyd} these numerical methods for solving these problems 
are so well structured that they can be considered a technology:

\begin{center}
\begin{minipage}{12cm}
 \textit{``Solving} [LP and SDP] \textit{is a technology that can be reliably used by many people who do not know, and do not need to know, 
 the details.''}
\end{minipage}
\end{center}

 In this Section, we present the formal definitions of linear programming (LP), semidefinite programming (SDP), and integer programming (IP).

\subsection*{Linear optimization}

A linear programming (LP) is an optimization problem \eqref{eqOptimization} where the objective and constraint functions are affine.
An LP can be expressed as
\begin{align}\label{eqLP}
 (LP):\begin{cases}
       \min & \braket{c}{x} \\
       \st & A\ket{x}\geq \ket{b}
      \end{cases}
\end{align}
where all the constraint functions are expressed in a unique vector inequality\footnote{And remember that an equality constraint $a=b$ can always be expressed
as two inequality constraints: $a\geq b$ and $a\leq b$. }, $A\ket{x}\geq \ket{b}$, which represents a component-wise relation
$\ket{a} \geq \ket{b} \Leftrightarrow a_i \geq b_i, \;\forall\; i$. 
$\ket{x} \in \R^n$, $\ket{c}\in \R^n$, $\ket{b} \in \R^m$, and
$A$ is a $m \times n$ matrix. We are making use of Dirac's notation\footnote{The called ``braket'' notation, used in quantum theory, was introduced by Dirac in
Ref. \cite{braket}.} for consistency with the other chapters.

\subsection*{Semi-definite programming}
In order to generalize an LP one can relax the linearity condition of the objective or the constraint functions. However another way to 
generalize an LP that leads to a class of very interesting problems is to keep the objective and constraint functions linear 
but to relax the meaning of $\geq$ in the inequality constraints.

The order relation $\geq$ in an LP, Eq. \eqref{eqLP}, is a coordinate-wise relation where $\ket{a}=(a_1,\ldots,a_n)$ and $\ket{b}=(b_1,\ldots,b_n)$ satisfy
\begin{align}
 \ket{a} \geq \ket{b} \Leftrightarrow \DE{a_i \geq b_i, \;\forall\; i=1,\ldots,m}.
\end{align}
However a partial order relation $\geq$ can be defined in a more general framework.
A good partial ordering is completely determined by a subset $K$, of
a vector space $E$, where the relation $\geq_K$ is defined as:
\begin{align}
 a\geq_K b \Leftrightarrow a-b \geq_K 0 \Leftrightarrow a-b \in K,
\end{align}
and $K$ determines the set of positive elements:
\begin{align}
 K=\DE{a \in E \mid a \geq 0}.
\end{align}
In order to satisfy some expected properties (see \cite{BentalConvexOpt} for more details) the set $K$ cannot be arbitrary and it has to be a 
\emph{pointed cone}, \ie
\begin{enumerate}[(i)]
 \item $K$ is nonempty and closed under addition: $a,a' \in K \Rightarrow a+a' \in K$.\label{1}
 \item $K$ is a conic set: $a \in K, \lambda \geq 0\Rightarrow \lambda a \in K$.\label{2}
 \item $K$ is pointed: $a\in K$ and $-a \in K\;\Rightarrow\; a=0$.\label{3}
\end{enumerate}

An optimization problem whose constraints are defined by the partial ordering $\geq_K$, for a set $K$ satisfying properties \eqref{1}-\eqref{3},
is called a \emph{conic problem}.

Now let $K$ be the cone of symmetric positive semidefinite $m \times m$ matrices $S_+^m$, this defines a semidefinite programming (SDP):
\begin{align}\label{eqSDP}
 (SDP):\begin{cases}
        \min & \braket{c}{x}\\
        \st & \mathcal{A}(\ket{x})-B \geq_{S_+^m}0
       \end{cases}
\end{align}
where $\mathcal{A}: \R^n \longrightarrow S^m$ is a linear map from vectors in $\R^n$ to the space of symmetric $m \times m$ matrices $S^m$.
$B \in S^m$, and $\ket{c}$ and $\ket{x}$ are vectors in $\R^n$.

The standard form of an SDP (and the one we are going to deal with in the following chapters) is
\begin{align}\label{eqSDPpadrao}
 (SDP):\begin{cases}
        \min & \Tr (CX)\\
        \st & \Tr (A_iX)=b_i, \; i=1,\ldots,r\\
        & X \geq_{S_+^m} 0        
       \end{cases}
\end{align}
where $X, C, A_i \in S^m$.

The formulations \eqref{eqSDP} and \eqref{eqSDPpadrao} are equivalent, and a problem in form \eqref{eqSDP} can 
always be rephrased into the from \eqref{eqSDPpadrao} and vice-versa\cite{Boyd}. 

From now on we are going to omit the subscript in the ordering relation $\geq_{S_+^m}$, but from the context it will be clear 
which ordering relation is being applied.

\subsection*{Integer programming}

An integer programming (or integer linear programming) is a problem where the objective and constraint functions are affine functions but the variables are restricted to be integers\footnote{
This restriction can be seen as a non-linear constraint, for example, if $x$ is restricted to assume values $\DE{-1,1}$, this is equivalent to require the 
quadratic constraint $x^2=1$ to be satisfied.}. $\DE{0,1}$-integer programming is the particular case of integer programming where the variables are 
restricted to assume the values 0 or 1. A general $\DE{0,1}$-integer programming (IP) can be written as

\begin{align}\label{eqIPpadrao}
 (SDP):\begin{cases}
        \min & \braket{c}{x}\\
        \st & A\ket{x}-\ket{b} \geq 0\\
        & \vec{x} \in \DE{0,1}^n
       \end{cases}
\end{align}
where $\ket{b} \in \R^m$ and $A$ is an $m\times n$ real matrix. 

Problems of the form \eqref{eqIPpadrao} appear in the study of nonlocality, in the calculus of the classical value of a Bell expression.
Note that a Bell expression is a linear function of the joint probability distributions $P(a,b|x,y)$, and one can easily 
 check that imposing  no-signaling constraints (which are linear constraints) together with determinism, \ie $P(a,b|x,y) \in \DE{0,1}$, is equivalent to 
 impose deterministic locality: $ P(a,b|x,y)=D(a|x)D(b|y)$.
Consequently, the search  over all deterministic local boxes, which
is sufficient to obtain the classical value of a Bell expression (see Chapter \ref{chapterNL}),  is a $\DE{0,1}$-integer programming.

As previously mentioned, $\DE{0,1}$-integer programming was shown to be an NP-complete problem \cite{Karp21}, therefore the task of obtaining the 
classical value of a Bell inequality is hard in general.

\section{Duality}\label{secduality}

The idea of the dual of a problem $(P)$ is to play with the constraint inequalities (summing equations, adding trivial inequalities $1>0$ and so on)
in order to obtain a quantity that is always smaller then the optimal value of problem $(P)$. In this Section we will study the theory 
of \emph{Lagrange duality} \cite{Boyd}.

Given an optimization problem $(P)$:
\begin{align}
 (P):\begin{cases}
      \min & f_0(x)\\
      \st & f_i(x)\geq 0 \;,\;i=1,\ldots,m \\
      & h_i(x)=0\;,\;i=1,\ldots,p
     \end{cases}
\end{align}
the \emph{Lagrangian function}, $L:\mathcal{D} \times \R^m \times \R^p \longrightarrow \R$, is defined as
\begin{align}\label{eqLagrangian}
 L(x,\lambda, \nu) =f_0(x)-\sum_{i=1}^m \lambda_i f_i(x) +\sum_{i=1}^p\nu_i h_i(x).
\end{align}
where $\lambda=(\lambda_1,\ldots,\lambda_m)$ and $\nu=(\nu_1,\ldots,\nu_p)$.
Note that if $\lambda \geq 0$, for every feasible point $x \in \mathcal{F}$ of $(P)$ we have $L(x, \lambda, \nu) \leq f_0(x)$.

Now we define the \emph{dual function}, $g: \R_+^m \times \R^p \longrightarrow \R$ as
\begin{align}
 g(\lambda,\nu)=\inf_{x \in \mathcal{D}}L(x, \lambda, \nu).
\end{align}
The dual function $g$ is the point-wise infimum of a family of affine functions of the variables $(\lambda,\nu)$, and 
 hence $g(\lambda,\nu)$ is always  concave\footnote{A point-wise infimum of an affine
 function over variable $\lambda$ can in general be written as $ g(\lambda)=\inf_{x} a(x)\lambda+b(x)$. Now, taking $c_1,c_2 \geq 0$,
 $c_1+c_2=1$ we have
 \begin{align*}
  g(c_1\lambda_1+c_2\lambda_2)&=\inf_{x} \De{c_1(a(x)\lambda_1+b(x))+c_2(a(x)\lambda_2+b(x))}\\
  &\geq c_1\inf_{x} \De{a(x)\lambda_1+b(x)}+c_2\inf_{x}\De{a(x)\lambda_2+b(x)}\\
  &=c_1g(\lambda_1)+c_2g(\lambda_2).
  \end{align*}
} even though no structure was assumed about the problem $(P)$.

Let $p^*$ be the optimal solution of the problem $(P)$. By construction we have that 
\begin{align}
 g(\lambda, \nu)=\inf_{x \in \mathcal{D}}L(x, \lambda, \nu)\leq \inf_{x \in \mathcal{F}}L(x, \lambda, \nu)\leq \inf_{x \in \mathcal{F}}f_0(x)= p^* .
\end{align}
And therefore $ g(\lambda, \nu)$ is always a lower bound to the value of problem $(P)$. One can then look for the best lower bound that can be 
obtained from $g$, and that is the 
idea of the \emph{Lagrange dual problem} $(D)$:
\begin{align}
 (D):\begin{cases}
      \max &  g(\lambda, \nu)\\
      \st & \lambda \geq 0
     \end{cases}.
\end{align}
{Note that the Lagrange dual problem is the maximization of a concave function subjected to linear
constraints. The maximization of $g$ is equivalent to the minimization of $-g$, which is then a convex function. Therefore 
$(D)$ is a  convex optimization problems.} 

\begin{theorem}[Weak duality]\label{thmweakdual}
 Let $p^*$ be the optimal value of a problem $(P)$, and $(D)$ be the corresponding Lagrange dual problem with optimal value $d^*$. It holds that
 \begin{align}\label{eqweakdual}
  d^* \leq p^*.
 \end{align}
 Where $p^*-d^*\geq 0$ is the \emph{duality gap}.
\end{theorem}

The Weak duality Theorem follows by construction of the dual problem. 
Weak duality holds in general for any kind of optimization problem, as no restriction on the nature of the objective and constraint functions was made
for the construction of the Lagrangean. However for many convex problems an even stronger result holds, that the optimal value of the dual is 
actually equal to the optimal value of the problem $(P)$. This is stated by Slater's condition.

\begin{theorem}[Strong duality- Slater's condition]\label{thmstrongdual}
 Given a convex optimization problem $(P)$ of the form
 \begin{align}
  (P):\begin{cases}
      \min & f_0(x)\\
      \st & f_i(x)\geq_K 0 \;,\;i=1,\ldots,m \\
      & Ax-b=0\;
     \end{cases}
\end{align}
where $f_0$ is a convex function bounded below and $f_1, \ldots, f_m$ are $K$-convex functions\footnote{A function $f$ is $K$-convex if
$ f(cx+(1-c)y)\leq_K cf(x)+(1-c)f(y)$, for $0\leq c\leq 1$.
}.  If there exists a \emph{strictly feasible} point\footnote{A point $x$ belongs to the relative interior of a set $C\subseteq \R^n$, $x\in\text{relint}(C)$ if
\begin{align}
\exists r>0 \;\text{s.t.} \;\forall y\in \text{Aff}(C), |y-x|< r\; \text{then}\; y\in C.
\end{align}
where $\text{Aff}(C)=\DE{x|x=\lambda x_1+(1-\lambda)x_2, \; \text{for}\;  x_1,x_2 \in C, \lambda \in \R}$.
}
\begin{align}\label{eqslater}
 x \in \text{relint}(\mathcal{D}) \st\;\; f_i(x)>_K0\;,\;i=1,\ldots,m \;\; \text{and}\;\;Ax-b=0 ,
\end{align}
then $ d^*=p^*$.
\end{theorem}

There are other results which establishes conditions for strong duality for non-convex problems. These conditions are in general called 
\emph{constraint qualifications} \cite{Boyd}.

\section{SDP relaxations of hard problems}

The idea of a relaxation is the following: In an optimization problem $(P)$ we want to find the optimum value of a function $f_0(x)$ searching over the 
domain $\mathcal{F}$, which is determined by the constraints of the problem. However even when the function $f_0(x)$ is simple (a linear function for example),
it might be the case (and it is the case in many interesting problems) that the domain $\mathcal{F}$ is extremely hard to characterize. 
An alternative way to deal with this difficulty is to consider the problem $(P')$, where, instead of searching for the minimum of $f_0(x)$ over the set $\mathcal{F}$, we make
the search over a bigger (relaxed) set $\mathcal{F}'$ (see Figure \ref{figrelaxation}), $\mathcal{F}\subseteq\mathcal{F}'$, which is simpler to describe.

\begin{figure}[h]
\begin{center}
 \includegraphics[scale=0.6]{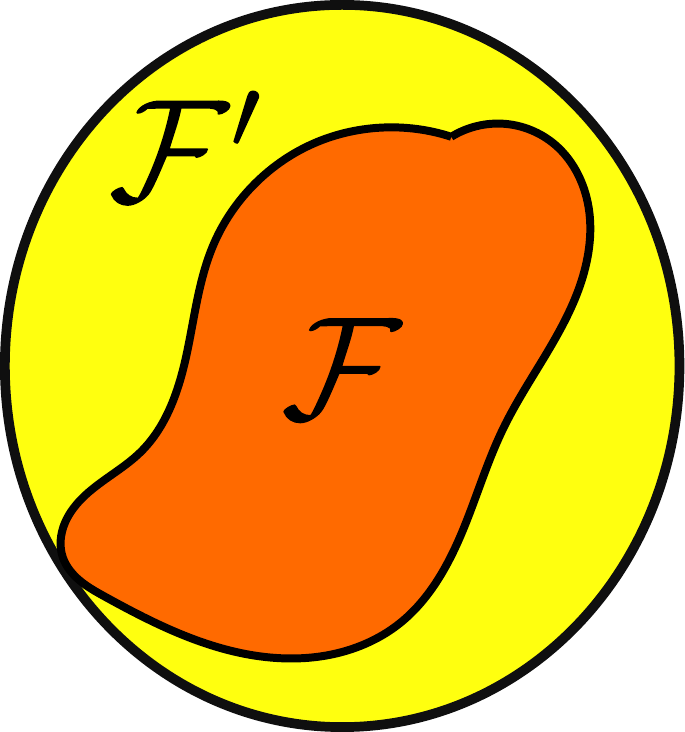}
 \caption{ The idea of a relaxation: Instead of optimizing $f_0(x)$ over the set $\mathcal{F}$, one considers the problem of finding the optimum 
 value of $f_0(x)$ over the simpler 
 set $\mathcal{F}'$.}\label{figrelaxation}
 \end{center}
 \end{figure}
 
Since $\mathcal{F}\subseteq\mathcal{F}'$, the optimal value obtained for problem $(P')$ is smaller than or equal the optimal value obtained for problem $(P)$.
Hence a relaxation is a way to get lower bounds for an optimization problem $(P)$.

A relaxation is wanted to satisfy two features: it should be \textit{efficiently solvable} (\ie the relaxed set has to be nicely characterized), and at 
the same time it should be \textit{good} meaning that the value obtained in the relaxation is close to the actual value (we do not want a relaxation
that gives a completely non-informative result).

As an example let us consider the problem of finding the independence number of a graph $\gr(V,E)$\footnote{The independence number of a graph is the 
maximum number of vertices such that no two of each are connected by an edge (see Chapter \ref{chaptergraphNL}).} with $n$ vertices, which can be formulated as the 
following $\DE{0,1}$-integer programming:
\begin{align}\label{eqalpha}
 \alpha =\begin{cases}
          \max &\sum_{i=1}^n x_i\\
          \st & x_i+x_j\leq 1 \;\text{if}\; \DE{x_i,x_j}\in E\\
          & x_i \in \DE{0,1}\; \forall i.
         \end{cases}
\end{align}

As we have argued before, this is an NP-complete problem and hence considered intractable for large $n$. A simple relaxation can be derived by just turning
the nonlinear constraint $ x_i \in \DE{0,1}$ into a linear one
\begin{align}\label{eqalpharelax}
 \alpha' =\begin{cases}
          \max &\sum_{i=1}^n x_i\\
          \st & x_i+x_j\leq 1 \;\text{if}\; \DE{x_i,x_j}\in E\\
          & 0\leq x_i \leq 1\; \forall i.
         \end{cases}
\end{align}

We have that $ \alpha' \geq \alpha$ for every graph $\gr(V,E)$, since we now allow $x_i$ to assume all the values between 0 and 1. 
Moreover, problem \eqref{eqalpharelax}
is a linear program which can be efficiently solved.

For a long time the only known practical relaxations were the LP ones. However, with the advent, over the last decades, of techniques for efficiently solving 
semidefinite programs, it came the possibility of exploring semidefinite relaxations, which has become a fruitful area
of research\footnote{For a detailed discussion see the Preface of Ref. \cite{BentalConvexOpt}.}.
Semidefinite relaxations have been shown to be particularly useful for combinatorial
problems. We are going to see, in Section \ref{seczeroerror}, an SDP relaxation for the independence number problem \eqref{eqalpha} (the Lovász number).


\subsection*{NPA hierarchy}\label{secNPA}

We have stated that calculating the classical value of a Bell inequality is a $\DE{0,1}$-integer
programming, and thus it is NP-hard. 
For the quantum value the situation is even worse! Note that in Chapter \ref{chapterNL}, when we define the quantum boxes, we make no restriction over
the dimension of the system, and hence to obtain the quantum value one should optimize over all possible states and measurements in all possible
dimensions. There is no known algorithm to determine the quantum value of a general Bell inequality\footnote{We are going to see in Chapter 
\ref{chaptergames} that for a particular class of Bell inequalities, the {\xor} games, the quantum value can be determined efficiently by an SDP.} in a 
finite number of steps and therefore
this problem may even be uncomputable (In Ref. \cite{BelenCombinatorial} the authors conjecture that it is actually non-computable.).

The most general method known to deal with this intractable problem was introduced by Navascués, Pironio and Acín, in Ref. \cite{NPA}: the \emph{NPA hierarchy}.
The NPA hierarchy is a hierarchy of semidefinite programs where each level corresponds to an optimization over a tighter relaxation of the quantum set of correlations.
These sets are nicely constrained by a $\geq_{S_m^+}$ relation and therefore calculating
the optimal value of a Bell inequality over a level of the hierarchy is a semidefinite program.
The hierarchy is proved to converge to the set $Q'$ which is defined as following:
\begin{definition}
 The set $Q'$ is the set of boxes $\boxp$ such that
 \begin{align}
  P(a,b|x,y)= \bra{\psi}E_x^aE_y^b \ket{\psi}
 \end{align}
 for some state $\ket{\psi} \in \Hi$ and projective measurements\footnote{Since we do not fix dimension, there is no loss of generality in restrict to 
 pure states and projective measurements. This is due to Naimark's Theorem, see 
\url{https://cs.uwaterloo.ca/~watrous/CS766/LectureNotes/05}.} $\DE{E_x^a}$ and $\{E_y^b\}$, acting on $\Hi$, satisfying
\begin{align}
\De{E_x^a,E_y^b}=0\; \forall a,b,x,y.
\end{align}
\end{definition}

Note that the set of quantum correlations $Q$ is contained in $Q'$, since all local measurements of Alice commutes with local measurements of Bob.
But whether or not $Q=Q'$ is an open problem known as Tsirelson's problem \cite{TsirelsonProblem} (see discussion in Section \ref{sec.QuantumCorr}). 

The NPA hierarchy constitutes one of the most powerful tools in the field of nonlocality, and it has led to the derivation of innumerous results.
However the quality of the approximation achieved by these bounds remains unknown in general.
Moreover for a Bell inequality with $m$ inputs and $d$ outputs per party, the $n$-th level of the hierarchy involves  a
matrix of size $O((2md)^n)$ as an SDP variable, so, in general, the complexity
increases exponentially with the level of the hierarchy.




%% file: games.tex
\chapter{Nonlocal games}\label{chaptergames}

Some Bell inequalities can be naturally phrased in the framework of a game. A \emph{nonlocal game} is a cooperative task where the players 
receive questions from a referee and they are supposed to give answers in order to maximize some previously defined payoff function.
Upon starting the game, the players are not allowed to communicate anymore, therefore any strategy has to be agreed in advance. 

Nonlocal games have a wide range of applications. 
They play an important role in the study of communication complexity \cite{NLCommCompl,BruknerCommComplexPhys}  (and vice-versa) and in the 
{formulation of device-independent cryptographic protocols \cite{Ekert1991, ColbeckRennerAmplify}}. 

In a computer science language, a nonlocal game with $n$ players can be seen as the particular case of 
\emph{multiprover interactive proof systems}  with $n$ provers and one round.
An interactive proof system consists of an all powerful\footnote{Powerful in a computational sense, meaning that the prover has unbounded
resources, although only classical resources, and unlimited computational power.}, but untrusted, prover
who wants to convince a verifier, who has limited computational
power, of the truth of some statement by exchanging messages in many rounds. A multiprover interactive proof system is an interactive proof system with
many provers, who may be bounded not to communicate during the proof.
Multiprover interactive proof systems were introduced in Ref. \cite{MIP88} as an alternative to allow for the performance of some cryptographic 
tasks without relying on extra assumptions, such as the existence of one-way functions\footnote{A one-way function $f$ is a function
that can be computed in polynomial
time for any input $x$, however the function $f$ is hard to invert. The existence of one-way functions would imply 
that P$\neq$NP.}
or limitations on the computational power. With the  introduction of many provers
these extra assumptions can be replaced by the condition of physical separation of the provers during the course of the protocols.
For further remarks on the connection of interactive proof systems with entangled provers and the quantum value of 
nonlocal games see Ref. \cite{CleveNLgames}.


In this chapter we present definitions and results on nonlocal games.
In the first sections we restrict the presentation to the case of $2$-player games. The case of $n$-player games is 
discussed in Section \ref{secnplayer}. For a nice introduction to nonlocal games see Ref. \cite{CleveNLgames}.

\section{Definitions}

 \begin{definition}[Nonlocal Game]
A nonlocal game $g(V,p)$ is a cooperative task 
where $2$ players, Alice and Bob, {who} are not allowed to communicate after 
the game starts, receive respectively questions $x \in \inp_A$ and $y \in \inp_B$, chosen
from a probability distribution $p(x,y)$ by a referee. 
 {Upon} receiving the questions, Alice is supposed to give an answer $a \in \out_A$ and Bob $b \in \out_B$. 
 The winning condition of the game is defined by the payoff function $V(a,b|x,y)$ which
 assumes value $1$ to indicate when the players win and value $0$ to indicate when they lose. 
 \end{definition}
 
 Given a particular strategy applied by the players, which is specified by a box $\boxp$, 
 the figure of merit that we are interested in analyzing is the \emph{average probability of success}  given by
 \begin{align}\label{eqw}
  \omega(g)= \sum_{a,b,x,y} p(x,y) V(a,b|x,y) P(a,b|x,y).
 \end{align}
 Note that $\omega(g)$ can be regarded as a Bell expression, since it is a linear function of the joint probability distributions $P(a,b|x,y)$.

 \subsection*{Classical strategies}
%
 
 The maximum average probability of success optimizing over all possible classical strategy is the \emph{classical value} of the game, 
denote $\omega_c(g)$. In order to obtain $\omega_c(g)$ we have to optimize over the local boxes of the particular Bell scenario
defined by the game. As we argued before, the maximum value of $\omega(g)$ is attained by a deterministic strategy, hence:
\begin{align}
 \omega_c(g)=\max_{\vec{D}(a|x),\vec{D}(b|y)}\,\sum_{a,b,x,y} p(x,y) V(a,b|x,y) D(a|x)D(b|y), 
\end{align}
where $\vec{D}(a|x)$ and $\vec{D}(b|y)$ are deterministic probability distributions.
The number of possible deterministic strategies for a particular game is ${|\out_A|}^{|\inp_A|}\times {|\out_B|}^{|\inp_B|}$,
which increases exponentially with the number
of inputs.

\subsection*{Quantum strategies}

A general quantum strategy is described by the players sharing a bipartite quantum state $\rho$ of arbitrary dimension and giving their
answers according to the result of
 local measurements, $\DE{M_{x}^{a}}$ and $\{M_{y}^{b}\}$, that they perform in their systems:
 \begin{align}
    P(a,b|x,y)= \Tr \De{(M_{x}^{a}\otimes M_{y}^{b}) \rho}.
 \end{align}

Since we do not make any restriction on the Hilbert space dimension of the system we can actually
restrict ourselves to pure states and projective measurements\footnote{It is a consequence of Naimark's Theorem that an arbitrary POVM  on 
a system of
Hilbert space $\Hi$ is equivalent to a global projective measurement in $\Hi \otimes \Hi'$, where $\Hi'$ is the Hilbert space of an auxiliary system. See 
\url{https://cs.uwaterloo.ca/~watrous/CS766/LectureNotes/05}.}. 
Therefore, the \emph{quantum value} of the game, $\omega_q(g)$, which is the maximum average probability of success of 
players applying a quantum strategy, is given by:
  \begin{align}\label{eqwqgeral}
\omega_q(g)=\sup_{\ket{\psi},\{M_{x}^{a}\},\{M_{y}^{b}\}}\;\sum_{a,b,x,y} p(x,y) V(a,b|x,y) \bra{\psi}M_{x}^{a} \otimes M_{y}^{b}\ket{\psi}.
 \end{align}
where $\ket{\psi} \in \Hi_A \otimes \Hi_B$, for arbitrary Hilbert spaces $\Hi_A$ and $\Hi_B$, and $\DE{M_{x}^{a}}$ and $\{M_{y}^{b}\}$ are projective measurements.

As discussed before, we do not put any restriction on the dimension of the system
and, as a particular class of Bell inequalities,
it is not known if the quantum value of a nonlocal game is computable in general (see Section \ref{secNPA}).

\subsection*{No-signaling strategies}

We also define the \emph{no-signaling value} of the game, denoted $\omega_{NS}$ 
  \begin{align}
\omega_{NS}(g)=\max_{\boxp\in \mathcal{NS}}\;\sum_{a,b,x,y} p(x,y) V(a,b|x,y) P(a,b|x,y).
 \end{align}
 
 The no-signaling value is easily calculated by a linear programming, since we have a linear function of the variables $P(a,b|x,y)$, subjected
 to linear constraints given by the no-signaling conditions \eqref{eqNS}.

\section{{\xor} games}

In this Section we focus on the so-called {\xor} games, introduced in Ref. \cite{CleveNLgames}.
{\xor} games are the simplest class of nonlocal games where the players have two possible answers, $a,b \in \DE{0,1}$, and the 
payoff function only depends on the sum modulo two\footnote{The sum modulo 2 is equivalent to the logical operation exclusive or, also 
denoted {\xor}, that is the reason for the name of the games.} of 
their outputs:  
\begin{align}
 V(a,b|x,y)=\begin{cases}
             1 & \text{if}\; a \oplus b = f(x,y)\\
             0 &\text{otherwise},
            \end{cases}
\end{align}
for some function $f: \inp_A\times \inp_B \longrightarrow \DE{0,1}$.

\begin{definition}[{\xor} games]\label{defxor}
 An {\xor} game, $g^{\oplus}(f,p)$ is a nonlocal game where Alice and Bob receive 
 respectively questions $x \in \inp_A$ and $y \in \inp_B$, chosen
from a probability distribution $p(x,y)$ by a referee. 
 {Upon} receiving the questions, Alice outputs a bit $a$ and Bob outputs bit $b$. 
 The players  win the game if $a \oplus b = f(x,y)$. 
\end{definition}

The \emph{average probability of success} \eqref{eqw} for an {\xor} game can be written as\footnote{Where we have used the fact that 
 \begin{align}
  P(a \oplus b=f(x,y)|x,y)=\frac{1}{2} + \frac{P(a \oplus b=f(x,y)|x,y)-  P(a \oplus b\neq f(x,y)|x,y)}{2}.
 \end{align}}
 \begin{align}\label{wxor}
  \omega(g^{\oplus})= \frac{1}{2}+\frac{1}{2}\de{\sum_{x,y} p(x,y)(-1)^{f(x,y)}\De{ P(a\oplus b=0|x,y)-P(a \oplus b=1|x,y)}}.
 \end{align}
The first $\frac{1}{2}$ on the RHS of Eq. \eqref{wxor} can be interpreted as the probability of success of the players when they apply a totally 
random strategy (\ie if upon receiving her input Alice tosses a coin to determine her output $a$, 
and Bob does the same).

The bias of the game represents how much a particular strategy is better (or worse) than the completely random strategy, and is defined as:
\begin{align}\label{bias}
\epsilon \defeq& \;2 \omega - 1\\
=&\sum_{x,y} p(x,y)(-1)^{f(x,y)}\De{ P(a\oplus b=0|x,y)-P(a \oplus b=1|x,y)}. \nonumber
\end{align}

In expression \eqref{wxor} we see that $\omega(g^{\oplus})$ is related to the expected value of binary 
observables (see Eq. \eqref{eqobservable}), $\mean{A_xB_y}=P(a\oplus b=0|x,y)-P(a \oplus b=1|x,y)$. Therefore {\xor} games 
are equivalent to the important class of Bell expressions that involve only terms with correlators $\mean{A_xB_y}$, the
\emph{full-correlation Bell inequalities}, which is a widely studied class of Bell inequalities \cite{StephanieSDP,WWfullcorre}. The
CHSH inequality \eqref{eqCHSH} being the most remarkable example of a full-correlation Bell inequality.

\subsection*{The game matrix}

To every {\xor} game we can associate a $|\inp_A|\times |\inp_B|$ matrix, the \emph{game matrix}
$\Phi$, which carries all the information
necessary to describe the game: the inputs' probability distribution and the winning condition specified by $f(x,y)$:
\begin{align}\label{eqxormatrix}
 \Phi=\sum_{x,y}p(x,y)(-1)^{f(x,y)}\ketbra{x}{y},
\end{align}
where $\DE{\ket{x}}$ and $\DE{\ket{y}}$ define orthonormal basis and $x$ and $y$ run over the inputs of Alice and Bob respectively.

\subsection*{SDP characterization of the quantum value of an {\xor} game}

For the particular class of bipartite {\xor} games, a theorem due to Tsirelson \cite{Tsirelsonthm1,Tsirelsonthm2}, leads to the result that the quantum 
value of these games can be computed efficiently by a semidefine program \cite{CleveNLgames, StephanieSDP}.


\begin{theorem}[Tsirelson \cite{Tsirelsonthm1}]\label{thmTsirelson}
 Let $\mathbf{A}_1,\ldots,\mathbf{A}_m$ and $\mathbf{B}_1,\ldots,\mathbf{B}_n$ be observables with eigenvalues in the interval $\De{-1,1}$ acting on $\Hi_A$ and 
 $\Hi_B$ respectively. Then, for any state $\ket{\psi}\in \Hi_A \otimes \Hi_B$, there exist unit vectors $\ket{u_1},\ldots,
 \ket{u_m},\ket{v_1},\ldots,\ket{v_n} \in \R^{m+n}$ such that
 \begin{align}
  \bra{\psi}\mathbf{A}_x \otimes \mathbf{B}_y\ket{\psi}=\braket{u_x}{v_y}.
 \end{align}
 
 Conversely, let $\DE{\ket{u_x}}_{x=1}^m$, $\DE{\ket{v_y}}_{y=1}^n \in \R^{N}$ be unit vectors. Then, for the maximally entangled state
 $\ket{\Phi^+} \in \Hi_A \otimes \Hi_B$,
 with $\Hi_A=\Hi_B=\C^d$, $d=2^{\frac{\lceil N\rceil}{2}}$, there exist  $\pm 1$-observables, $\DE{\mathbf{A}_x}$ acting on $\Hi_A$ and $\DE{\mathbf{B}_y}$ acting on $\Hi_B$,
 such that 
 \begin{align}
 \braket{u_x}{v_y}= \bra{\Phi^+}\mathbf{A}_x \otimes \mathbf{B}_y\ket{\Phi^+}.
 \end{align}
\end{theorem}

Tsirelson's Theorem establishes a one-to-one relation between quantum strategies and the inner product of real vectors.

From now on let us set $|\inp_A|=m_A$ and $|\inp_B|=m_B$. 
The optimal quantum strategy of an {\xor} game is given, in general, by Alice and Bob measuring $\pm 1$ observables $\mathbf{A}_{x}$ and $\mathbf{B}_{y}$ on a shared pure quantum state 
$\ket{\psi}$ of arbitrary dimension. 
By Tsirelson's theorem \ref{thmTsirelson} the expected value of these observables can be
replaced by the inner product $\braket{u_x}{v_y}$ of unit vectors in $\mathbb{R}^{m_A + m_B}$. This implies
that calculating the quantum value of an {\xor} game can be formulated as a 
semidefinite program $(\mathcal{P})$ \cite{StephanieSDP}.

\begin{theorem}\label{thmxorSDP}
 The optimal quantum bias, $\epsilon_q$, of an {\xor} game with game matrix $\Phi$ is given by
\begin{align}\label{eqxorSDP}
\epsilon_q =\begin{cases}         
 \max& \Tr{\Phi}_{s} \mathcal{X}	\\
\text{s.t.} &\text{diag}(\mathcal{X}) = |\1 \rangle, \\
         &\mathcal{X} \geq 0,
         \end{cases}
\end{align}
where $\text{diag}(\mathcal{X})$ is a vector whose entries are the diagonal elements of matrix $\mathcal{X}$, $|\1 \rangle$ is the all 1's vector, 
 ${\Phi}_{s} = \left(\begin{smallmatrix}
0& \frac{1}{2} {\Phi} \\ \frac{1}{2} {\Phi}^T& 0
\end{smallmatrix} \right)$ and $\mathcal{X} = \left(\begin{smallmatrix}A& S\\ S^T& B\end{smallmatrix} \right)$. $S$ is the \emph{strategy matrix} defined as $S_{x,y} = \langle u_{x} | v_{y} \rangle$. 
The matrices $A, B$ with $A_{x, x'} = \langle u_{x} | u_{x'} \rangle$ and $B_{y,y'} = \langle v_{y} | v_{y'} \rangle$ represent local terms. 
 \end{theorem}

 \begin{proof}
 Let us consider an {\xor} game with $|\inp_A|=m_A$ questions for Alice, $|\inp_B|=m_B$ questions for Bob, 
 associated game matrix $\Phi$, and winning condition determined by function $f(x,y)$.
 We have that
 \begin{align}
  \epsilon =\sum_{x,y} p(x,y)(-1)^{f(x,y)}\De{ P(a\oplus b=0|x,y)-P(a \oplus b=1|x,y)},
 \end{align}
and for a quantum strategy where Alice and Bob perform the measurements $\DE{{M}_x^0,{M}_x^1}$ and $\{{M}_y^0,{M}_y^1\}$, 
respectively, in a 
quantum state $\ket{\psi}$ 
\begin{align}
P(a\oplus b=0|x,y)-P(a \oplus b=1|x,y)=\bra{\psi}\mathbf{A}_x \otimes \mathbf{B}_y \ket{\psi},
\end{align}
where $\mathbf{A}_x$ and $\mathbf{B}_y$ are defined as in Eq. \eqref{eqdefobservable}.
Hence 
\begin{align}
 \epsilon =\sum_{x,y} p(x,y)(-1)^{f(x,y)} \bra{\psi}\oA_x \otimes \oB_y \ket{\psi}.
\end{align}

Now we can define the vectors
\begin{align}\label{vectorsTsirelson}
 \begin{split}
  \ket{u_x}&\defeq \mathbf{A}_x \otimes \I \ket{\psi},\\
  \ket{v_y}&\defeq \I \otimes \mathbf{B}_y \ket{\psi},
 \end{split}
\end{align}
so that
\begin{align}
 \epsilon=\sum_{x,y} p(x,y)(-1)^{f(x,y)}\braket{u_{x}}{v_{y}}.
\end{align}

Let
\begin{align}
S\defeq \sum_{x,y}\braket{u_{x}}{v_{y}}\ketbra{x}{y}.
\end{align}
be the \emph{strategy matrix} that represents this particular strategy.
Together with the game matrix given by Eq. \eqref{eqxormatrix} we have that
\begin{align}
 \epsilon =\Tr \Phi S^T=\Tr \Phi^T S.
\end{align}

Now let us define the matrix whose columns are composed by the vectors $\DE{\ket{u_x}}$ and  $\DE{\ket{v_y}}$:
\begin{align}
 X=\begin{pmatrix}
 & & & & & \\
    \ket{u_1}& \ldots & \ket{u_{m_A}}& \ket{v_1}& \ldots & \ket{v_{m_B}}\\
    & & & & &
   \end{pmatrix},
\end{align}
and set $ \mathcal{X}=X^{\dagger}X$. $\mathcal{X}$ is the so-called \emph{Gram matrix} of the set of vectors $\DE{\ket{u_x},\ket{v_y}}$, 
and $\mathcal{X} = \left(\begin{smallmatrix}A& S\\ S^T& B\end{smallmatrix} \right)$. Note that if the  vectors $\DE{\ket{u_x},\ket{v_y}}$ are
normalized, all the diagonal elements of $\mathcal{X}$ are equal to 1.
By defining $\tilde{\Phi}_{s} = \left(\begin{smallmatrix}
0& \frac{1}{2} {\Phi} \\ \frac{1}{2} {\Phi}^T& 0
\end{smallmatrix} \right)$ we have
\begin{align}
 \Tr\tilde{\Phi}_{s} \mathcal{X}=\frac{1}{2}\Tr \Phi S^T+\frac{1}{2}\Tr \Phi^T S= \epsilon.
\end{align}

Finally, we use the fact that 
$\mathcal{X}\geq 0$ iff it is the Gram matrix of a set of vectors \cite{matrix}, in order to 
 guarantee that, for every feasible $\mathcal{X}$ in problem $(\mathcal{P})$, there exist a set of normalized vectors $\DE{\ket{u_x}}$, 
$\DE{\ket{v_y}}$ $\in \R^{N}$, for some $N$. 
And by Tsirelson's theorem \ref{thmTsirelson}, each solution $\mathcal{X}$ can then be described by
$\pm 1$-observables 
applied to a quantum state $\ket{\psi}$. Which ends the proof.
\end{proof}

For a classical deterministic strategy, all vectors $|u_{x} \rangle$ and $|v_{y} \rangle$ are equal 
to $\pm|w\rangle$, for a single unit vector $|w\rangle$, since the 
expected values $\mean{A_xB_y}$ assume values $\pm 1$. 
The strategy matrix $S_c$ of a classical deterministic strategy is thus a matrix with $\pm 1$ entries with all columns (and rows) being proportional to each other.

\subsection*{An upper bound to the quantum value of {\xor} games}

We now show that the quantum value of an {\xor} game can be upper bounded by a quantity related to the spectral norm of the game matrix $\Phi$.

\begin{theorem}\label{thmboundxor}
Given a bipartite \xor-game, with $m_A$ inputs for Alice, $m_B$ inputs for Bob, and an associated game matrix $\Phi$, the quantum value is upper bounded by
\begin{align}\label{eqxorbound}
 \omega_{q}^{\oplus}\leq \frac{1}{2}\de{1+\sqrt{m_A m_B}\norm{ \Phi}},
\end{align}
where $\norm{ \Phi}$ is the maximum singular value, or the spectral norm, of the matrix $\Phi$.
\end{theorem}

\begin{proof}
We follow the approach of Ref. \cite{NLC}.
 Given a bipartite \xor-game, $g^{\oplus}(f,p)$, its quantum value is given by
 \begin{align}
  \omega_{q}&= \max_{\ket{\psi},\DE{\oA_x},\{\oB_y\}}  \frac{1}{2} \de{1+ \sum_{x,y} p(x,y)(-1)^{f(x,y)} \bra{\psi}\oA_x \otimes \oB_y \ket{\psi}}.
 \end{align}
 Note that we have replaced the supremum of Eq. \eqref{eqwqgeral} by a maximum, since Tsirelson's Theorem \ref{thmTsirelson}
 guarantees that
 the optimal quantum value for {\xor} games  is always achieved by a quantum strategy (and even more, that this strategy involves a finite dimensional system).
 
 Now let $\ket{\psi'},\DE{\oA'_x},\{\oB'_y\}$ be the quantum state and observables corresponding to the optimal strategy for the game $g^{\oplus}(f,p)$, and 
 let us define the unit vectors
 \begin{subequations}
 \begin{align}
  \ket{\alpha}&= \frac{1}{\sqrt{m_A}}\sum_{x=1}^{m_A} \oA'_x\otimes \I_B \otimes \I_x \ket{\psi'}\otimes \ket{x},\\
  \ket{\beta}&=\frac{1}{\sqrt{m_B}}\sum_{y=1}^{m_B} \I_A\otimes \oB'_y\otimes \I_y \ket{\psi'}\otimes \ket{y}.
  \end{align}
  \end{subequations}
  
  So that  we have
 \begin{align}
   \omega_{q}&= \frac{1}{2}\de{1+\sqrt{m_A m_B} \bra{\alpha} \I_{AB} \otimes \Phi \ket{\beta} }\nonumber\\
   &\leq \frac{1}{2}\de{1+\sqrt{m_A m_B}\,\norm{ \I_{AB} \otimes \Phi}}\\
    &=\frac{1}{2}\de{1+\sqrt{m_A m_B}\,\norm{ \Phi}}.\nonumber
 \end{align}
\end{proof}

The bound \eqref{eqxorbound} was also derived and studied in details in Ref. \cite{Epping}, where the authors provide 
necessary and sufficient conditions for the bound to be tight.
For the case of $m_A=m_B=m$, Theorem \ref{thmboundxor} can be derived
from the dual $(\mathcal{D})$ of problem \eqref{eqxorSDP}. 
One can show that $m\norm{ \Phi}$ is a feasible solution to the dual problem $(\mathcal{D})$, and hence, by the weak duality theorem \ref{thmweakdual}, it is
an upper bound on the quantum bias $\epsilon_q$ \cite{StephanieSDP}
(see Appendix \ref{A-xorgames}).

\subsection*{No-signaling value of an {\xor} game}

An \xor-game can always be won with certainty by a no-signaling strategy, \ie 
$ \omega_{NS}^{\oplus}=1$.

In order to see that, consider the strategy defined by
\begin{align}\label{eqboxwns1}
P(a,b|x,y) \defeq \begin{cases}
 \frac{1}{2} & \text{if } a \oplus b = f(x,y), 	\\
0 & \text{otherwise}.	\\                   
                  \end{cases}
\end{align}
This strategy is no-signaling, because all the marginals are uniform
\begin{align}
P(a|x) = P(b|y) = \frac{1}{2} \; \forall a,b,x,y,
\end{align}
and by construction it wins the game with certainty. The no-signaling boxes defined by Eq. \eqref{eqboxwns1} are the generalization of the PR-boxes 
(Eq. \eqref{eqprbox}), introduced in Chapter \ref{chapterNL}, for the case
of more inputs per party.

\section{Linear games}

Another class of games that we are going to consider are the called \emph{linear games}.
They are a generalization of \textsc{xor} games to a larger alphabet output size.
Linear games constitute a particular case of a more general class of nonlocal games, the \emph{unique games}. 

Unique games have
been extensively used in the 
study of hardness of approximation of some NP-complete problems, 
in attempts to identify the existence of polynomial time algorithms
to approximate the optimal solution of the problem to within a constant factor \cite{Hastad, Khot}.
Unique games is the class of nonlocal games where for each pair of questions, $(x,y)$, there is an associated 
permutation and the players win the game iff Bob's output corresponds to the permutation of Alice's output. 
Linear games constitute the particular case
of unique games where all the associated permutations commute and therefore it can be defined in terms of an Abelian group $(G,+)$.

\begin{definition}\label{defLG}
A two-player linear game $g^{\ell}(G,f,p)$ is a nonlocal game where two players Alice and Bob receive questions $x$, $y$ 
from sets $\inp_A$ 
and $\inp_B$ respectively, chosen from a probability distribution $p(x,y)$ by a referee. They reply with respective answers $a, b \in G$ where $(G,+)$ is
a finite Abelian group with associated operation $+$. The winning condition of the game is defined by a function
$f : Q_A \times Q_B \rightarrow G$, such that $V(a,b|x,y)=1$ if $a + b = f(x,y)$ and 0 otherwise.  
\end{definition}

Some concepts and formal definitions on groups are stated  in Appendix \ref{A-groups}.

The average probability of success of the players in a linear game $g^{\ell}(G,f,p)$ can be written as
\begin{align}
 \omega(g^{\ell})=\sum_{x,y}p(x,y)P(a+b=f(x,y)|x,y).
\end{align}

An {\xor} game can be seen as a linear game with  $(G,+)=\mathbb{Z}_2$. We  can also define
the generalized {\xor} games, \emph{\textsc{xor}-$d$ games}, denoted $g^{\oplus_d}$, as the 
class of linear games associated 
to the cyclic group $\mathbb{Z}_d$ (the set $[d]=\DE{0,\ldots, d-1}$ with the operation 
of addition modulo $d$). 


\subsection*{Perfect no-signaling strategy for linear games}

The existence of a perfect no-signaling strategy that wins the game with probability 1 also holds for linear games.
For every linear game $g^{\ell}(G,f,p)$
there exist a no-signaling strategy  that perfectly wins the game: 
\begin{align}  
 \omega_{NS}(\textsl{g}^{\ell}) = 1.
\end{align}
Such a strategy is defined as
\begin{align}\label{eqprboxLG}
P(a,b|x,y) = \begin{cases}
              \frac{1}{|G|} &\text{if}\; a + b = f(x,y) \\
              \; 0 &\text{otherwise} .
             \end{cases}
\end{align}
The strategy \eqref{eqprboxLG} clearly wins the game, and also it is no-signaling since
all the marginals are fully random:
\begin{align}  
P(a|x) =\sum_b P(a,b|x,y) = \sum_b \frac{1}{|G|}\delta_{b,f(x,y)-a} =\frac{1}{|G|} \;\; \forall \;y ,
\end{align}
and analogously for $P(b|y)$.

\section{\texorpdfstring{$n$}{n}-player games}\label{secnplayer}

In this Section we present the generalization of the previous definitions for games with $n$ players.

 \begin{definition}[$n$-player nonlocal Game]
An $n$-player nonlocal game $g_n(V,p)$ is a cooperative task 
where $n$ players $A_1,\ldots,A_n$,  {who} are not allowed to communicate after 
the game starts, receive respectively questions $x_1,\ldots,x_n$, where $x_i \in \inp_i$, chosen
from a probability distribution $p(x_1,\ldots,x_n)$ by a referee. 
 {Upon} receiving question $x_i$ player $A_i$ is supposed to answer $a_i \in \out_i$. 
 The winning condition of the game is defined by a payoff function $V(a_1,\ldots,a_n|x_1,\ldots,x_n)$ which
 assumes value $1$ to indicate when the players win and value $0$ to indicate when they lose. 
 \end{definition}
 
The probability of success of the players for a particular strategy defined by the box $\vec{P}(a_1,\ldots,a_n|x_1,\ldots,x_n)$ 
is given by
 \begin{align}
  \omega(g_n)= \sum_{\vec{a}\in {\out},\vec{x}\in {\inp}} p(\vec{x}) V(\vec{a}|\vec{x}) P(\vec{a}|\vec{x}),
 \end{align}
 where $\vec{x}=(x_1,\ldots,x_n)$ denotes the input string, ${\inp}=\inp_1\times \ldots \times \inp_n$  and analogously for $\vec{a}$ and $\out$.

The definition of linear games can be straightforwardly generalized for an $n$-player game in the following way:

\begin{definition}\label{defLGn}
An $n$-player linear game $g_n^{\ell}(G,f,p)$ is a 
nonlocal game where the players answer with 
$a_1,\ldots,a_n \in G$, where $(G,+)$ is a finite Abelian group with associated operation $+$, and the predicate function $V$ only 
depends on the sum of the players outputs:
\begin{align}
V(\vec{a}|\vec{x})=\begin{cases}
                    1 \;,\;\text{if}\; a_1 + \ldots + a_n=f(x_1,\ldots,x_n)\\
                    0 \;,\; \text{otherwise}
                        \end{cases}
\end{align}
for $f:{\inp}\longrightarrow G$.
\end{definition}

The probability of success of a particular strategy in a $n$-player linear game $g_n^{\ell}$ is given by
 \begin{align}\label{eqwLG}
 \omega(g_n^{\ell})=\sum_{\vec{x}\in {Q}}p(\vec{x})P(a_1+ \ldots + a_n=f(\vec{x})|\vec{x}).
\end{align}

\subsection*{The classical value of an $n$-player game}

 The classical value of an $n$-player nonlocal game $\omega_c(g_n)$ is obtained by an optimization over deterministic local 
 strategies
 \begin{align}
\omega_c(g_n)=\max_{\DE{\vec{D}(a_i|x_i)}}\sum_{\vec{a},\vec{x}} p(\vec{x}) V(\vec{a}|\vec{x}) D(a_1|x_1)\ldots D(a_n|x_n)
 \end{align}
where $\vec{D}(a_i|x_i)$ represents a deterministic probability distribution.

The number of possible deterministic strategies also increases exponentially with the number of parties. For a game with $n$ players, $|\inp|$ questions
per player and $|\out|$ outputs per question, the number of deterministic strategies is $|\out|^{n \times |\inp|}$.

\subsection*{The quantum value of an $n$-player game}

A general $n$-partite quantum strategy can be described by the players sharing an $n$-partite pure state $\ket{\psi}$ of arbitrary dimension and
 performing local measurements $\DE{M_{x_i}^{a_i}}$ on it. The quantum value of the $n$-player game  $g_n$ is then given by
  \begin{align}
\omega_q(g_n)=\sup_{\ket{\psi},\{M_{x_i}^{a_i}\}}\,\sum_{\vec{a},\vec{x}} p(\vec{x}) V(\vec{a}|\vec{x}) \bra{\psi}M_{x_1}^{a_1} \otimes \ldots \otimes M_{x_n}^{a_n}\ket{\psi}.
 \end{align}

 The case of multiplayer games is even more challenging than the bipartite case. 
For bipartite XOR games, we have seen  that Tsirelson's theorem (Theorem \ref{thmTsirelson}) 
guarantees that the best performance of quantum players can be calculated exactly and efficiently using a
semidefinite program, Eq.  \eqref{eqxorSDP}. For three players, even for the case of {\xor} games the quantum value is known
to be NP-hard\footnote{In terms of the number of possible questions.} to approximate \cite{Vidick3xorhard}.
\vspace{1em}

In chapter \ref{chapterxor} we deal with bipartite {\xor} games and present results of Ref. \cite{GMxorShannon} concerning games with no-quantum advantage and
results on the Shannon capacity of {\xor}-game graphs. In Chapter \ref{chaptergamesd} we provide an upper bound to the quantum value of 
a 2-player linear game and present results of Ref. \cite{GMxord}. In Chapter 
 \ref{chapternplayer} we generalize the results for $n$-player linear games presenting the results of Ref. \cite{GMmultiplayer}.

%% file: graph.tex
\chapter{Graph theoretic approach to nonlocality}\label{chaptergraphNL}

The graph theoretic approach to quantum correlations was introduced by Cabello, Severini and Winter \cite{csw0,csw} in the framework  
of contextuality scenarios\footnote{In 
a contextuality scenario we do not necessarily have parties. The hypothesis in question is whether an observable $\mathbf{O}$ that belongs to two
different sets of mutually commuting observables, called the \emph{contexts}, $\mathbf{O}\in\mathcal{C}_1$ 
and $\mathbf{O}\in\mathcal{C}_2$, can have a description independent of the context. In a Bell scenario the commutation of observables that form each context is
guaranteed by the tensor product structure of the parties' experiments.}, 
where Bell scenarios can be seen as a particular case of it.
Further refinements to the particular case of nonlocality scenarios where made in Refs. \cite{BelenCombinatorial, xorgraphs, multilovasz}.
In this chapter, we will see how the classical, quantum, and no-signaling values of a Bell expression can be associated to graph parameters.

\section{A bit of zero-error information theory}\label{seczeroerror}

In 1956 Shannon \cite{Shannon0Error} studied the concept of zero-error capacity of a communication channel.
The zero-error capacity $C_0$ is defined as the maximum rate at which it is possible to transmit 
information with zero probability of error through a channel $\mathcal{C}$. 

A channel is described by a set of input letters $i \in \mathcal{I}$, a set of output letters $o \in \mathcal{O}$, and the transition
probabilities $p_i(o)$ that represents the probability that an input $i$ will generate an output $o$. 
In the analysis of zero-error capacity we are not interested in the particular values of the transition probabilities but only whether they are zero or
not. In order to deal with the important properties of the zero-error capacity $C_0$, we can associate a graph $\mathcal{G}$  to the channel
in consideration.

\begin{definition}[Confusability Graph]
 The \textit{confusability graph} associated to a channel $\mathcal{C}$ is a graph $\mathcal{G}(V,E)$ 
 where each input letter of the channel corresponds to a vertex
 \begin{align}
  V=\DE{i \,|\, i \in \mathcal{I}},
 \end{align}
 and there is an edge connecting two vertices $i$ and $j$ if these inputs can be confused by the channel 
 \begin{align}
 \DE{i,j}\in E \;\; \text{iff} \;\;\exists \; o \in \mathcal{O}\;\;\text{s.t.}\;\; p_i(o) \neq 0 \;\; \text{and}\;\;p_j(o) \neq 0.
 \end{align}
\end{definition}

An alternative quantity that carries all the information contained in the graph $\mathcal{G}$ is the \emph{adjacency matrix}.

\begin{definition}[Adjacency Matrix]
 The \textit{adjacency matrix} $[\mathcal{A}_{ij}]$ associated to a graph $\mathcal{G}(V,E)$  is defined as
 \begin{align}
  \mathcal{A}_{i,j}=\begin{cases}
           1\;\;& \text{if } \DE{i,j}\in E ,\; \\
           0\;\;& \text{otherwise}.
          \end{cases}
 \end{align}
\end{definition}

The maximal number of 1-letter messages, $M_0(1)$, which can be sent through a channel $\mathcal{C}$ without confusion can be extracted 
directly from a graph invariant. A bit of thinking\footnote{Or maybe a few bits. But note that since the independence number picks the maximum number 
of vertices such that no two of which are adjacent, this represents the maximum number of letters such that none of them can be confused.} leads us to conclude that $M_0(1)$ is
equal to the independence number of the confusability graph $\mathcal{G}$:
\begin{align}
 M_0(1)=\alpha(\mathcal{G}).
\end{align}
The independence number of a graph is the cardinality of the maximal independent set, which is a subset of vertices such that none of which are adjacent .
It is described by
 the optimization problem \eqref{eqalpha}.

Formally, the zero-error capacity $C_0$ is defined as
\begin{align}\label{eqC0}
 C_0=\sup_{n} \frac{1}{n} \log M_0(n),
\end{align}
where $M_0(n)$ is the largest number of $n$-letter messages that can be sent through the channel without confusion.
The number of non-confusable $n$-letter messages is given by the independence number of the graph $\mathcal{G}^n$:
\begin{align}
 M_0(n)=\alpha(\gr^n),
\end{align}
where $\gr^n$ denotes the \emph{strong product} of graph $\gr$ with itself $n$ times.


\begin{definition}\label{defstrongprod}
 The strong product $\gr \boxtimes \mathcal{F}$ of graphs $\gr$ and $\mathcal{F}$ is such that
 \begin{itemize}
  \item $V(\gr \boxtimes \mathcal{F})= V(\gr)\times V(\mathcal{F})$,
  \item $\DE{(u,u'),(v,v')} \in E(\gr \boxtimes \mathcal{F}) \Leftrightarrow \DE{u,v} \in E(\gr)$ and $\DE{u',v'}\in E(\mathcal{F})$.
 \end{itemize}
\end{definition}

In terms of a channel $\mathcal{C}$, Definition \ref{defstrongprod} captures the idea that two 2-letter words are confusable if they are confusable in the 
first and in the second letter.

So we have that
\begin{align}\label{C0}
 C_0=\sup_{n} \frac{1}{n} \log \alpha(\mathcal{G}^n)\defeq\; \log \Theta(\gr),
\end{align}
where
\begin{align}\label{eqTheta}
 \Theta(\gr)=\sup_{n} \sqrt[n]{\alpha(\mathcal{G}^n)}.
 \end{align}
From now on we refer to $\Theta(\gr)$ as the \emph{Shannon capacity} of graph\footnote{The supremum   in Eq. \eqref{eqC0} can actually be replaced
by a limit, since the strong product satisfies
\begin{align}
 \alpha(\mathcal{G}^2)\geq  \alpha(\mathcal{G})^2,
\end{align}
which can be proved by noting that given an independent set $A$ of graph $\mathcal{G}$, one can
generate an independent set for $\mathcal{G}^2$
by taking $(i,j) \st \;\;i,j\in A$.
} $\mathcal{G}$.

As an example let us consider the channel with confusability graph given by the pentagon $C_5$:

\begin{center}
 \includegraphics[scale=0.4]{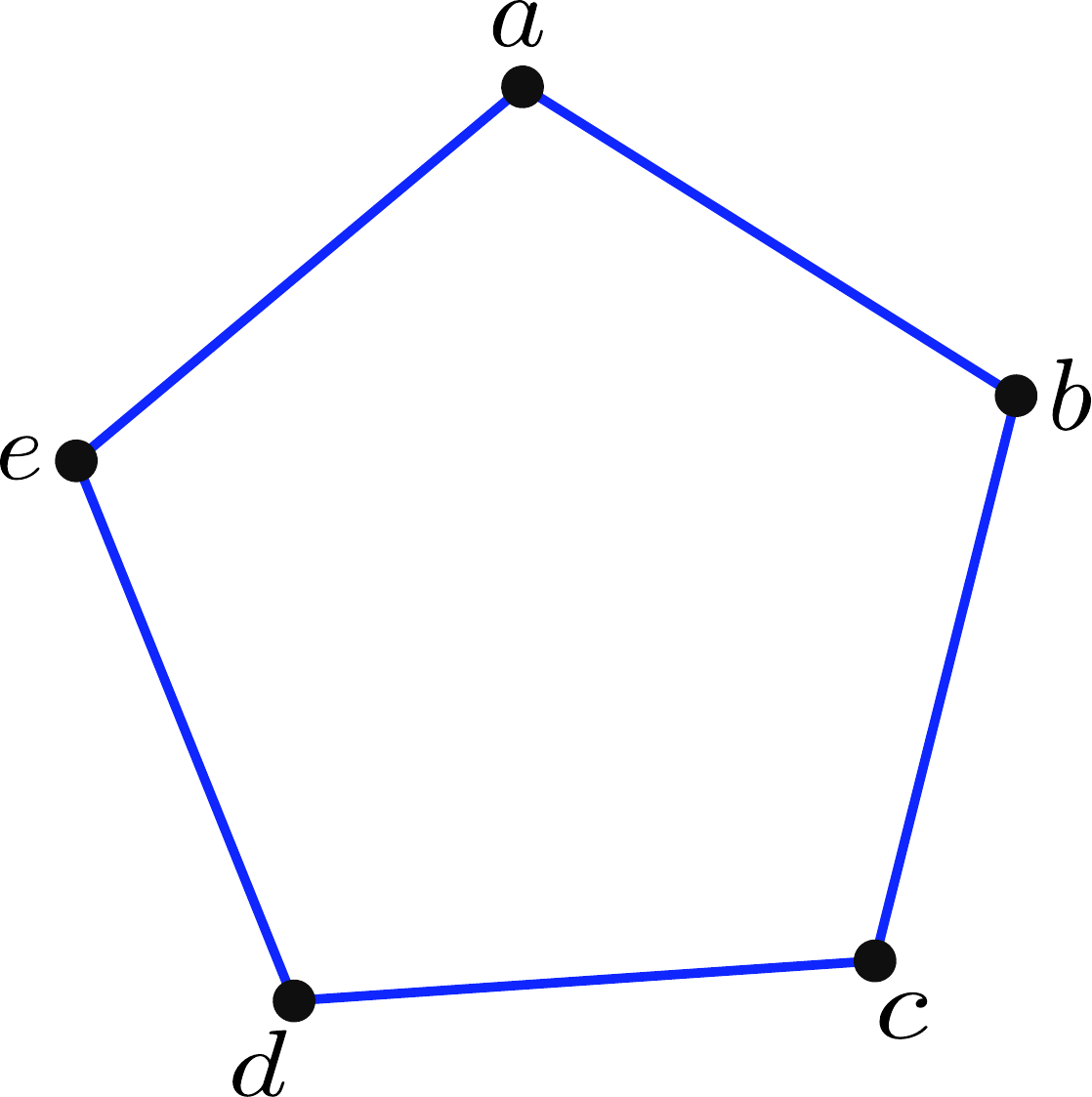}
\end{center}

The maximum number of 1-letter messages we could send through this channel without confusion is 2, and we have $\alpha(C_5)=2$.
We could for example send letter  $a$ or $c$. 
Now, considering
 2-letter messages, we could obviously make the four non-confusable words: $aa, ac, ca$, and $cc$. 
 However one can check that the following five words of size two:
$aa,bc,ce,db,ed$, are such that no two of them can be confused, and actually $\alpha({C_5}^2)=5$. 
Therefore, we conclude that $\Theta(C_5) \geq \sqrt{5}$.

In contrast to the ordinary channel capacity where a vanishing probability of error is allowed \cite{Shannon1948} and the capacity can be efficiently
calculate, it is not know whether the 
zero-error capacity is computable. By the definition given in Eq. \eqref{eqTheta},  we have a supremum over $n$ and at the moment there 
is no algorithm to decide the value of $\Theta$ in a finite number of steps for an arbitrary graph.

In 1979 Lovász \cite{Lovasznumber} introduced an efficiently computable upper bound to the Shannon capacity $\Theta$ (and therefore also for the 
independence number $\alpha$). The Lovász number $\vartheta$ is a graph invariant
defined as follows:

\begin{definition}[Lovász Number]
 Consider an $n$-vertex graph $\gr$. Let the set of vectors $\de{\ket{v_1},\ldots,\ket{v_n}}$ be an orthonormal representation  of $\bar{\gr}$ 
 and $\ket{\psi}$ be a unit vector. The Lovász number is given by
 \begin{align}\label{eqlovasz}
  \vartheta(\gr)=\max_{\DE{\ket{v_i}},\ket{\psi}} \sum_{i=1}^n \de{\braket{\psi}{v_i}}^2.
 \end{align}
\end{definition}

An orthonormal representation of the complementary graph $\bar{\gr}$  is a set of real vectors $\de{\ket{v_1},\ldots,\ket{v_n}}$
associated to the vertices of graph $\gr$ such that $\braket{v_i}{v_i}=1\; \forall\, i$ and
$\DE{i,j}\in E \Rightarrow \braket{v_i}{v_j}=0$.

By calculating $\vartheta(C_5)$, 
Lovász proved \cite{Lovasznumber} that the previously known lower bound for the 
Shannon capacity of the pentagon, $\sqrt{5}$,
is actually the exact value: $\Theta(C_5) =\sqrt{5}$.

There are many equivalent characterizations of the Lovász number \cite{LovaszNotes,KnuthlovaszS}, here we consider one
 that will be useful in Chapter
\ref{chapterxor}:

\begin{theorem}[\cite{Lovasznumber}]\label{thmlovaszA}
 Given a graph $\gr$, consider the family of $n \times n$ symmetric matrices $A$ such that
 \begin{align}\label{eqmatrixA}
  A_{ij}=1 \;\; \text{if}\;\; i=j \;\; \text{or}\;\; \DE{i,j} \notin E.
 \end{align}
 The Lovász number is the minimum of the maximum eigenvalue of such matrices: 
 \begin{align}\label{eqlovaszA}
  \vartheta(\gr)=\min_{A}\lambda_{\max}(A).
 \end{align}
\end{theorem}

Theorem \ref{thmlovaszA} can be formulated as the semidefinite program $(P)$:
\begin{align}\label{eqlovaszsdp}
 (P)=\begin{cases}
      \min &t\\
      \st& t\,\I-A\geq 0\\
      & A_{ij}=1 \;\;\text{if} \;\;i=j\;\; \text{or}\;\; \DE{i,j} \notin E\\
       \end{cases}
\end{align}
where $t\,\I-A\geq 0$ imposes that $t$ is greater than or equal to the maximum eigenvalue of a symmetric matrix $A$ satisfying \eqref{eqmatrixA}, and the
minimization picks the
smallest of the maximum eigenvalues. Characterization \eqref{eqlovaszsdp} makes it clear (see Section \ref{secopt}) 
that the Lovász number can be efficiently determined.

\section{The exclusivity graph}

The graph theoretic approach to contextuality and nonlocality, that we call CSW approach,  was first introduced in Ref.\cite{csw0} (and 
published later in Ref. \cite{csw}).
The idea is to associate with a given Bell scenario a graph called the \emph{exclusivity graph}. The exclusivity graph has the vertices labeled by the 
events\footnote{By event $(a,b|x,y)$ we denote the occurrence of Alice measuring $x$ and obtaining outcome $a$, and Bob measuring $y$ and 
obtaining outcome $b$.} of the particular Bell scenario $(a,b|x,y)$, and the edges connect vertices which correspond to exclusive events 
(\textit{e.g.}, Alice measuring $x$ 
and obtaining $a=0$ is exclusive with any event where Alice measures the same $x$ and obtains $a=1$). The idea of the CSW approach is to take into account that
the probabilities of exclusive events should sum up to no more than one.

\begin{definition}[Exclusivity graph]
 The exclusivity graph associated to a Bell scenario is a graph $\gr$ with vertices representing the possible events: $(a,b|x,y)$, and 
 adjacencies defined by
 \begin{align}
  (a,b|x,y)\sim (a',b'|x',y')\Leftrightarrow (x=x'\;\wedge\; a\neq a') \;\vee\;(y=y'\; \wedge\; b \neq b'),
 \end{align}
 where $\wedge$ and $\vee$ denotes respectively the logical operations AND and OR.
\end{definition}

Given a Bell scenario with boxes $\vec{P}(a,b|x,y)$, we denote a Bell expression by $\mathcal{S}=\DE{s^{a,b}_{x,y}}$ such that its value for a 
particular box is given by $\mathcal{S}(\vec{P})=\sum_{a,b,x,y}s^{a,b}_{x,y}P(a,b|x,y)$. Note that every Bell expression  can
be written with non-negative coefficients only\footnote{One just has to use the normalization condition of probability 
distributions, $\sum_{a,b}P(a,b|x,y)=1$ $\forall x,y$, in order to avoid the negative coefficients.}, $s^{a,b}_{x,y}\geq0$.

Now, given a Bell expression with positive coefficients $\mathcal{S}$, we can consider the \emph{weighted exclusivity graph} 
 $\gr(\mathcal{S})$, where the weight $s^{a,b}_{x,y}$  is attributed to each vertex 
$(a,b|x,y)$ of the exclusivity graph $\gr$. What is shown in the CSW approach \cite{csw} is that the local,
quantum and no-signaling values of a Bell expression $\mathcal{S}$, respectively  $\mathcal{S}_c$, $\mathcal{S}_q$, and $\mathcal{S}_{NS}$, are 
related with invariants of the associated weighted exclusivity graph.

\begin{theorem}[\cite{csw}]\label{thmcsw}
Given a Bell expression $\mathcal{S}$, and the weighted exclusivity graph $\gr(\mathcal{S})$, it holds
that
\begin{align}
 \alpha(\gr(\mathcal{S}))=\mathcal{S}_c \leq \mathcal{S}_q \leq \vartheta(\gr(\mathcal{S}))\leq \mathcal{S}_{NS}=\alpha^*(\gr(\mathcal{S})),
\end{align}
where $ \alpha(\gr(\mathcal{S}))$, $\vartheta(\gr(\mathcal{S}))$, and $\alpha^*(\gr(\mathcal{S}))$ are respectively the weighted independence number,
the weighted Lovász number, and the weighted fractional packing number.
\end{theorem}

The graph invariants that appear in Theorem \ref{thmcsw} are defined bellow.
\begin{definition}
 The weighted version of the independence number $\alpha(\gr(\mathcal{S}))$, the Lovász number $\vartheta(\gr(\mathcal{S}))$, and the 
 fractional packing number $\alpha^*(\gr(\mathcal{S}))$
 of the weighted graph $\gr(\mathcal{S})$ are defined by the following optimization problems:
 
 \begin{align}\label{eqalphaS}
 \alpha(\gr(\mathcal{S}))=\begin{cases}
          \max &\sum_{i\in V} s_i \omega_i\\
          \st & \omega_i+\omega_j\leq 1 \;\text{if}\; \DE{i,j}\in E\\
          & \omega_i \in \DE{0,1}\; \forall i \in V.
         \end{cases}
\end{align}
 \begin{align}\label{eqlovaszS}
 \vartheta(\gr(\mathcal{S}))=\begin{cases}
          \max &\sum_{i\in V} s_i \de{\braket{\psi}{v_i}}^2\\
          \st & \braket{v_i}{v_j}=0 \;\text{if}\; \DE{i,j}\in E\\
          & \braket{v_i}{v_i}=1 \; \forall i \in V\\
          &\braket{\psi}{\psi}=1.
         \end{cases}
\end{align}
\begin{align}\label{eqfracS}
 \alpha^*(\gr(\mathcal{S}))=\begin{cases}
      \max &\sum_{i\in V}s_i\omega_i\\
      \st&  \omega_i \geq 0 \; \forall i \in V\\
      & \sum_{i\in K} \omega_i \leq 1\;\; \forall \; \text{clique} \; K\;\text{of}\; \gr(\mathcal{S}).
       \end{cases}
\end{align}
\end{definition}

In order to understand Theorem \ref{thmcsw}, note that the weighted independence number \eqref{eqalphaS} captures the idea that a deterministic probability $\omega_i$ 
is associated to vertex $i$, and therefore $\alpha(\gr(\mathcal{S}))$
gives the classical value of a Bell inequality $\mathcal{S}$. Concerning the Lovász number \eqref{eqlovaszS}, 
the term ${\braket{\psi}{v_i}}^2$ can be interpreted as the probability resultant of a
projective measurement (determined by the vector $\ket{v_i}$) in the quantum state $\ket{\psi}$. However note that no tensor product structure was imposed
for these measurements and state, which justifies the upper bound $\mathcal{S}_q \leq \vartheta(\gr(\mathcal{S}))$.
For the fractional packing number \eqref{eqfracS}, the constraint that $\sum_{i\in K} \omega_i \leq 1$ for all clique\footnote{A
clique in the graph $\gr$ is a subset of vertices $K\subseteq V$ such that every two vertices in $K$ are adjacent.} $K$ captures the 
no-signaling conditions \eqref{eqNS} and normalization of the probabilities ($\sum_{a,b}P(a,b|x,y)=1 \forall (x,y)$), and therefore it justifies
$\mathcal{S}_{NS}=\alpha^*(\gr(\mathcal{S}))$.

The weighted Lovász number \eqref{eqlovaszS} was introduced and studied in Ref. \cite{KnuthlovaszS}, and it also admits an SDP characterization.
The weighted independence number \eqref{eqalphaS} is a $\DE{0,1}$-integer program which is in general NP-complete, and the weighted 
fractional packing number\footnote{The fractional packing number is easy to determine once the cliques of the graph are known, 
however if this is not the case, the problem of finding the cliques of a graph  is also NP-complete.} \eqref{eqfracS} is defined by an LP.

The CSW approach had great impact in the field of nonlocality and contextuality, allowing us to use techniques of
combinatorics to derive many results in these fields (see for example \cite{barbaragraph,xorgraphs} and references therein).
The CSW approach was further refined in Ref. \cite{BelenCombinatorial} where the authors consider hypergraphs from which it is possible 
to extract information 
about the complete nonlocality/contextuality scenario, and not only about the Bell inequalities.

For contextuality it holds that given an exclusivity graph $\gr$ there 
always exists a contextuality scenario and an inequality whose quantum value 
is equal to the Lovász number of the graph $\gr$. This can be seen by the definition \eqref{eqlovasz}, where one can interpret the vectors $\ket{v_i}$
as projective measurements being performed on the quantum state $\ket{\psi}$. However for nonlocality, equality does not hold in general
since the tensor product structure is not captured by the Lovász number.
In Ref. \cite{multilovasz} the authors introduce the \emph{exclusivity multigraph} with 
the aim to capture the structure 
presented in the nonlocality scenario.
In the exclusivity multigraph, exclusive events are connected by colored edges, which make explicit if the exclusivity is due to
Alice's measurement or due to Bob's measurement. They defined the \emph{multigraph Lovász number} $\theta(\gr(\mathcal{S}))$ which satisfies
$\mathcal{S}_q=\theta(\gr(\mathcal{S}))$ for some Bell inequality $\mathcal{S}$. Therefore, $\theta(\gr(\mathcal{S}))$ is
a quantity that may be uncomputable, but 
an hierarchy of SDPs 
(based on the NPA hierarchy \cite{NPA}) can be used to derive upper bounds on  $\theta(\gr(\mathcal{S}))$.

%% file: xorbound.tex
\part{Results}


\chapter{{\xor} games with no-quantum advantage and the Shannon capacity of graphs}\label{chapterxor}
\chaptermark{{\xor} games and Shannon capacity of graphs}

In this chapter we present results of Ref. \cite{GMxorShannon}:
\begin{center}
\begin{minipage}{12.5cm}
\noindent\textit{Characterizing the Performance of {\xor} Games and the Shannon Capacity of Graphs}\\
   R. Ramanathan, A. Kay, \textbf{G. Murta} and P. Horodecki\\
 \href{http://link.aps.org/doi/10.1103/PhysRevLett.113.240401}{\textbf{Phys. Rev. Lett., \textbf{113}, 240401, (2014)}}.
\end{minipage}
\end{center}
Our main result is to use insights from the field of nonlocality in order to derive a result
in classical information theory. More specifically:
We study the performance of quantum players in two-player {\xor} games, and we derive a 
set of necessary and sufficient conditions such that quantum players cannot perform any better than classical players. 
We then consider the exclusivity graph associated to an {\xor} game and examine
its Shannon capacity. This allows us to specify new families of graphs for which the Shannon capacity can be determined. 

\section{Motivation}

In Chapter \ref{chaptergraphNL} we have seen that, according to the CSW approach \cite{csw}, we can associate a weighted graph $\gr(\mathcal{S})$
to a Bell inequality $\mathcal{S}= \DE{s^{a,b}_{x,y}}$ such that the classical, 
quantum, and no-signaling values of $\mathcal{S}$ are related to graph invariants.

For an {\xor} game with $m$ questions per player, uniformly chosen by the referee, 
$g^{\oplus}(f,\frac{1}{m^2})$, we
can associate the non-weighted graph $\gr$ which only contains the vertices representing events satisfying
the winning condition ($(a,b|x,y)$ s.t. $a\oplus b=f(x,y)$), and we have the relation
\begin{align}\label{eqrelgame}
m^2\omega_c(g^{\oplus})=\alpha(\gr)\leq m^2\omega_q(g^{\oplus}) \leq\vartheta(\gr).
\end{align}
At the same time we have seen that the confusability graph $\gr$ associated to a classical channel $\mathcal{C}$ satisfies
\begin{align}\label{eqrelgraph}
 \alpha(\gr) \leq \Theta(\gr) \leq \vartheta(\gr).
\end{align}

In what concerns the Shannon capacity of a graph $\gr$, very few classes of graphs are known for which $\Theta(\gr)$ has been
established analytically. Among these classes are:
\begin{itemize}
 \item perfect graphs\footnote{A graph is perfect if the chromatic number of every induced subgraph equals the 
 size of the largest clique of that subgraph.}:
 have $ \alpha(\gr) =\alpha^*(\gr)$ (see \cite{Cabelloperfectgraph}).
 \item Kneser graphs\footnote{Kneser graphs $KG_{n,k}$ are graphs whose vertices correspond to $k$-element 
 subsets of a set of $n$ elements 
 ($S_i\subset\DE{1,\ldots,n}$, $|S_i|=k \;\forall i$), and vertices $i$ and $j$ are adjacent
 $i\sim j \Leftrightarrow S_i\cap S_j = \emptyset$.}$KG_{n,k}$:
 satisfy $\alpha(KG_{n,k})=\vartheta(KG_{n,k})$ \cite{Lovasznumber}.
 \item vertex-transitive self-complementary graphs: satisfy $\Theta(\gr)=\vartheta(\gr)$  \cite{Lovasznumber}.
 \item K\"onig-Egerv\'ary graphs\footnote{A graph $\gr$ is a K\"onig-Egerv\'ary graph if it satisfies $\alpha(\gr)+\nu(\gr)=|V|$, where
 $\nu(\gr)$ is the maximum size of a matching (for the definition of a matching, see footnote \ref{ftmatching}).}:  satisfy $\alpha(\gr)=\alpha^*(\gr)$ \cite{KEgraphs}.
\end{itemize}


In Section \ref{secresultsxor} we are going to discuss a bit more about some properties of these families of graphs. However we are not going to give any
detailed presentation, but rather state only the properties that will be useful for our 
discussion. For the reader interested in learning more about graphs we refer to Ref. \cite{Diestelgraph}.

Note that, except for the vertex-transitive self-complementary graphs, all the classes satisfy $\alpha(\gr)=\vartheta(\gr)$.
And actually this constitutes a simple way to determine the Shannon capacity of a graph: to prove that $\alpha(\gr)=\vartheta(\gr)$.
From Eq. \eqref{eqrelgame}, we see 
that for an {\xor}-game exclusivity graph $\gr$ this is only possible if 
$\omega_c(g^{\oplus})=\omega_q(g^{\oplus})$. Our goal in this chapter is to characterize {\xor} games with no-quantum advantage, 
$\omega_c(g^{\oplus})=\omega_q(g^{\oplus})$, and, with their corresponding graphs in hand,  to study the value
of $\vartheta(\gr)$. 

\section{{\sc xor} games and their graphs} 

Let us consider an {\sc xor} game $g^{\oplus}$ where each player receives one among $m$ possible questions, $|\inp_A|=|\inp_B|=m$, chosen by 
the referee with probability $p(x,y)$. The 
associated game matrix \eqref{eqxormatrix} is given by 
\begin{equation}
\Phi=\sum_{x,y\in[m]}(-1)^{f(x,y)}p(x,y)\ketbra{x}{y},
\end{equation}
where $a\oplus b=f(x,y)$ is the winning condition of the game, and $[m]$ denotes the set $\DE{1,\ldots,m}$.

If the referee choses questions with uniform distribution, $p(x,y)=\frac{1}{m^2}$, we can simply consider the non-normalized game matrix:
\begin{equation}\label{eqphiequi}
\tilde{\Phi} = \sum_{x,y \in [m]} \Phi_{xy}\ketbra{x}{y},
\end{equation}
where $\Phi_{xy} \defeq (-1)^{f(x,y)}$.


Following the CSW approach \cite{csw}, every {\sc xor} game has an associated graph $\gr$ \cite{xorgraphs,csw,BelenCombinatorial}, where the 
vertices are the events $(a,b|x,y)$ that satisfy the winning condition\footnote{Note that we could include in the graph all the events
$(a,b|x,y)$ and
just give a weight zero to the ones that do not satisfy $a\oplus b=f(x,y)$.} of the game (\ie such that $a\oplus b=f(x,y)$),
and the edges are determined by
\begin{align}
 (a,b|x,y)\sim(a',b'|x',y') \Leftrightarrow \de{x=x' \wedge a \neq a'} \vee \de{y=y' \wedge b \neq b'}.
\end{align}

Now, by making use of the winning relation $(-1)^{a\oplus b}=\Phi_{xy}$, we can pa\-ram\-e\-ter\-ize the vertices of an 
{\xor}-game graph with only
three parameters $(x,y,a)$. 
\begin{definition}\label{defxorgraph}
The graph $\mathcal{G}(\Phi)$ associated to the {\sc xor} game with game matrix $\Phi$ consists of $2m^2$ vertices, which 
can be labeled as $(x,y,a)$ where $x,y\in\DE{1,\ldots,m}$ and $a\in\{0,1\}$. Two vertices $(x,y,a),(x',y',a')\in V$ form 
an edge of the graph iff
\begin{align}
(x=x' \wedge a\neq a') \;\vee\; (y=y' \wedge (-1)^{a\oplus a'} \neq \Phi_{xy}\Phi_{x'y'}).
\end{align}
\end{definition}

Some properties of the {\xor}-game graphs are given in the following proposition.

\begin{proposition}
An {\xor} game-graph $\mathcal{G}(\Phi)$  has the following properties:
\begin{itemize}
 \item $\mathcal{G}(\Phi)$  is $(2m-1)$-regular\footnote{A 
 graph $\gr$ is said to be \emph{$k$-regular} if every vertex has degree $k$, \ie every vertex is connected to $k$
 other vertices.}.
 \item  $\mathcal{G}(\Phi)$ is triangle-free\footnote{A triangle is a set of three vertices such that all of them are connected. 
 A graph is 
 triangle-free if it has no subset of vertices forming a triangle.}.
\item $\mathcal{G}(\Phi)$  has perfect matching\footnote{A \emph{matching} in $\gr$ is a set of edges such that no two of them have a vertex in common. A 
perfect matching is a matching that covers all the vertices of the graph.\label{ftmatching}}.
\end{itemize}
Moreover, the adjacency matrix of $\mathcal{G}(\Phi)$ can be expressed as
\begin{align} \label{eqadjmatrix}
\mathcal{A}(\mathcal{G}(\Phi)) =& \I_m \otimes(\proj{\1}-\I_m)\otimes \sigma_X+\half\proj{\mathbf{1}}\otimes\I_m\otimes(\I_2+\sigma_X)	\\
&-\half [D(\proj{\1}\otimes\I_m)D]\otimes(\I_2-\sigma_X)\nonumber
\end{align}
where $\sigma_X$ is the Pauli-$X$ matrix, $\ket{\1}$ is the all-ones vector $\ket{\1}=\sum_{x\in[m]}\ket{x}$, and the matrix $D$ is defined as
\begin{equation}
D=\sum_{x,y\in[m]}\Phi_{xy}\proj{x,y}.
\end{equation}
\end{proposition}

 In order to see that an {\xor}-game graph $\mathcal{G}(\Phi)$  is $(2m-1)$-regular note that every vertex has $m$ neighbors due to exclusivity with 
 respect to Alice's
 measurement, and $m$ neighbors due to exclusivity with respect to Bob's measurement, 
 where one of these neighbors is exclusive due to both. Therefore, each 
 vertex has $2m-1$ neighbors.
 
By taking two adjacent vertices $(x,y,a)\sim(x',y',a')$, a straightforward analysis shows that if there exists a third vertex $(x'',y'',a'')$ which
 is adjacent to $(x,y,a)$ and to $(x',y',a')$ that would lead to a contradiction, hence 
 there can be no triangle in $\mathcal{G}(\Phi)$.
 
 A perfect matching for $\mathcal{G}(\Phi)$ is obtained by taking the edges corresponding to the same inputs for both parties $\DE{(x,y,0),(x,y,1)} \; \forall x,y$.

Concerning the adjacency matrix \eqref{eqadjmatrix}, the first term, $\I_m \otimes(\proj{\1}-\I_m)\otimes \sigma_X$,
accounts for the edges representing exclusivity
only due to Alice's measurement since it corresponds to
\begin{align}
 \sum_{x,y,y',a}\ketbra{xya}{xy'(a\oplus 1)} -\sum_{x,y,a}\ketbra{xya}{xy(a\oplus1)}.
\end{align}
The remaining terms, $\half\proj{\1}\otimes\I_m\otimes(\I_2+\sigma_X)-\half [D(\proj{\1}\otimes\I_m)D]\otimes(\I_2-\sigma_X)$, 
 accounts for exclusivity due to Bob and due to both, and corresponds to:
 \begin{align}
  \sum_{x,x',y,a,a'}\frac{1}{2}\de{1-(-1)^{a\oplus a'}\Phi_{xy}\Phi_{x'y}}\ketbra{xya}{x'ya'}.
 \end{align}

Finally we present the spectrum of the adjacency matrix, which will be very important in the proof of the main result of this chapter.

\begin{theorem}\label{thmAeigen}
The adjacency matrix of an {\xor}-game graph $\mathcal{G}(\Phi)$, $\mathcal{A}(\mathcal{G}(\Phi))$, has the following spectrum
and corresponding degeneracies:
\begin{equation}
\label{adj-spec}
\text{spec}(\mathcal{A}(\mathcal{G}(\Phi)))=\begin{cases}
2m-1 & \times  1 	\\
m-1 & \times 2m-2	\\
-1 & \times (m-1)^2 \\
1-m\pm \lambda_z & \times 1 \\
1 &  \times m(m-2) 
\end{cases}.
\end{equation}
where $\lambda_z$ denotes the $m$ singular values of $\tilde{\Phi}$. 
 \end{theorem}

 The proof of Theorem \ref{thmAeigen} is presented in Appendix \ref{A-proofs}.

\section{No-quantum advantage} \label{secnoqadvantage}

We have seen in Chapter \ref{chapterNL} that quantum mechanics can lead to the violation of Bell inequalities,
contradicting the hypothesis of local realism  (Definition \ref{deflocalrealism}). This fact opens the possibility to 
explore quantum systems in many tasks, going beyond what can be achieved with classical systems.
In the framework of nonlocal games, quantum theory can make the players to succeed with higher probability.
However, it is not only the tasks in which quantum theory brings advantage that are of interest. The tasks in which
quantum systems can perform
no better than classical systems also tell us something about nature. Most of the proposed principles to explain the set of quantum correlations were based on 
tasks where quantum theory brings no advantage, as for example the principle of \emph{Information Causality} \cite{IC} and the principle of 
\emph{Local Orthogonality} \cite{LO} (which was based on the guess your neighbor's input (GYNI) multiplayer game \cite{GYNI}).
Another task which brings no quantum advantage, and also constitutes a class of {\xor} games, is 
the \emph{Nonlocal Computation} (NLC) task investigated in Ref. \cite{NLC}.

\begin{definition}[Nonlocal computation]\label{defNLC}
Consider that the referee picks a string of $n$ bits $\vec{z}=(z_1,\ldots,z_n)$ with an arbitrary probability distribution $p(\vec{z})$, and for each
bit $z_i$ he chooses randomly $x_i$ and $y_i$ such that $z_i=x_i\oplus y_i$. Then the  referee gives the input bit string $\vec{x}=(x_1,\ldots, x_n)$ for Alice
and $\vec{y}=(y_1,\ldots, y_n)$ for Bob. Upon receiving their inputs, Alice and Bob give respective binary outputs $a$ and $b$. Their goal
is to satisfy the winning condition
\begin{align}
 a\oplus b=f(\vec{z})=f(x_1 \oplus y_1, \ldots, x_n \oplus y_n)
\end{align}
for an arbitrary function $f(\vec{z})$.
\end{definition}

The constraint that the function $f$ depends only on the bits $z_i$, which are distributed between Alice and Bob by the 
relation $x_i\oplus y_i =z_i$, gives to this game the interpretation of
a distributed computation of function $f$ (nonlocal computation). Moreover, given a string $\vec{z}$, the referee
chooses with uniform distribution, among the $x_i$'s and $y_i$'s that satisfy $z_i=x_i \oplus y_i$, \ie
\begin{align}
 p(\vec{x},\vec{y})= \frac{1}{2^n} p(\vec{z}).
\end{align}
Therefore, even upon receiving their respective inputs Alice and
Bob get no information about the string $\vec{z}$.

In Ref. \cite{NLC}, the authors have shown that no matter with which probability distribution the referee chooses among strings $\vec{z}$,
the players sharing quantum resources cannot perform any better in the computation of function $f(\vec{z})$ than players with
only classical resources. Interestingly enough, as this is an {\xor} game, there always exist a no-signaling strategy which can 
compute $f(\vec{z})$ perfectly.
Given the lack of advantage even with the freedom of the referee's choice, and by the fact that more general probabilistic theories
can perform this task perfectly, the \emph{no-advantage for nonlocal computation} was also considered one of the principles
to distinguish the quantum set of correlations.

It was known that the nonlocal computation inequalities, for 2-bit and 3-bit input strings, do not constitute facets of 
the local polytope \cite{GYNI}, and recently it was proved that this is also
the case for any number of inputs \cite{GMfaces}\footnote{This result
was derived before the conclusion of this Thesis, but the full work was completed and 
published only a few months later. However I am allowing myself to violate causality and include
the journal reference here.}. Therefore, all the non-local computation inequalities are faces of lower dimension.
A facet defining Bell inequality with no-quantum advantage would also be a facet of the  set of quantum correlations.
It is an open problem whether there exist such facets in the set of bipartite quantum correlations (for multipartite scenarios
the GYNI game constitutes
facets of $\mathcal{Q}$ \cite{GYNI}).

\section{Results}\label{secresultsxor}

We now present the main results derived in Ref. \cite{GMxorShannon}.

\subsection*{{\xor} games with no quantum advantage}

We are interested in characterizing {\sc xor} games which have the property of no quantum advantage.
In the following we present a necessary and sufficient condition for $\omega_c = \omega_q$, and a very simple sufficient 
condition to guarantee no-quantum advantage.


%
\begin{theorem}\label{thmxoriff}
Consider a two-party {\sc xor} game with game matrix ${\Phi}$ (with no all-zero row or column)
for which $S_c=\ket{s^A}\bra{s^B}$ represents the optimal classical strategy matrix. 
Let $\Sigma = \text{Diag}(\{\bra{i}{\Phi}\ket{s^B}\braket{s^A}{i}\}_{i=1}^m)$ 
and $\Lambda=\text{Diag}(\{\braket{i}{s^B}\bra{s^A}{\Phi}\ket{i}\}_{i=1}^m)$.
There is no quantum advantage for ${\Phi}$ if and only if $\Sigma, \Lambda > 0$ and 
\begin{equation}
\label{nec-suff-1}
\rho(\Lambda^{-1} {\Phi}^T \Sigma^{-1} {\Phi}) = 1,
\end{equation}
where $\rho(.)$ denotes the spectral radius\footnote{The spectral radius is the maximum eigenvalue in modulus of a matrix.}. 
\end{theorem} 

\begin{proof}
We have seen in Chapter \ref{chaptergames} that the quantum bias $\epsilon_q$ of an {\xor} game can be 
calculated by the semidefinite program $(\mathcal{P})$  \eqref{eqxorSDP}:
\begin{align}
\epsilon_q =\begin{cases}         
 \max& \; \Tr {\Phi}_{s} \mathcal{X}	 \\
\text{s.t.} &\text{diag}(\mathcal{X}) = |\1 \rangle \oplus |\1 \rangle, \\
         &\mathcal{X} \geq 0.
         \end{cases}
\end{align}
where $\text{diag}(\mathcal{X})$ is a vector whose entries are the diagonal elements of matrix $\mathcal{X}$, 
 ${\Phi}_{s} = \left(\begin{smallmatrix}
0& \frac{1}{2} {\Phi} \\ \frac{1}{2} {\Phi}^T& 0
\end{smallmatrix} \right)$ and $\mathcal{X} = \left(\begin{smallmatrix}A& S\\ S^T& B\end{smallmatrix} \right)$. $S$ is the strategy matrix, $S_{x,y} = \langle u_{x} | v_{y} \rangle$, and  
$A, B$ are local terms, $A_{x, x'} = \langle u_{x} | u_{x'} \rangle$ and $B_{y,y'} = \langle v_{y} | v_{y'} \rangle$. 

By the Lagrange duality theory,
the bias can be bounded from above by a feasible solution of the Lagrange dual problem $(\mathcal{D})$. The application 
of Lagrange duality, presented in Section \ref{secduality}, leads us to: 
\begin{align}\label{eqxordual}
(\mathcal{D})\begin{cases}         
 \min \;& \sum_{i=1}^{2m}y_i \\
\text{s.t.} &\text{Diag}(y) \geq  {\Phi}_s.
         \end{cases}
\end{align}
where the $y_i$ are $2m$ variables and Diag($y$) denotes the diagonal matrix with entries $y_i$. 

Problem $(\mathcal{P})$ satisfies strong duality (Theorem \ref{thmstrongdual}) since $\mathcal{X}=\I$ 
is a strictly feasible point and, therefore, Slater's conditions \eqref{eqslater} are satisfied. 
As such, we need to derive the conditions under which the solution of $(\mathcal{D})$ \eqref{eqxordual} achieves the classical value
$\bra{s^B} {\Phi}\ket{s^A}$, which may also be written as $\bra{s^s} \Phi_s\ket{s^s}$,
for $\ket{s^s}=\ket{s^A}\oplus\ket{s^B}$ being the direct sum of vectors $\ket{s^A}$ and $\ket{s^B}$.
So we require that
\begin{equation}
\Tr\left((\text{Diag}(y)-\Phi_s)\proj{s^s}\right)=0.
\end{equation}
By the semi-definite condition $\text{Diag}(y) -  {\Phi}_s\geq 0$, this means that $\ket{s^s}$ is an eigenvector, with zero eigenvalue,
of $\text{Diag}(y)- \Phi_s$:
\begin{equation}\label{eq:equality}
\text{Diag}(y)\ket{s^s}=\Phi_s\ket{s^s}.
\end{equation}
Now, since  $\ket{s^s}$ is a vector of $\pm 1$ entries (as it corresponds to a classical deterministic strategy), we have that
\begin{align}
\bra{i} \text{Diag}(y)\ket{s^s}\braket{s^s}{i}=y_i (\De{s^s}_i)^2=y_i.
\end{align}
Therefore, an element by element comparison in Eq. \eqref{eq:equality} gives us
\begin{align}
 \bra{i} \text{Diag}(y)\ket{s^s}\braket{s^s}{i}= \bra{i}\Phi_s\ket{s^s}\braket{s^s}{i},
\end{align}
from which we can derive that whenever a classical strategy achieves the optimal quantum value, we have:
\begin{align}
\text{Diag}(y)=\left(\begin{matrix}
\half\Sigma &0\\ 0& \half\Lambda
\end{matrix} \right).
\end{align}

The constraint $\text{Diag}(y) \geq  {\Phi}_s$ can be rewritten as
 $\left(\begin{smallmatrix}
\Sigma & -{\Phi}\\ - {\Phi}^T& \Lambda
\end{smallmatrix} \right) \geq 0$. And since $\Phi$ has no all-zero row or column, this condition is satisfied 
only if $\Sigma,\Lambda>0$ (see observation 7.1.10 in Ref. \cite{matrix}).
Theorem 7.7.9 in Ref. \cite{matrix} 
states that in these conditions
\begin{align}\label{eqiifmatrix}
 \left(\begin{matrix}
\Sigma & -{\Phi}\\ - {\Phi}^T& \Lambda
\end{matrix} \right) \geq 0 \Leftrightarrow \rho( {\Phi}^T \Sigma^{-1}  {\Phi} \Lambda^{-1}) \leq 1.
\end{align}
Finally, when the optimal solution of the dual problem is given by the classical strategy, the inequality on the LHS of Eq. \eqref{eqiifmatrix} is saturated. Therefore, condition
$\rho( {\Phi}^T \Sigma^{-1}  {\Phi} \Lambda^{-1}) \leq 1$ can be replaced by equality.
\end{proof}

When $S_c = S_c^T$ and $ {\Phi} = {\Phi}^T$, the condition of Theorem \ref{thmxoriff} 
reduces to $\Sigma > 0$ and $\rho(\Sigma^{-1} {\Phi}) = 1$.

\begin{corollary} \label{corxorif}
If the (non-normalized) singular vectors corresponding to the maximum singular value of $ {\Phi}$ can be written with entries $\pm 1$, 
then there is no quantum advantage 
for players of the game $ {\Phi}$.
\end{corollary}
\begin{proof}
Let the (unnormalised) maximum singular vectors with $\pm1$ elements be $\ket{\lambda^A}$ and $\ket{\lambda^B}$ such 
that $\bra{\lambda^A} {\Phi}\ket{\lambda^B}=\lambda m$, where $\lambda$ is the maximum singular value of $\Phi$ and $m$ is the square of the norm of 
$\ket{\lambda ^A}$. In this case, $S_c=\ket{\lambda^A}\bra{\lambda^B}$. 
Then $ {\Phi}S_c^T=\lambda\proj{\lambda^A}$ such that $\Sigma=\Lambda=\lambda\I$. Evidently, these are positive and
\begin{equation}
\rho( {\Phi}^T \Sigma^{-1}  {\Phi} \Lambda^{-1})=\frac{1}{\lambda^2}\rho( {\Phi}^T {\Phi})=1.
\end{equation}
\end{proof} 
This is only a sufficient condition and not a necessary one. For example, the maximum eigenvector of
\begin{equation}	\label{eqn:ex}
 {\Phi}_{ex} = \frac{1}{16} \left(\begin{matrix}
1 & -1 & -1 & 1 \\ -1 & -1 & 1 & -1 \\ -1 & 1 & -1 & -1 \\ 1 & -1 & -1 & 1
\end{matrix} \right)
\end{equation} 
does not consist of $\pm 1$ elements, and yet it can be verified  that $\omega_q( {\Phi}_{ex}) = \omega_{c}( {\Phi}_{ex})=\frac{7}{8}$.

\subsection*{Examples of no-quantum advantage}

Corollary \ref{corxorif} gives us a simple way to construct games with no quantum advantage. 
In the uniform probability case, it suffices to construct any symmetric matrix $\Phi\in\{\pm 1\}^{m\times m}$ for which the total 
of every row is the same, and at least $\half m$, which ensures that the all 1's vector $\ket{\1}$ is a maximum eigenvector. 

The NLC game (Definition \ref{defNLC}) is a trivial example of Corollary \ref{corxorif} because the associated game matrix $ {\Phi_{NLC}}$ is diagonal
in the Hadamard basis for any number of input bits \cite{NLC}. 
So now we address the question posed in Ref. \cite{NLC}: \textit{`Finding families of {\sc xor} games that differ from NLC, but with no 
quantum advantage'}.
Let us consider an anti-circulant matrix\footnote{An anti-circulant matrix is a matrix where each row has the same elements, but shifted one position
to the left, with respect to the previous row.} $\Phi$, then for any $m$, $\ket{\1}$ is an eigenvector, and if $m$ is
even, so is the alternating signs vector. All we have to do is to restrict the matrix elements to guarantee that one of these two eigenvectors
yields the eigenvalue of maximum modulus. For $m=4$, let 
$(\gamma_0,\gamma_1,\gamma_2,\gamma_3)$, subject to $\sum_i|\gamma_i|=\frac{1}{4}$, specifies the first row of $\Phi$. If
\begin{equation}
\max\left(\left(\sum_{i=0}^3\gamma_i\right)^2,\left(\sum_{i=0}^3\gamma_i(-1)^i\right)^2\right)\geq (\gamma_0-\gamma_2)^2+(\gamma_1-\gamma_3)^2,
\end{equation}
we have a game for which there is no quantum advantage. A sufficient condition for this to happen is that
\begin{align}\label{eqcondgamma}
\gamma_0\gamma_2+\gamma_1\gamma_3\geq 0.
\end{align}

Many different patterns for the probability distribution satisfy  condition \eqref{eqcondgamma}, as for example
\begin{equation}\label{eqexcor}
\Phi=\left(\begin{array}{cccc}
p & q & q & -p \\
q & q & -p & p \\
q & -p & p & q \\
-p & p & q & q
\end{array}\right),
\end{equation}
where $|p|+|q|=1/8$ and $p,q\in\mathbb{R}$. Which leads to matrices that (except for the uniformly distributed case) are 
not diagonal in the Hadamard basis.


In order to construct a new families of games with no quantum advantage, we can base on the following observation: 
\begin{enumerate}[(\textasteriskcentered)]
\item  If $ \Phi^1$ and $ \Phi^2$ are two game matrices 
satisfying Corollary \ref{corxorif}, then it follows that $ \Phi^1\otimes \Phi^2$ also satisfies Corollary \ref{corxorif}.
\end{enumerate}
This observation allows us to extend any examples for small input size to examples with an arbitrarily large number of inputs. 

Actually observations (\textasteriskcentered) can be extended for any two games satisfying Theorem \ref{thmxoriff} by 
the additivity property of the quantum value
of these games.
In Ref. \cite{parallelrep},
it was shown that the quantum value of {\xor} games satisfy additivity: Consider two {\xor} games $g_1^{\oplus}(f_1,p_1)$
and  $g_2^{\oplus}(f_2,p_2)$, the sum of these 
games is the {\xor} game 
\begin{align}  \label{eqsumgame}
 g_{1+2}^{\oplus}(f_1\oplus f_2, p_1p_2),
\end{align}
 where the referee picks questions $((x_1,x_2),(y_1,y_2))\in (\inp_{A_1} \times \inp_{A_2})\times (\inp_{B_1} \times \inp_{B_2})$
with probability $p_1(x_1,y_1)p_2(x_2,y_2)$ and the winning condition of the game is defined by
\begin{align}
 a\oplus b =f_1(x_1,y_1)\oplus f_2(x_2,y_2).
\end{align}
The game is said to be additive if the optimal strategy is achieved when the players play each game individually and Alice outputs bit $a=a_1\oplus a_2$ 
and Bob outputs bit $b=b_1\oplus b_2$. The classical value does not satisfy additivity in general, however the quantum value does \cite{parallelrep}.
And the game matrix of game $g_{1+2}^{\oplus}(f_1\oplus f_2, p_1p_2)$ is the tensor product of the individual game matrices 
$\Phi(g_1^{\oplus})\otimes \Phi(g_2^{\oplus})$.

\subsection*{Shannon capacity of game graphs} 
Game graphs for which $\omega_c=\omega_q$ are good candidates 
to have a Shannon capacity $\Theta(\gr)=\alpha(\gr)$. We now prove that it is the case for games $g^{\oplus}(f,\frac{1}{m^2})$ 
satisfying the hypothesis of Corollary \ref{corxorif}.

\begin{theorem}\label{thmxorshannon}
Every two-party {\sc xor} game with $m$ uniformly chosen inputs for each party and satisfying the hypothesis of Corollary \ref{corxorif} has a 
game graph for which $\Theta(\gr) = \alpha(\gr)$.
\end{theorem}
\begin{proof}
To establish the Shannon capacity, our strategy is to find both $\alpha(\gr)$ and $\vartheta(\gr)$, and to show that they are equal. 
$\alpha(\gr)$ is straightforward, it coincides
with the optimal strategy, and by Corollary \ref{corxorif} it is specified by the maximum singular value\footnote{In terms
of the normalized game matrix $\Phi$ we have $\omega_c=\frac{1}{2}\de{1+m\norm{\Phi}}$. Note, however, that we are using the 
unnormalized game matrix $\tilde{\Phi}$.} of $\tilde{\Phi}$:
\begin{equation}
\alpha(\gr)=m^2\omega_c=\half m(m+\norm{\tilde{\Phi}}).
\end{equation}

In order to compute $\vartheta(\gr)$, we use the characterization of the Lov\'asz number given by 
Theorem \ref{thmlovaszA} \cite{Lovasznumber},
in which $\vartheta(\gr)$ is upper bounded by the 
largest eigenvalue of any symmetric 
matrix $(A_{i,j})_{i,j=1}^{N}$  such that
 \begin{align}\label{eqcondlovasz}
  A_{ij}=1 \;\; \text{if}\;\; i=j \;\; \text{or}\;\; \DE{i,j} \notin E.
 \end{align}
Our goal is to find a symmetric matrix $A$ whose maximum eigenvalue matches $\alpha(\gr)$ and
since $\alpha(\gr)\leq\Theta(\gr)\leq \vartheta(\gr)$ we would finish the proof. 

Consider the matrix
\begin{equation}
A \defeq \proj{\1}\otimes\proj{\1}\otimes(\I+\sigma_X)+a \mathcal{A}(\gr)+b\I \otimes\I \otimes \sigma_X.
\end{equation}
The matrix $A$ satisfies the conditions \eqref{eqcondlovasz}, and moreover, all of the three terms in $A$
commute with each other. Therefore, the diagonalization is the same as for the adjacency matrix $\mathcal{A}(\gr)$,
and the eigenvalues can be readily obtained: (see proof of Theorem \ref{thmAeigen} in Appendix \ref{A-proofs}):
\begin{equation}
\text{spec}(\mathcal{A}(\mathcal{G}(\Phi)))=\begin{cases}
2m^2+a(2m-1)+b & \times  1 	\\
2m+a(m-1)+b & \times 2m-2	\\
-a+b & \times (m-1)^2 \\
a(1-m\pm \lambda_z)-b & \times 1 \\
a-b &  \times m(m-2) 
\end{cases}.
\end{equation}

It is now our task to 
select $a,b$ such that the largest eigenvalue is  $\alpha(\gr)$. If we set $a=-m$ and $b=\alpha(\gr)-m$, this 
yields a maximum eigenvalue equal to $\alpha(\gr)$. We conclude then that $\alpha(\gr)=\vartheta(\gr)$, and therefore $\Theta(\gr)=\alpha(\gr)$.
\end{proof}
This proof automatically covers all NLC games, but also includes many other XOR games (Eq. \eqref{eqexcor}, for example).
A further consequence that follows from the proof of Theorem \ref{thmxorshannon} is that whenever
$\omega_q=\half (1+ \frac{1}{m}\|\Phi\|)$, we know that $m^2 \omega_q=\vartheta(\gr)$. The CHSH game \cite{CHSH} is an example of this.

\subsection*{Previously known families of graphs with $\Theta(\gr)=\alpha(\gr)$}

We now come back to the previously known families of graphs for which the Shannon capacity is equal to the independence number,
and we compare it with the graphs specified by Theorem \ref{thmxorshannon}:
\begin{itemize}
\item For \emph{Perfect graphs}, it was shown in Ref. \cite{Cabelloperfectgraph} that 
the classical value coincides with the no-signalling value of any Bell expression associated to a perfect graph, this implies $\omega_c=1$
for {\xor} games associated to perfect graphs. 
\item For \textbf{Kneser graphs} on $n$ vertices   it holds that \cite{Lovasznumber}
\begin{equation}
\vartheta(\gr) = \frac{-n \lambda_{\min}}{\lambda_{\max} - \lambda_{\min}} 
\end{equation}
where $\lambda_{\max/\min}$ are the corresponding maximum and minimum eigenvalues of the adjacency matrix. 
From Eq. \eqref{adj-spec}, we have that
\begin{align}
\lambda_{\min}=1-m-\norm{\tilde{\Phi}} \;\;\; \text{and}\;\;\; \lambda_{\max}=2m-1,
\end{align}
and from the hypothesis of Corollary \ref{corxorif} $\vartheta(\gr)=\frac{1}{2}m\de{m+\norm{\tilde{\Phi}}}$.
So we have that in this case $\norm{\tilde{ \Phi}}= m$, which again implies $\omega_c = 1$.
\item \textbf{K\"onig-Egerv\'ary (KE) family}  satisfies $\alpha(\gr) + \nu(\gr)=|V|$.
Since the {\sc xor} game graphs have a perfect matching, $\nu(\gr)=m^2$, they can only belong to
the KE family in the case of $\alpha(\gr)=m^2 \Rightarrow\omega_c = 1$.
  \end{itemize}

So we see that the previously known families of graphs for which $\alpha(\gr)=\vartheta(\gr)$ could only
correspond to an {\xor} game 
in the trivial case where $\omega_c=1$.
Therefore the non-trivial {\xor} games  satisfying Corollary \ref{corxorif} define  new families of graphs for which $\alpha(\gr)=\vartheta(\gr)$,
the NLC class \cite{NLC} being a remarkable example.

\section{Discussion and open problems} 

We have presented a necessary and sufficient condition for a bipartite {\sc xor} game 
to have no quantum advantage. And with this characterization we could single out 
new families of games with no 
quantum advantage,
even when the referee has some freedom in the probability of choosing the inputs. 
However, the generated examples (based on condition \eqref{eqcondgamma} and observation (\textasteriskcentered))  
rely on ensuring that the optimal 
classical strategy coincides with the vectors of the maximum singular value of the game matrix. It would be interesting to know
whether there exists any such classes which do not require this condition.

We have also shown that for games satisfying the condition of Corollary \ref{corxorif}, the associated graphs have $\Theta(\gr)=\alpha(\gr)$.
This is an entirely classical result derived as a consequence of 
insights provided from the study of quantum nonlocality. 
The proof of this result required Corollary \ref{corxorif} to hold. However, we believe that this restriction can be dropped, 
and that a necessary and  sufficient condition for a game graph to have $\alpha(\gr)=\vartheta(\gr)$ is given by Theorem \ref{thmxoriff}.
So far, the proof of this statement remains an open point.

A more challenging goal would be to prove similar results concerning no quantum advantage and the Shannon capacity of the 
corresponding graphs for games with more outputs and more players. However, the difficulty in 
extending the techniques used here relies mainly on the fact that there is no equivalent of Tsirelson's Theorem \ref{thmTsirelson} for these
more general scenarios.
In Ref. \cite{GMxord} we were able to find a generalization of the principle of no-advantage for nonlocal computation for a class 
of functions with  $d$ possible values. This result will be presented in the next Chapter.
 
 Finally, we want to remark that relations \eqref{eqrelgame} and \eqref{eqrelgraph} show that the quantum value of a Bell inequality and 
 the Shannon capacity of the corresponding exclusivity graph are both bounded by the same graph parameters ($\alpha$ and $\vartheta$). Moreover, as
 we have discussed previously, both quantities fall into the class of problems not known to be computable. We point out
 as an interesting question whether there is a fundamental relation between the quantum value of a Bell inequality and the Shannon capacity of the 
 corresponding graph. In the affirmative case, the physical interpretation of this relation can shed light on foundational aspects of quantum 
 theory.

%% file: xordbound.tex
\chapter{Linear games and the task of nonlocal computation}\label{chaptergamesd}
\chaptermark{Linear games and the task of NLC}

In this chapter we present results of Ref. \cite{GMxord}:
\begin{center}
\begin{minipage}{12.5cm}
 \textit{Generalized {\xor} games with $d$ outcomes and the task of nonlocal computation}\\
  R. Ramanathan, R. Augusiak, and \textbf{G. Murta}\\
 \href{http://link.aps.org/doi/10.1103/PhysRevA.93.022333}{\textbf{Phys. Rev. A, 92, 022333 (2016)}}.
\end{minipage}
\end{center}
We now focus on bipartite linear games, $g^{\ell}(G,f,p)$. Our main result is to prove an efficiently computable upper 
bound to the quantum value of a linear game.
We then use the bound to derive, in a straightforward way, an upper bound to the quantum value of the CHSH-$d$ game. 
We also show that boxes that would lead to trivialization of communication complexity are not realized in quantum theory.
As the main application of the bound, we prove a larger alphabet generalization of the principle of no-advantage for nonlocal
computation discussed in Section \ref{secnoqadvantage}. 

\section{Motivation}

As we have discussed in Chapters \ref{chaptercomput} and \ref{chaptergames}, it is not known, in general, whether the quantum value 
of a Bell inequality is computable or not (since there is a priori no restriction on the dimension of the Hilbert space 
for the quantum states and measurements). 
Only in some instances it is possible to compute the value efficiently, as it is the case for 2-player {\xor} games. 
Apart from this, typically
the NPA hierarchy \cite{NPA} is
used to get upper bounds to the quantum value. However, the quality of the approximation achieved by these bounds is not known in general, 
and the size of these programs increases exponentially with the 
level of the hierarchy. 


In Ref. \cite{UGeasy} it was shown that the quantum value of unique games can be efficiently approximated. Formally, the authors present 
an efficient algorithm such that, given a unique game with quantum value $\omega_q(g)=1-\epsilon$,
it outputs a specific
quantum strategy that achieves a value $\omega'(g)\geq1-6\epsilon$.
This remarkable result was very useful to prove parallel repetition for the quantum value of this class of games \cite{UGeasy}. 
However, this  approximation is only good 
when the quantum value of the game is close to unity, which is not the case for simple examples like the CHSH-$d$ game that we 
are going to discuss here.

Recently, there has been an increasing interest in developing applications 
of higher-dimensional entanglement (see, for example, \cite{BT00, angularmomentEntang, qutritcrypto, WL06, Qudit-randomness, q16dit-key-dist}) for which 
Bell inequalities with more than two 
outcomes are naturally suited. 
Moreover, in Ref. \cite{MatthiasAdan}  it was shown that, in general, $d$-chotomic measurements cannot be explained
as a classical selection of intrinsically 
 dichotomic measurements.
Therefore, both for fundamental reasons as well as for these applications, the study of Bell
 inequalities with more outcomes
is crucial.

\section{An efficiently computable bound to the quantum value of linear games}

Let us recall that a linear game $g^{\ell}(G,f,p)$ is defined by Alice and Bob answering with elements of an
Abelian group $(G,+)$, $a,b \in G$, and the winning condition depending 
only on the group operation $+$ of their outputs (Definition \ref{defLG}).
In order to obtain the main result of this Chapter, which is an upper bound to the quantum value of a linear 
game, we are going to make use of the 
 Fourier transform on finite Abelian groups \cite{FourierTerras} and introduce the \emph{generalized correlators}.

\begin{definition}[Generalized correlators]\label{defFourier}
Let $a,b \in G$ be elements of a finite Abelian group $(G,+)$, where $+$ is the associated group operation. 
The generalized correlators $\langle A_{x}^{i} B_{y}^{j} \rangle$ are defined via the Fourier transform of the probabilities $P(a,b|x,y)$ as
 \begin{align}
\langle A_{x}^{i} B_{y}^{j} \rangle = \sum_{a,b \in G} \bar{\chi}_{i}(a) \bar{\chi}_j(b) P(a,b|x,y),
\end{align}  
where ${\chi_i}$ are the characters of the Abelian group $(G,+)$ and  $\bar{\chi}_{i}$ is the conjugate character.
\end{definition}

The characters of an Abelian group are complex numbers which satisfy the following relations (see Appendix \ref{A-groups}):
\begin{align}\label{eqcharacprop}
\begin{cases}
\text{Homomorphism:} & \chi_i(a+b)=\chi_i(a)\chi_i(b) \; \; \forall a,b \in G\\
\text{Reflexivity:} & \bar{\chi}_i(a)=\chi_i(-a)\\
\text{Orthogonality:} & \sum_{a \in G}\chi_i(a)\bar{\chi}_j(a)=|G|\,\delta_{i,j}
\end{cases}.
\end{align}
More details on groups and the Fourier transform on Abelian groups are presented in Appendix \ref{A-groups} (see also Ref. \cite{FourierTerras}).

Given Definition \ref{defFourier}, the probabilities are recovered by the inverse Fourier transform:
\begin{equation}
P(a,b|x,y) = \frac{1}{|G|^2} \sum_{i,j \in G} \chi_{i}(a) \chi_{j}(b) \langle A_{x}^{i} B_{y}^{j} \rangle.
\end{equation}
And in terms of the correlators, normalization is expressed as
\begin{align}
\langle A_{x}^{e} B_{y}^{e} \rangle = 1 \;\; \forall \; (x,y) \in \inp_A \times \inp_B.
\end{align}

The one-party correlators $\langle A_{x}^{i} \rangle$ are defined as
\begin{equation}
\langle A_{x}^{i} \rangle \defeq \langle A_{x}^{i} B_{y}^{e} \rangle = \sum_{a,b \in G} \bar{\chi}_{i}(a) \bar{\chi}_{e}(b) P(a,b|x,y) = \sum_{a \in G}\bar{\chi}_{i}(a) P(a|x),
\end{equation}
where $e$ denotes the identity element of the group with $\chi_{e}$ being the trivial character ($\chi_{e}(b) = 1 \; \forall b \in G$) 
and we have used the no-signaling condition $\sum_{b \in G}P(a,b | x,y) = P(a|x)$. An analogous expression holds for 
$\langle B_{y}^{j}\rangle$.

In order to determine the success probability in a linear game, we are only interested in terms of the
form $P(a+b=f(x,y)|x,y)$. Consequently, we can use the characters' properties
in order to get a very simplified expression.
\begin{lemma}\label{lemmacorre}
 The average probability of success for a particular box $\boxp$, in a linear game $g^{\ell}(G,f,p)$, can be written as
 \begin{align}\label{eqwcorrelators}
  \omega(g^{\ell})=\sum_{x,y} p(x,y)\frac{1}{|G|}\de{1+ \sum_{k \in G \setminus \DE{e}} \chi_{k}(f(x,y)) \langle A_{x}^{k} B_{y}^{k} \rangle}. 
 \end{align}
\end{lemma}

\begin{proof}
 In order to prove the lemma we start by evaluating the probability:
 \begin{align}\label{eqcharacters}
 \begin{split}
P(a + b = f(x,y) | x,y) &=\sum_a P(a, f(x,y)-a|x,y)\\
			&=\sum_a \frac{1}{|G|^2} \sum_{i,j \in G} \chi_{i}(a) \chi_{j}(f(x,y)-a) \langle A_{x}^{i} B_{y}^{j}\rangle \\
			&=\frac{1}{|G|^2} \sum_{i,j \in G} \chi_{j}(f(x,y)) \de{\sum_a \chi_{i}(a) \chi_{j}(-a)} \langle A_{x}^{i} B_{y}^{j}\rangle \\
			&=\frac{1}{|G|} \sum_{j \in G} \chi_{j}(f(x,y)) \langle A_{x}^{j} B_{y}^{j}\rangle 
			\end{split}
			\end{align}
where we have used the characters' properties \eqref{eqcharacprop}.

Finally, taking the average sum over the inputs $x$ and $y$, and using normalization ($\langle A_{x}^{e} B_{y}^{e} \rangle = 1$),
we have the desired result.
\end{proof}

Lemma \ref{lemmacorre} inspire us to define a set of $|G|-1$ game matrices associated to the linear game $g^{\ell}(G,f,p)$ which carry
information about the probability distribution of the inputs and the winning condition of the game. 
These matrices are the analogue of the {\xor}
game matrices, Eq. \eqref{eqxormatrix},
to the case of linear games.

\begin{definition}[Linear game matrices]
Given a linear game $g^{\ell}(G,f,p)$, the associated game matrices are
\begin{align}
\Phi_k =\sum_{x,y} p(x,y) \chi_k(f(x,y)) \ketbra{x}{y} \;\; \text{for}\;\; k\in G\setminus \DE{e},
\end{align}
where $\DE{\ket{x}}$ and $\DE{\ket{y}}$ are orthonormal basis in $\C^{|\inp_A|}$ and $\C^{|\inp_B|}$ respectively. 
\end{definition}

Now, let us analyze the meaning of the generalized correlators for a quantum strategy. In a quantum strategy, local projective measurements 
$\DE{M_x^a}$ and $\{M_y^b\}$ are performed by each player in a shared quantum state $\ket{\psi}$. 
Now if we define the (in general non-Hermitian) `observables'
\begin{align}
 A_x^i=\sum_{a\in G}\bar{\chi}_i(a)M_x^a \;\; \text{and} \;\;  B_y^j=\sum_{b\in G}\bar{\chi}_j(b)M_y^b,
\end{align}
we have that the generalized correlators correspond to: 
\begin{align}
\langle A_{x}^{i} B_{y}^{j} \rangle = \bra{\psi} A_{x}^{i} \otimes B_{y}^{j} \ket{\psi}.
\end{align}
And therefore, the average success probability of a quantum strategy, specified by $\DE{\DE{M_x^a},\{M_y^b\},\ket{\psi}}$, is given by
 \begin{align}\label{eqwquantumLG}
  \omega(g^{\ell})=\frac{1}{|G|}\de{1+\sum_{x,y} p(x,y) \sum_{k \in G\setminus\DE{e}} \chi_{k}(f(x,y))\bra{\psi} A_{x}^{k}\otimes B_{y}^{k} \ket{\psi}}. 
 \end{align}

 Now we are ready to state the main result of this Chapter.
 
 \begin{theorem}\label{thmLGwqbound}
\label{norm-bound}
The quantum value of a linear game $g^{\ell}(G,f,p)$, with input sets $\inp_A$ and $\inp_B$,  can be bounded as 
\begin{align}\label{eqLGbound-2}
\omega_{q}(g^{\ell}) \leq \frac{1}{|G|} \de{ 1 + \sqrt{|\inp_A| |\inp_B|} \sum_{k \in G\setminus \{e\}} \norm{\Phi_{k}} },
\end{align}
where $\Phi_{k}$ are the game matrices, $\chi_{k}$ are the characters of the group $(G,+)$, and $\norm{\Phi_k}$ denotes the maximum
singular value of matrix $\Phi_k$ (the spectral norm).
\end{theorem}

\begin{proof}
In order to derive the upper bound to $\omega_q(\textsl{g}^{\ell})$, let us consider a 
quantum strategy
given by the measurements $\DE{M_x^a}$ and $\{M_y^b\}$, from which we derive the observables $A_{x}^{i}$ and $B_{y}^{j}$,
being applied to  the pure state $\ket{\psi}$. Now, let us define the unit vectors:
\begin{align}\label{vecqstrategy}
\begin{split}
\ket{\alpha_{k}}&= \sum_{x \in \inp_A} \frac{1}{\sqrt{|\inp_A|}}\de{{A_{x}^{k}}^{\dagger} \otimes \I_B \otimes \I_x} \ket{\psi}\otimes \ket{x} ,  \\
\ket{\beta_{k}}&= \sum_{y \in \inp_B} \frac{1}{\sqrt{|\inp_B|}}\de{\I_A \otimes B_{y}^{k} \otimes \I_y} \ket{\psi}\otimes \ket{y} . 
\end{split}
\end{align}
By substituting into Eq. \eqref{eqwquantumLG} we have 
\begin{align}
\begin{split}
\omega_{q}(\textsl{g}^{\ell}) &= \sup_{\DE{\ket{\alpha_k}},\DE{\ket{\beta_k}}} \frac{1}{|G|} \de{1 + \sqrt{|\inp_A||\inp_B|}\sum_{k \in G\setminus \DE{e}} \langle \alpha_{k} | \I_{AB} \otimes \Phi_{k} | \beta_{k} \rangle} \\
& \leq  \frac{1}{|G|} \de{ 1 + \sqrt{|\inp_A||\inp_B|} \sum_{k \in G\setminus \{e\}} \Vert \I_{AB} \otimes \Phi_{k}\Vert } \\
                    &= \frac{1}{|G|} \de{ 1 + \sqrt{|\inp_A||\inp_B|} \sum_{k \in G\setminus \{e\}} \Vert \Phi_{k}\Vert }  ,         
\end{split}
                    \end{align}
where the supremum in the first equation is taken over all vectors $\DE{\ket{\alpha_k}}$ and $\DE{\ket{\beta_k}}$ 
that can be expressed as in Eq. \eqref{vecqstrategy}.
\end{proof}

The particular case of linear games corresponding to the cyclic group $\mathbb{Z}_d$ (the set  $[d]=\DE{0,\ldots,d-1}$ with the 
operation of sum modulo $d$) we denote by \emph{generalized {\xor} games}, or simply \textsc{xor}-$d$ games. The characters 
of the cyclic group $\mathbb{Z}_d$ correspond to the $d$-th roots of unity $\chi_j(a)=\zeta^{ja}$, where $\zeta = \exp{(2 \pi \mathrm{i}/d)}$.
For an {\xor}-$d$, the Eq. \eqref{eqLGbound-2} reduces to 
\begin{eqnarray}
\label{xor-d-bound}
\omega_{q}(\textsl{g}^{\oplus_d}) \leq \frac{1}{d} \de{ 1 + \sqrt{|\inp_A| |\inp_B|} \sum_{k= 1}^{d-1} \norm{ \Phi_{k}}},
\end{eqnarray}
with 
\begin{align}
 \Phi_k = \sum_{x, y} p(x,y) \zeta^{k f(x,y)} \ketbra{x}{y}.
\end{align}

\subsection*{The computational complexity of our bound}

Theorem \ref{thmLGwqbound} states an upper bound to the quantum value of linear games based on the spectral norm of the game matrices.
Although we still do not know how good the bound is in general,
it satisfies one of the requirements of a good relaxation: the bound is easy to compute.

The spectral norm of a matrix $A$, $\norm{A}$, is equal to its maximum singular value. 
The singular value decomposition (SVD) of an $m\times n$ matrix $A$ is the decomposition of $A$ into the form
\begin{align}\label{eqSVD}
 A= U \Sigma V^{\dagger},
\end{align}
where
\begin{itemize}
 \item $U$: is an $m\times m$ matrix whose columns are composed by a set of orthonormal vectors which are called the \emph{left singular 
 vectors} of $A$.
 \item  $V$: is an $n\times n$ matrix whose columns are composed by a set of orthonormal vectors which are called the \emph{right singular 
 vectors} of $A$.
 \item $\Sigma$: is an $m\times n$ matrix with nonnegative elements in the principal diagonal, the \emph{singular values of $A$}, and zero 
 elsewhere.
\end{itemize}

Many algorithms are known for the singular value decomposition (see Ref. \cite{SVDcomplex}).
The best known SVD algorithms have polynomial complexity in terms of the size of the matrices. If one is interested in determining only the 
singular values of an $n \times n$ matrix, there exists an algorithm with time complexity $T(n)=\mathcal{O}(n^3)$ \cite{SVDcomplex}.
In a linear game $g^{\ell}(G,f,p)$ with $d$ possible outcomes and $m$ questions per player, we have
$(d-1)$ $m\times m$ game matrices. Therefore the time complexity of computing our bound is $T(d,m)=\mathcal{O}(dm^3)$, which increases
polynomially in the number of inputs and outputs. 

\section{Applications of the bound}\label{secapplicationsxord}

\subsection*{The CHSH-$d$ game}

As a direct application of Theorem \ref{thmLGwqbound}, we consider 
 a $d$-output generalization, for $d$ prime or power of a prime, of the CHSH game. The CHSH-$d$ game is defined 
 with the operations of sum and multiplication
over the finite field $\mathbb{F}_d$ (For more details on finite fields see Appendix \ref{A-finitefield}).

\begin{definition}\label{defchshd}
The CHSH-$d$ game is a linear game with $d$ inputs and $d$ outputs per player, defined for $d$ prime or power of a prime.
In the CHSH-$d$ game, Alice and Bob receive questions $x$ and $y$ and output answers
$a$ and $b$ respectively, $a,b,x,y \in \mathbb{F}_d$, with the goal to satisfy
\begin{align}
 a +b = x \cdot y
\end{align}
where $+$ and $\cdot$ are operations defined over the field $\mathbb{F}_d$. 
\end{definition}

Definition \ref{defchshd} is the generalization of the CHSH game for more outputs considered in Ref.
\cite{Bavarian}. Similar definitions were previously considered in 
Refs. \cite{chshmodd,chshcorre,Liangchshd}. It is interesting to note that these games have recently found application in the security analysis of a 
relativistic bit commitment protocol \cite{BCJedrelativistic}.

In Ref. \cite{Bavarian}, an intensive study of this game was performed. The authors present results on
the asymptotic classical and
quantum values of the game. They also prove, for the first time, an upper bound on the quantum value of the CHSH-$d$ game. Their proof
is based on reducing these games to other information theoretic principles like no-advantage for nonlocal computation \cite{NLC} and
information causality \cite{IC}.
We now apply Theorem \ref{thmLGwqbound} to re-derive in a different way 
the upper
bound for the quantum value of the CHSH-$d$ game obtained in Ref. \cite{Bavarian}. 

\begin{theorem}[see also \cite{Bavarian}]The quantum value of the CHSH-$d$ game, for $d$ prime or power of a prime, is upper bounded by
\begin{equation}
\label{Bavarian-bound-app}
\omega_q(\text{CHSH-}d) \leq \frac{1}{d} + \frac{d-1}{d \sqrt{d}}.
\end{equation} 
\end{theorem}

\begin{proof} 
The proof follows from the explicit analysis of the game matrices for the CHSH-$d$ game.
For the CHSH-$d$ game, the inputs are uniformly distributed and the winning condition is defined by $f(x,y) = x \cdot y$. Therefore 
the  game matrices are:
\begin{align}
 \Phi_k=\sum_{x,y=0}^{d-1}\frac{1}{d^2}\chi_k(x\cdot y)\ketbra{x}{y}
\end{align}
 where $\chi_k$ is the character of
the additive group formed by the elements of the field $\mathbb{F}_d$.

Now we evaluate $\Phi_k^{\dagger}\Phi_k$ using the characters relations \eqref{eqcharacprop}:
\begin{align}
\begin{split}
 \Phi_k^{\dagger}\Phi_k &=\frac{1}{d^4} \sum_{x,y=0}^{d-1}\sum_{x',y'=0}^{d-1} \xbar{\chi}_k(x \cdot y)\chi_k(x' \cdot y') \ket{y}{\braket{x}{x'}}\bra{y'}\\
 &=\frac{1}{d^4} \sum_{x,y=0}^{d-1}\sum_{y'=0}^{d-1}  \chi_k(- x \cdot y)\chi_k(x \cdot y') \ketbra{y}{y'}\\
 &=\frac{1}{d^4} \sum_{y,y'=0}^{d-1} \underbrace{\de{\sum_{x=0}^{d-1}  \chi_k(x\cdot (y'-y))}}_{d \delta_{y,y'}} \ketbra{y}{y'}\\
 &=\frac{1}{d^3} \sum_{y=0}^{d-1} \ketbra{y}{y}.
 \end{split}
\end{align}
Therefore $\Phi_k^{\dagger} \Phi_k =  \I/d^3$, 
so that $\norm{ \Phi_k} = 1/d \sqrt{d}, \; \; \forall k \in \DE{1,\ldots,d-1}$. 
Substitution into Eq. \eqref{xor-d-bound}, with $|\inp_A| = |\inp_B| = d$, yields the desired result.
\end{proof}


Comparison with the numerical results of Ref. \cite{Liangchshd} (see Table III in Ref. \cite{Liangchshd}) 
indicates that the bound \eqref{Bavarian-bound-app} is not tight in general\footnote{For $d=3$, results of Ref. \cite{Liangchshd} show that the optimal 
value is smaller than the bound \eqref{Bavarian-bound-app}. And moreover it is
attained with the maximally entangled state of dimension 3.} but might correspond to the 
value obtained for the first level of the NPA hierarchy.
Note that in Ref. \cite{Liangchshd} the authors only present 
 the value attained at the first level of the hierarchy up to $d=7$.
 This is probably due to the fact that the NPA hierarchy becomes impractical  for dealing with Bell inequalities with high number of inputs and outputs,
 which shows that our simple bound can be very powerful for these cases.

\subsection*{No trivialization of communication complexity}

We now address the question of whether there exist linear games that can be won perfectly with a quantum strategy, \ie 
if there exist  games $g^{\ell}(G,f,p)$ for which $\omega_q(g^{\ell})=1$.
 The interest in this question comes from communication complexity. 
In Ref. \cite{vanDam} it was shown that if Alice and Bob had unlimited access to PR-boxes \eqref{eqprbox}, they
could compute any distributed Boolean function with only one bit of communication. This result was later generalized \cite{PRd} to 
functions with $d$ possible values. In Ref. \cite{PRd} it was shown that the called \emph{functional boxes}, which are no-signaling boxes
that win perfectly some \xor-$d$ games, would lead to a trivialization of communication complexity, where a distributed function
could be computed with a single $d$it of communication (See Appendix \ref{A-CC} for more details).

The \xor-$d$ games for which a perfect no-signaling strategy can trivialize communication complexity are
the uniformly distributed total function games with the winning condition given by a non-additively-separable function. A total function game is 
one for which all inputs
have a probability strictly greater than zero to be chosen by the referee, $p(x,y)>0\; \forall x,y$. And a function $f(x,y)$ is 
additively separable if it can be decomposed into the form $f(x,y)=f_1(x)+f_2(y)$.


We now make use of Theorem \ref{thmLGwqbound} to show that, for \textsc{xor}-$d$ games with uniformly chosen inputs, 
there is a quantum strategy that wins the game with probability one if and only if this game is trivial, \ie 
when $\omega_c(g^{\oplus_d})=1$.

\begin{theorem}
\label{comm-comp-lem}
For \textsc{xor}-$d$ games $g^{\oplus_d}$  with $m$ questions per player and uniformly distributed  
inputs, $p(x,y) = 1/m^2$, $\omega_q(g^{\oplus_d}) = 1$ iff $\omega_c(g^{\oplus_d}) = 1$.
\end{theorem}

\begin{proof}
The constraint of uniformly distributed questions, $p(x,y) = 1/m^2$ for all $(x,y)$, 
is equivalent to $\norm{\Phi_k} \leq 1/m$ since both the maximum absolute value column sum and row sum of the matrix are equal to $1/m$.
Hence, from our bound, Eq. \eqref{xor-d-bound}, we have that 
\begin{align}
 \omega_q(\textsl{g}^{\oplus_d}) = 1\Rightarrow \norm{\Phi_k}=\frac{1}{m}\; \forall k \in \{1, \dots, d-1\}.
\end{align}

Now, let us consider the matrix $\Phi_1^{\dagger} \Phi_1$:
\begin{align}
\De{\Phi_1^{\dagger} \Phi_1}_{y,y'} = \sum_{x=0}^{m-1} \frac{1}{m^4} \zeta^{-f(x,y) + f(x,y')},
\end{align}
where $\zeta = \exp{(2 \pi \mathrm{i}/d)}$. 
Let $\ket{ \lambda}=\de{\lambda_0,\ldots,\lambda_{m-1}}$ be an eigenvector corresponding to the maximum eigenvalue $1/m^2$ of $\Phi_1^{\dagger} \Phi_1$, 
with complex entries $\lambda_j = \vert \lambda_j \vert \zeta^{{\theta}_j}$. Assume, without loss of generality, that
the entries of the eigenvector are ordered by absolute value, 
$\vert \lambda_0 \vert \geq \dots \geq \vert \lambda_{m-1} \vert$. From the eigenvalue equation corresponding to the first entry of
$\ket{\lambda}$ we have 
\begin{equation}
\sum_{x, y = 0}^{m-1} |\lambda_y| \zeta^{-f(x,0) + f(x,y) + \theta_y} = m^2 |\lambda_0| \zeta^{\theta_0}.
\end{equation}
Since $|\lambda_0| \geq |\lambda_j| \; \forall j$, the above equation can only be satisfied when 
\begin{subequations}
\begin{align}
 \vert \lambda_j \vert& = \vert \lambda_{0} \vert \;\; \forall j\\
 f(x,y) - f(x,0) + \theta_y &= f(x',y') - f(x',0) + \theta_{y'}\; \; \forall x,y,x',y'
\end{align}
\end{subequations}
 in particular choosing $x = x'$ we get 
 \begin{align}
 f(x,y) - f(x,y') = \theta_{y'} - \theta_{y} \; \forall x,y,y',
 \end{align}
 where all the operations are modulo $d$.
 With all $|\lambda_{j}|$ equal, the rest of the eigenvalue equations
(for $j \neq 0$) lead to similar consistent constraint equations. 
We then deduce that $\omega_q(\textsl{g}^{\oplus_d}) = 1$ only when the columns of the game matrix $\Phi_1$ are 
proportional to each other, the proportionality 
factor between columns $y, y'$ being $\zeta^{f(x,y) - f(x,y')} = \zeta^{\theta_{y'} - \theta_y}$, and therefore the game 
matrix has $\text{rank}(\Phi_1) = 1$. 
Now consider $a_0$ and $b_0$ satisfying $a_0+ b_0=f(0,0)$, by setting
\begin{align}
 a_x&=f(x,0)-b_0\\
 b_x&=b_0 +(\theta_0-\theta_y)
\end{align}
we have a classical strategy that wins the game with probability 1.
\end{proof}

A more general result was recently proved in Ref. \cite{Ravinoextpoint}, showing that all the extremal points of the no-signaling
polytope, in any Bell scenario, cannot be realized
within quantum theory.  
Here, by a direct application of the norm bound (Theorem \ref{thmLGwqbound}) we are able to exclude the quantum realization of those
boxes corresponding to \xor-$d$ games that would lead to trivialization of communication complexity.

\section{\texorpdfstring{\xor-$d$}{XOR-d} games and the task of nonlocal computation}

In Section \ref{secnoqadvantage} we have introduced the principle of no-advantage for nonlocal computation proposed in Ref. \cite{NLC}, which 
corresponds to a class of {\xor} games for which $\omega_q(g^{\oplus})=\omega_c(g^{\oplus})$.
The question of the generalization of this class to a larger alphabet size was also left posed as an 
open question in Ref. \cite{NLC}. Here we use Theorem \ref{thmLGwqbound} in order to characterize a 
class of \xor-$d$ games, that resembles NLC, for which there is no quantum advantage.


Consider the following generalization of the non-local computation task to the computation of 
 a particular function $f(z_1, \dots, z_n)$ on $n$ $d$its, $z_i \in \mathbb{F}_d$ for $d$ prime.

\begin{definition}[$NLC_d$]
The generalized nonlocal computation of a $d$-nary function, $NLC_d$, for $d$ prime, is the task where
Alice and Bob each receives a $n$-$d$it string from a referee, $\vec{x}_n = (x_1, \dots, x_n)$ 
and $\vec{y}_n = (y_1, \dots, y_n)$, $\vec{x}_n,\vec{y}_n \in \mathbb{F}_d^n$, which obey $x_i + y_i = z_i$.
They output respectively dits $a, b\in \mathbb{F}_d$ with the goal to 
satisfy
\begin{align}\label{func-NLC}
a + b = h(\vec{x}_{n-1} + \vec{y}_{n-1}) \cdot (x_n+ y_{n}),
\end{align}
for a previously agreed function $h:\mathbb{F}_d^{n-1} \times \mathbb{F}_d^{n-1} \longrightarrow \mathbb{F}_d$, where $+,\cdot$ are sum and multiplication
modulo $d$. Moreover, the input strings are chosen by the 
referee according to the distribution
\begin{align}\label{probdistr}
p(\vec{x}_n,\vec{y_n})=\frac{1}{d^{n+1}} \tilde{p}(\vec{x}_{n-1} + \vec{y}_{n-1})
\end{align}
for $\tilde{p}(\vec{z}_{n-1})$ being an arbitrary probability distribution. 
\end{definition}

We now prove that the games $NLC_d$, as defined above, exhibit no quantum advantage. The idea behind the proof 
is to show that the matrices $\Phi_k^{\dagger} \Phi_k$ for these games are diagonal in
a basis composed of tensor products of the Fourier vectors of dimension $d$. We then present a classical strategy which 
achieves the quantum value, which is essentially given by the maximum singular vectors of $\Phi_1$. 
\begin{theorem}
\label{thm-nlc}
The games $NLC_d$ for arbitrary prime $d$ and input distribution satisfying \eqref{probdistr} have no quantum advantage, i.e., $\omega_c(NLC_d) = \omega_q(NLC_d)$.
\end{theorem}

\begin{proof}
We first consider the case of uniformly chosen inputs. The games $NLC_d$ consider functions of the following form 
\begin{equation}
\label{eqwinNLCd}
a + b = h(x_1 + y_1, \dots, x_{n-1} + y_{n-1}) \cdot (x_n+ y_{n}),
\end{equation}
where $+$ and $\cdot$ are sum and multiplication modulo $d$, and
 $h$ is an arbitrary function. Given the winning condition \eqref{eqwinNLCd}, 
 the game matrices of $NLC_d$ are composed of ``building-block games" $g(t)$:
\begin{equation}
\label{single-game}
g(t):=\DE{a + b= t \cdot (x + y)},
\end{equation}
with $t \in \{0, \dots, d-1\}$, i.e., $f(x,y) = t \cdot (x + y)$. 

There are $d$ different games $g(t)$, each with a single $d$it input for each party (which we will take to be $x_n$ and $y_n$).
Every game $g(t)$ has the classical value $\omega_c(g(t))= 1$. Explicitly, the classical strategy
\begin{align}
 a = t \cdot x \;\; \text{and}\;\; b = t \cdot y
\end{align}
wins the game $g(t)$ with probability one.
The corresponding (non-normalized) game matrices $\tilde{\Phi}^{(1)}_{k}(t)$ for game $g(t)$ are given by 
\begin{equation}
\label{single-game-mat}
\tilde{\Phi}^{(1)}_{k}(t) := \sum_{x,y \in [d]} \zeta^{k t (x + y)} \ketbra{x}{y},
\end{equation}
with $\zeta = \exp{(2\pi \mathrm{i}/d)}$. Here the superscript $(1)$ denotes that these matrices correspond to the $NLC_d$ 
game matrices for $n=1$. 

Let us state some properties of the matrices $\tilde{\Phi}^{(1)}_{k}(t)$:
\begin{enumerate}[(i)]
 \item ${\tilde{\Phi}^{(1)}_{k}(t)}^{\dagger} \tilde{\Phi}^{(1)}_{k}(t)$ for any $k, t$ is diagonal in the Fourier basis defined by the Fourier vectors $|v_j\rangle$ with
\begin{equation}
|v_{j} \rangle = \left(1, \zeta^j, \zeta^{2j}, \dots, \zeta^{(d-1)j}\right)
\end{equation}
with $j \in \{0, \dots, d-1\}$.
\item Each ${\tilde{\Phi}^{(1)}_{k}(t)}^{\dagger} \tilde{\Phi}^{(1)}_{k}(t)$ has only one eigenvalue (=$d^2$) different from zero and
this corresponds to the eigenvector $\ket{v_{d-k \cdot t}}$. 
\end{enumerate}

Properties (i) and (ii) imply the 
orthogonality 
\begin{align}
 {\tilde{\Phi}^{(1)}_{k}(t)}^{\dagger} \tilde{\Phi}^{(1)}_{k'}(t') = 0\;\; \text{for}\;\;k \cdot t \neq k' \cdot t'.
\end{align}
Since, we will be concerned with finding the maximum singular
vectors corresponding to a fixed $k$, we can encapsulate the above properties by the equation 
\begin{equation}
\label{single-game-prop}
\de{ {\tilde{\Phi}^{(1)}_{k}(t)}^{\dagger} \tilde{\Phi}^{(1)}_{k}(t') } |v_{j} \rangle = d^2 \delta_{t, t'} \delta_{j, d- k.t}  \ket{v_{j}}.
\end{equation}

Now we use the properties of $\tilde{\Phi}^{(1)}_k(t)$ in order to analyze 
the game matrices $\tilde{\Phi}^{(n)}_k$ for the general $NLC_d$ games with $n$-$d$it strings of input.
Due to the structure of the function in Eq. \eqref{func-NLC}, namely the fact that the winning condition
depends only on the $d$it-wise sum 
of the $n$ dits, and moreover the dependence on the last $d$its, $x_n,y_n$, is given by the games $g(n)$,
 we see that ${{\tilde{\Phi}}_k^{(n)\dagger}} \tilde{\Phi}^{(n)}_k$ acquires a block circulant structure
 (for $1 \leq i \leq n$ the corresponding matrices ${\tilde{\Phi}^{(i)\dagger}_k} \tilde{\Phi}^{(i)}_k$ for each $k$ 
 are block-wise circulant matrices). For example, if for $n=2, d=3$ an unnormalized game matrix $\tilde{\Phi}^{(2)}$  has the form
\begin{equation}
\tilde{\Phi}^{(2)}\defeq\scalebox{0.8}{ \begin{tabular}{| l |  l | r  | }
  \hline
    $ \tilde{\Phi}^{(1)}(0)$ & $\tilde{\Phi}^{(1)}(1)$  & $\tilde{\Phi}^{(1)}(2)$   \\ \hline
    $ \tilde{\Phi}^{(1)}(1)$ & $\tilde{\Phi}^{(1)}(2)$  & $\tilde{\Phi}^{(1)}(0)$  \\  \hline
     $\tilde{\Phi}^{(1)}(2)$  & $\tilde{\Phi}^{(1)}(0)$ & $\tilde{\Phi}^{(1)}(1)$  \\ \hline
  \end{tabular}  }
 \end{equation}
with $\tilde{\Phi}^{(1)}(t)$ defined as in Eq. \eqref{single-game-mat}, we would have $\tilde{\Phi}^{(2)\dagger} \tilde{\Phi}^{(2)}$ equals to 
\begin{equation}
\tilde{\Phi}^{(2)\dagger} \tilde{\Phi}^{(2)}=\scalebox{0.8}{ \begin{tabular}{| l  | l | l  | }
  \hline
     $\sum_{i} {\tilde{\Phi}^{(1)}(i)}^{\dagger} {\tilde{\Phi}^{(1)}(i)}$ & $\sum_{i} {\tilde{\Phi}^{(1)}(i)}^{\dagger} {\tilde{\Phi}^{(1)}(i+1)}$  & $\sum_{i} {\tilde{\Phi}^{(1)}(i)}^{\dagger} {\tilde{\Phi}^{(1)}(i+2)}$   \\ \hline
     $\sum_{i} {\tilde{\Phi}^{(1)}(i)}^{\dagger}{\tilde{\Phi}^{(1)}(i+2)}$ & $\sum_{i} {\tilde{\Phi}^{(1)}(i)}^{\dagger} {\tilde{\Phi}^{(1)}(i)}$  & $\sum_{i} {\tilde{\Phi}^{(1)}(i)}^{\dagger}{\tilde{\Phi}^{(1)}(i+1)}$  \\  \hline
     $\sum_{i} {\tilde{\Phi}^{(1)}(i)}^{\dagger}{\tilde{\Phi}^{(1)}(i+1)}$  & $\sum_{i} {\tilde{\Phi}^{(1)}(i)}^{\dagger} {\tilde{\Phi}^{(1)}(i+2)}$ & $\sum_{i} {\tilde{\Phi}^{(1)}(i)}^{\dagger} {\tilde{\Phi}^{(1)}(i)}$ \\ \hline
  \end{tabular}  }
\end{equation}
which is a block-wise circulant matrix. 
In general, the entries of ${\tilde{\Phi}^{(n){\dagger}}_k} \tilde{\Phi}^{(n)}_k$ are explicitly given by 
\begin{align}
 \begin{split}
\De{{\tilde{\Phi}^{(n){\dagger}}_k} \tilde{\Phi}^{(n)}_k}_{\vec{x}_{n-1}, \vec{y}_{n-1}} =& \\
\sum_{u_1, \dots, u_{n-1} =0}^{d-1} &{\tilde{\Phi}^{(1)}_{k}(h(\vec{x}_{n-1}+ \vec{u}_{n-1}))}^{\dagger} \tilde{\Phi}^{(1)}_{k}(h(\vec{u}_{n-1}+ \vec{y}_{n-1})).
\end{split}
\end{align}
Due to this block circulant structure, we have that ${\tilde{\Phi}^{(n)\dagger}_k} \tilde{\Phi}^{(n)}_k$ for any $n$ and $k$ 
is diagonal in the basis formed by the tensor products of the Fourier vectors 
$\{|v_{i_1}\rangle \otimes \ldots \otimes|v_{i_{n}} \rangle\}$ with $i_1, \dots, i_n \in \{0, \dots, d-1\}$. 

We now proceed to  find the eigenvector corresponding to the maximum eigenvalue of ${\tilde{\Phi}^{(n)\dagger}_{k}} \tilde{\Phi}^{(n)}_{k}$ 
among the basis formed by $\{|v_{i_1}\rangle \otimes \ldots \otimes |v_{i_{n}} \rangle\}$.
Using the properties of the game matrices $\tilde{\Phi}^{(1)}_{k}(t)$ encapsulated by Eq. \eqref{single-game-prop}, 
 one can see that for any fixed $i_n$, the eigenvalue corresponding to $|v_0 \rangle^{\otimes n-1} \otimes |v_{i_n} \rangle$ cannot
 be smaller than that corresponding to any other $|v_{i_1} \rangle \otimes \ldots \otimes |v_{i_n} \rangle$. 
 Therefore we can concentrate only on the vectors $|v_0 \rangle^{\otimes n-1} \otimes |v_{i_n} \rangle$.

Let us compute the eigenvalues corresponding to $|v_0\rangle^{\otimes n-1} \otimes |v_{i_n} \rangle$ 
for $i_n \in \{0, \dots, d-1\}$. To do this, let us fix an input string $\vec{x}_{n-1}$ (say $(0,\dots,0)$) and 
vary over $\vec{y}_{n-1}$, in other words we consider the first row block of
$\tilde{\Phi}^{(n)}_k$ corresponding to the game blocks 
$\tilde{\Phi}^{(1)}_{k}(t)$, with $t=h(\vec{0}_{n-1} + \vec{y}_{n-1})$. 
Denote by $\lambda^{\vec{x}_{n-1}}(i_n, k)$ the number of times the 
game $g{(d-k^{-1} \cdot i_n)}$ appears for this  choice of $\vec{x}_{n-1}$ in the matrix $\tilde{\Phi}^{(n)}_k$. Due to the symmetry of the winning condition, 
$\lambda^{\vec{x}_{n-1}}(i_n, k)$ is independent of the choice of row $\vec{x}_{n-1}$ so we may drop the superscript. 
Moreover, since $\tilde{\Phi}^{(n)}_k$ is
a symmetric matrix, we also have $\lambda^{\vec{x}_{n-1}}(i_n, k) = \lambda^{\vec{y}_{n-1}}(i_n, k)$ for an analogously 
defined $\lambda^{\vec{y}_{n-1}}(i_n, k)$.

Let us define $\Lambda(k) \defeq \max_{i_n} \lambda(i_n, k)$
and let $\mu\defeq d-k^{-1}\cdot i_n$ for 
the value of $i_n$ for which the maximum of $\lambda(i_n, k)$ is achieved. 
Using Eq. \eqref{single-game-prop}, we have that 
\begin{equation}
\de{{\tilde{\Phi}^{(n){\dagger}}_k} \tilde{\Phi}^{(n)}_k } |v_0 \rangle^{\otimes n-1} \otimes |v_{i_n} \rangle = d^2 \lambda^2(i_n, k) |v_0 \rangle^{\otimes n-1} \otimes |v_{i_n} \rangle,
\end{equation}
from which we obtain that $\norm{\tilde{\Phi}^{(n)}_k} = d \Lambda(k)$. 

For prime $d$, multiplication (mod $d$) by $k\neq 0$ maps each game $g(t)$ into a game $g(t'=k\cdot t)$ such that if
$t_1\neq t_2$ then $t'_1 \neq t'_2$. Therefore the maximum number of elementary games of the same type composing
matrix $\Phi^{(n)}_k$ is the same for all $k$, $\Lambda(k) = \Lambda$.
Hence we obtain the following bound
on the quantum value of $NLC_d$ for the uniformly distributed inputs case
\begin{equation}
\label{uni-q-bound}
\omega_q(NLC_d)  \leq  \frac{1}{d}\left(1 + \frac{(d-1) \Lambda}{d^{n-1}} \right).
\end{equation}
We now consider the classical deterministic strategy where Alice outputs $a = \mu \cdot x_n$ independently of her inputs $\vec{x}_{n-1}$
and Bob outputs $b = \mu \cdot y_n$ 
independently of his input $\vec{y}_{n-1}$. Note that for the $d \times d$ blocks described by the 
game $g(\mu)$ all the $d^2$ constraints will be satisfied. 
On the other hand, for the blocks described by $g(t)$ for $t \neq \mu$, only $d$ constraints are satisfied (when $x_n + y_n=0)$.
Therefore the score achieved by this strategy is given by
\begin{equation}
\omega_c(NLC_d) = \frac{d^{n-1}}{d^{2n}}\De{\Lambda d^2 + (d^{n-1} - \Lambda) d},
\end{equation}
which equals the upper bound on the quantum value in Eq. (\ref{uni-q-bound}). 
This completes the proof for uniformly chosen inputs.\vspace{1em}

Now we consider the case of input probability distributions
\begin{equation}
p(\vec{x}_n,\vec{y}_n)=\frac{1}{d^{n+1}} \tilde{p}(\vec{x}_{n-1}+ \vec{y}_{n-1}).
\end{equation}
For this input distribution, the matrix ${\Phi}_{k}^{(n)}$ is still composed of the elementary games $\Phi^{(1)}_{k}(t)$ that can be classically saturated.
The difference is that a weight $ \tilde{p}(\vec{x}_{n-1} + \vec{y}_{n-1})/d^{n+1}$
is now attributed to each $d \times d$ block
\begin{equation}
\De{{\Phi}^{(n)}_k}_{\vec{x}_{n-1}, \vec{y}_{n-1}} = \frac{1}{d^{n+1}} \tilde{p}(\vec{x}_{n-1} + \vec{y}_{n-1}) \Phi^{(1)}_{k}(h(\vec{x}_{n-1}+ \vec{y}_{n-1})).
\end{equation}
This preserves the block-wise circulant structure of ${\Phi^{(n)}_{k}}^{\dagger} \Phi^{(n)}_{k}$
ensuring that these matrices are still diagonal in the basis formed by the tensor product of Fourier vectors. 
As in the case of uniformly distributed inputs, the properties of $\Phi^{(1)}_{k}(t)$ in Eq. \eqref{single-game-prop} 
imply that the maximum eigenvalue corresponds to one 
of the vectors $|v_0\rangle^{\otimes n-1} \otimes |v_{i_n} \rangle$. 

To compute the eigenvalues corresponding to $|v_0\rangle^{\otimes n-1} \otimes |v_{i_n} \rangle$, we now have 
to take into account the number of times a game $g{(d-k^{-1} \cdot i_n)}$ appears in a given row block as well as the respective weights.
Let us denote by $\tilde{\lambda}(i_n, k)$ the weighted sum of the times the game $g{(d-k^{-1} \cdot i_n)}$ appears in a row block, i.e.,
\begin{align}
\tilde{\lambda}(i_n, k)=\sum_{\stackrel{\vec{y}_{n-1} \text{ s.t.}}{ h(\vec{0}_{n-1}+ \vec{y}_{n-1})=d-k^{-1} \cdot i_n}}  \frac{1}{d^{n+1}} \tilde{p}(\vec{0}_{n-1} + \vec{y}_{n-1}).
\end{align}
As before, let us define $\tilde{\Lambda}(k) := \max_{i_n} \tilde{\lambda}(i_n, k)$ and let $\mu=d-k^{-1} \cdot i_n$ for the $i_n$ which achieves  the maximum.
For the game  matrix $\Phi^{(n)}_k$ we have
\begin{equation}
\de{{{\Phi}^{(n)\dagger}_k} {\Phi}^{(n)}_k }|v_0 \rangle^{\otimes n-1} \otimes |v_{i_n} \rangle = d^2 \tilde{\lambda}(i_n, k)^2 |v_0 \rangle^{\otimes n-1} \otimes |v_{i_n} \rangle.
\end{equation}
We therefore obtain that $\| {\Phi}^{(n)}_k \| = d \tilde{\Lambda}(k)$. 

Again, for prime $d$, the maximum of $\tilde{\Lambda}(k)$ is independent of  $k$. Therefore, we obtain the 
following upper bound on the quantum value of a general $NLC_d$ game
\begin{equation}
\label{nlc-q-bound}
\omega_q(NLC_d)  \leq  \frac{1}{d}\De{1 + d^{n+1}(d-1) \tilde{\Lambda}}.
\end{equation}
Consider the classical deterministic strategy where Alice outputs $a = \mu x_n$ independently of $\vec{x}_{n-1}$ and Bob 
outputs $b = \mu y_n$ 
independently of $\vec{y}_{n-1}$. 
Analogously to the uniformly distributed inputs case, the score achieved by this strategy is 
\begin{equation}
\omega_c(NLC_d) = d^{n-1} \De{\tilde{\Lambda} d^2 + \de{\frac{1}{d^{n+1}} - \tilde{\Lambda}}d },
\end{equation}
which again equals the upper bound on the quantum value in Eq.\eqref{nlc-q-bound}.
This completes the proof that quantum strategies cannot outperform classical strategies in the $NLC_d$ game.
\end{proof}

Note that our proof relies on the assumption that the winning constraint function has the form 
$h(\vec{x}_{n-1}+ \vec{y}_{n-1})\cdot (x_n + y_n)$, which seems more restrictive than stated for the binary $NLC$ \cite{NLC}.
Now, let us consider the  3-input and 3-output game $g^{eg}$, \ie $d=3$ and $n=1$, whose winning 
condition is specified by 
\begin{align}\label{gameeg}
 f(x,y)=\begin{cases}
         0\,, \,& \text{ if } x+y=0 \text{ or } x+y=1 \\
          1\,, \,& \text{ if } x+y=2 
        \end{cases}
\end{align}
and the associated game matrices are
$
 \Phi^{eg}_1=\frac{1}{9}\left(\begin{smallmatrix}
1&1&\zeta \\ 1&\zeta&1\\\zeta&1&1
\end{smallmatrix}\right)
$ and $\Phi^{eg}_2={\Phi^{eg}_1}^*$.
An SDP optimization over measurements for the maximally entangled state in dimension 3 shows that the quantum value overcomes the classical
value of the game \eqref{gameeg}. Therefore Theorem \ref{thm-nlc} cannot be extended to an arbitrary function $f(\vec{x}_n+ \vec{y}_n)$.

\section{Discussion and open problems} 
 In this chapter, we have presented an upper bound to the quantum value of linear games. 
 The bound is not tight in general but it is very simple and, as we have shown by examples, it allowed us to derive several 
 results:
 We have used the bound to
 rule out from the quantum set a class of no-signaling boxes that would result in trivialization of communication
 complexity; Also, we have shown 
 that the recently discovered bound on the quantum value of the CHSH-$d$ game, obtained  in Ref. \cite{Bavarian}, 
 can be derived in a simple manner using our bound; And finally, we have
 extended the principle of no-advantage for nonlocal computation to a class of functions with prime $d$ possible values.
 
 Moreover the derived bound is efficiently computable since, for a linear game with $d$ outcomes, it requires 
 the spectral norm of $d$ game matrices, where the size of these matrices grows polynomially with the number of questions in the game (a
 game with $m$ question per player has game matrices of size $m \times m$).

As a future direction, it would be particularly interesting to investigate if one can extend 
 the result of Theorem \ref{thmxorshannon} to the graphs associated to $NLC_d$.
 As a more challenging next step, we point to the generalization of the technique of norm bounds to classes of Bell 
inequalities beyond linear games. Due to its simplicity, it could 
lead to general results of fundamental and practical interest.

%% file: multiplayerbound.tex
\chapter{Multiplayer linear games and device-independent witness of genuine tripartite entanglement}\label{chapternplayer}
\chaptermark{Multiplayer linear games and DIEWs}

In this chapter we present results of Ref. \cite{GMmultiplayer}
\begin{center}
\begin{minipage}{12.5cm}
 \textit{Quantum bounds on multiplayer linear games and device-independent witness of genuine tripartite entanglement }\\
  \textbf{G. Murta}, R. Ramanathan, N. M\'oller, and M. Terra Cunha\\
   \href{http://link.aps.org/doi/10.1103/PhysRevA.93.022305}{\textbf{Phys. Rev. A, \textbf{93}, 022305, (2016)}}.
\end{minipage}
\end{center}
Now we consider the case of linear games with $n$ players. 
We generalize the bound obtained in the previous Chapter \cite{GMxord} to 
the quantum value of an $n$-player game. We extend the examples of the 2-player case, deriving an upper bound to the 
quantum value of a generalization of the CHSH-$d$ game for $n$ players, and also we exclude the possibility of quantum realization of 
multipartite functional boxes that would lead to trivialization of communication complexity in a multipartite scenario.
As our main result for the multipartite scenario, we use the bounds to design
devide-independent witnesses of genuine multipartite entanglement for tripartite systems.

\section{Motivation}
Multipartite scenarios bring fundamental and practical new features.
From the fundamental point of view, 
the possibility of having more parties interacting with each other brings the novelty of different classes of 
correlations (now all the parties can share nonlocal correlations or, else, only a subgroup of the parties can be non-classically correlated),
and therefore we can have a much richer nonlocality structure. 
Moreover, as shown in  Ref. \cite{NoBiPrinciple},
no bipartite principle is sufficient to single out the set of quantum correlations for an arbitrary number of parties, and 
hence the study of the intrinsically multipartite features is necessary.
From the practical point of view, multipartite scenarios
allow for the realization of many cryptographic tasks whose unconditional security cannot be guaranteed in a bipartite scenario \cite{MIP88}.
A remarkable example is the task of bit commitment for which no-go theorems \cite{BCMayers,BCLoChau} 
state the impossibility of having an unconditionally
secure bipartite protocol. This impossibility was circumvented by  Kent \cite{Kent11, Kent12} who proposed the idea of adding multiple
space-like separated 
agents  for each party (what is called a ``relativistic protocol''), allowing unconditional security as long as the agents remain
 separated. 
In Ref. \cite{BCJedrelativistic} a relativistic protocol was proposed, where the commitment can be made arbitrarily
long by the introduction of a rounding procedure. 
Interestingly, the security against classical adversaries is guaranteed by a mapping to the problem of estimating the performance of the players in a
multiplayer game.

Besides its undeniable importance, very few  results are known for  multipartite Bell scenarios. 
We now extend the techniques presented in Chapter \ref{chaptergamesd} \cite{GMxord} and provide an upper bound to the quantum value 
of multiplayer linear games based on game matrices.

\section{An efficiently computable bound to the quantum value of multiplayer linear games}\label{bound}

Our goal in this Section is to bound the performance of players sharing quantum resources in an $n$-player linear game $g_n^{\ell}(G,f,p)$
(Definition \ref{defLGn}).

 A generalization of Lemma \ref{lemmacorre} also holds for multiplayer games, and 
  the average probability of success on the game, $ \omega(g_n^{\ell})$, can be written in terms of the generalized correlators.  
 The \emph{multipartite generalized correlators} $\langle{A_1}_{x_1}^i\ldots {A_n}_{x_n}^j\rangle$ are defined as the 
Fourier transform of the probabilities
 \begin{align}
 \langle{A_1}_{x_1}^i\ldots {A_n}_{x_n}^j\rangle=\sum_{a_1,\ldots,a_n \in G} \bar{\chi}_i(a_1)\ldots \bar{\chi}_j(a_n)P(a_1,\ldots,a_n|x_1,\ldots,x_n),
 \end{align}
and the inverse Fourier transform gives us
  \begin{align}
P(a_1,\ldots,a_n|x_1,\ldots,x_n) =\frac{1}{|G|^n}\sum_{a_1,\ldots,a_n \in G} {\chi}_i(a_1)\ldots {\chi}_j(a_n)\langle{A_1}_{x_1}^i\ldots {A_n}_{x_n}^j\rangle.
 \end{align}
 
 \begin{lemma}\label{lemmacorren}
 Given a particular strategy $\vec{P}(a_1,\ldots,a_n|x_1,\ldots,x_n)$, the average probability of success in an $n$-player linear game
 $g_n^{\ell}(G,f,p)$ can be written as
 \begin{align}\label{eqwcorn}
  \omega(g_n^{\ell}) = \frac{1}{|G|}\de{1+\sum_{x_1,\ldots,x_n}\sum_{k \in G \setminus \DE{e}} {p(x_1,\ldots,x_n)} \chi_k(f(x_1,\ldots,x_n))\langle {A_1}_{x_1}^k\ldots {A_n}_{x_n}^k\rangle}.
 \end{align}
\end{lemma}

For the particular case of a 3-player linear game $g_3^{\ell}(G,f,p)$, Lemma \ref{lemmacorren} gives us
 \begin{align}\label{eqwcor3}
  \omega(g_3^{\ell}) = \frac{1}{|G|}\de{1+\sum_{x,y,z}\sum_{k \in G \setminus \DE{e}} {p(x,y,z)} \chi_k(f(x,y,z))\corre{A_x^kB_y^kC_z^k}}.
 \end{align}

\begin{proof}[Proof of Lemma \ref{lemmacorren}]
 We present the proof for the case of 3 players. The case of $n$ players 
 follows in an analogous way.
 
Given the probabilities in terms of the generalized correlators:
 \begin{align}
  P(a,b,c|x,y,z)=\frac{1}{|G|^3}\sum_{i,j,k \in G} {\chi}_i(a){\chi}_j(b){\chi}_k(c)\corre{A_x^iB_y^jC_z^k},
 \end{align}
we can proceed to calculate $P(a+b+c=f(x,y,z)|x,y,z)$:
  \begin{align}\label{eqcharacters3}
 \begin{split}
P(a + b +c= f(x&,y,z) | x,y,z) =\\
=&\sum_{a,b} P(a,b, f(x,y,z)-a-b|x,y,z)\\
			=&\sum_{a,b} \frac{1}{|G|^3} \sum_{i,j,k \in G} \chi_{i}(a) \chi_{j}(b) \chi_k(f(x,y)-a-b) \langle A_{x}^{i} B_{y}^{j} C_z^k\rangle \\
			=&\frac{1}{|G|^3} \sum_{i,j,k \in G} \chi_{k}(f(x,y)) \de{\sum_a \chi_{i}(a) \chi_{j}(-a)}\\
			&\;\;\;\;\;\;\;\;\;\;\;\;\;\;\times \de{\sum_b \chi_{j}(b) \chi_{k}(-b)}  \langle A_{x}^{i} B_{y}^{j}C_z^k\rangle \\
			=&\frac{1}{|G|} \sum_{k \in G} \chi_{k}(f(x,y)) \langle A_{x}^{k} B_{y}^{k}C_z^k\rangle 
			\end{split}
			\end{align}
 where we have used the characters' properties \eqref{eqcharacprop}.

The weighted sum over the inputs, gives us
the desired result. 
\end{proof}


Considering a particular quantum strategy given by the 
set of projective measurements $\{M_x^a\}$, $\{M_y^b\}$, $\{M_z^c\}$ performed  {on} the tripartite quantum state $\ket{\psi}$, the 
tripartite correlators correspond to
\begin{align}
 \corre{A_x^iB_y^jC_z^k}=\bra{\psi}A_x^i\otimes B_y^j \otimes C_z^k \ket{\psi},
\end{align}
where, as defined in the previous Chapter:
\begin{align}
 A_x^i=\sum_a\bar{\chi}_i(a)M_x^a,
\end{align}
and analogously for $B_y^j$ and $C_z^k$.

 Motivated by Lemma \ref{lemmacorren}, for tripartite linear games  we can also associate a set of $|G|-1$ (rectangular) matrices
 which carry information about
 the probability distribution with which the referee picks questions and also the winning condition of the game.
 
 \begin{definition}[Multiplayer linear game matrices]\label{defLGmatrices}
Given a linear game $g_3^{\ell}(G,f,p)$ the associated $|G|-1$ game matrices are defined as
 \begin{align}\label{phi3}
 \Phi_k = \sum_{(x,y,z) \in Q}p(x,y,z)\chi_k(f(x,y,z))\ketbra{x}{yz}\;\; \text{for}\;\; k\in G\setminus \DE{e}
\end{align}
{where $\DE{\ket{x}}$, $\DE{\ket{y}}$ and $\DE{\ket{z}}$ form orthonormal basis in $\C^{|\inp_1|}$, $\C^{|\inp_2|}$ and $\C^{|\inp_3|}$ 
respectively, and $Q=\inp_1\times \inp_2 \times \inp_3$.}
\end{definition}

Note that in Definition \ref{defLGmatrices} we have chosen to write the game matrices in terms of the partition  $x|yz$ of 
the inputs. However, we could equally chose any other partition 
$y|xz$ or $z|xy$ and define the respective game matrices in an analogous way. If the winning condition of the game $f(x,y,z)$ is not invariant 
over the permutation of parties, each partition would give rise to different matrices.

 Given all these definitions we are ready to state the main result of this Chapter which generalizes
  the norm bound presented in Chapter \ref{chaptergamesd} for multiplayer linear games. We start with a 3-player game.

\begin{theorem}\label{thmnorm3}
 The quantum value of a tripartite linear game, $g_3^{\ell}(G,f,p)$, where players $A, B$, and $C$ receive respectively questions 
 $x\in \inp_1, y \in \inp_2, z\in \inp_3$ and answer with elements of an
 Abelian group $(G,+)$, is upper bounded by
 \begin{align}\label{norm3}
  \omega_q(g_3^{\ell}) \leq \frac{1}{|G|}\de{1+\sqrt{|\inp_1||\inp_2||\inp_3|}\sum_{k \in G \setminus \DE{e}}\norm{\Phi_k}},
 \end{align}
where  {$\norm{\cdot}$} denotes the maximum singular value of the matrix, $e$ is the identity element of the group G, and
$ \Phi_k$ are the game matrices.
\end{theorem}

\begin{proof}
The proof follows analogously to the 2-player case.
Consider a quantum strategy where measurements $\{M_x^a\}$, $\{M_y^b\}$, $\{M_z^c\}$ are 
 performed on the tripartite quantum state $\ket{\psi}$, hence we have
 \begin{align}\label{swq}
  \omega(g_3^{\ell}) = \frac{1}{|G|}\de{1+\sum_{x,y,z}\sum_{k \in G\setminus \DE{e}}p(x,y,z)\chi_k(f(x,y,z))\bra{\psi}A_x^k\otimes B_y^k \otimes C_z^k\ket{\psi}}.
 \end{align}
And we can define the normalized vectors
 \begin{align}
 \begin{split}
  \ket{\alpha^k}=&\frac{1}{\sqrt{|\inp_1|}}\sum_{x \in \inp_1}{A_x^k}^{\dagger}\otimes \I_{BC} \otimes \I_{\inp_1} \ket{\psi}\ket{x}\\
  \ket{\beta^k}=& \frac{1}{\sqrt{|\inp_2||\inp_3|}}\sum_{x,y \in \inp_2\times \inp_3} \I_{A}\otimes B_y^k \otimes C_z^k  \otimes \I_{\inp_2,\inp_3} \ket{\psi}\ket{y,z}.
 \end{split}
 \end{align}

Now by making use of the game matrices $\Phi_k$ \eqref{phi3} we have the desired result:
 \begin{align}\label{swq3}
  \omega(g_3^{\ell}) &= \frac{1}{|G|}\de{1+\sqrt{|\inp_1||\inp_2||\inp_3|}\sum_{k \in G\setminus \DE{e}}\bra{\alpha^k}\I_{ABC}\otimes \Phi_k \ket{\beta^k}}\nonumber\\
  &\leq  \frac{1}{|G|}\de{1+\sqrt{|\inp_1||\inp_2||\inp_3|}\sum_{k \in G\setminus \DE{e}}\norm{\I_{ABC}\otimes \Phi_k}}\\
  &=  \frac{1}{|G|}\de{1+\sqrt{|\inp_1||\inp_2||\inp_3|}\sum_{k \in G\setminus \DE{e}}\norm{\Phi_k}}\nonumber.
 \end{align}
 \end{proof}

The generalization for $n$-player games is given by the following Theorem.
\begin{theorem}\label{thmnormn}
 Consider an $n$-player linear game,  $g_n^{\ell}(G,f,p)$. Let $S$ be a proper subset
 of the parties, $S \subset \DE{1,\ldots, n}$.
 The quantum value of an $n$-player linear game, $g_n^{\ell}(G,f,p)$, is upper bounded by
 \begin{align}\label{normmulti}
  \omega_q(g_n^{\ell}) \leq \min_{S} \frac{1}{|G|}\de{1+\sqrt{|\inp_1|\ldots |\inp_n|}\sum_{k \in G\setminus \DE{e}}\norm{\Phi^S_k}},
 \end{align}
where $\norm{\Phi^S_k}$ denotes the maximum singular value of matrix $\Phi^S_k$ and the game matrices for partition $S$ are defined as 
\begin{align}
 \Phi^S_k=\sum_{{\vec{x} \in {Q}_S, \vec{y} \in {\inp}_{S^c}}}p(\vec{x},\vec{y})\chi_k(f(\vec{x},\vec{y}))\ketbra{\vec{x}}{\vec{y}}.
\end{align}
$\vec{x} \in {{\inp}_S}$ denotes the vector of inputs of the players that belong to set $S$, and $S^c$ is the complement of $S$.
\end{theorem}

\begin{proof}
%
%
Let $S$ be a proper subset  of the parties $S \subset \DE{1,\ldots, n}$, and the associated game matrices be defined as 
 \begin{align}
 \Phi^S_k=\sum_{\vec{x} \in {Q}_S, \vec{y} \in {Q}_{S^c}}p(\vec{x},\vec{y})\chi_k(f(\vec{x},\vec{y}))\ketbra{\vec{x}}{\vec{y}}.
\end{align}
Now by defining the normalized vectors
  \begin{align}
  \begin{split}
  \ket{\alpha^k}=&\frac{1}{\sqrt{|\inp_{S}|}}\sum_{\vec{x} \in {Q}_{S}}\de{\bigotimes_{i \in S}{{A_i}_{x_i}^k}^{\dagger}}\otimes \I_{S^c} \otimes \I_{\inp_S} \ket{\psi}\ket{\vec{x}}\\
  \ket{\beta^k}=& \frac{1}{\sqrt{|\inp_{S^c}|}}\sum_{\vec{y} \in {Q}_{S^c}} \I_{S}\otimes \de{\bigotimes_{i \in S^c}{A_i}_{x_i}^k}  \otimes \I_{\inp_{S^c}} \ket{\psi}\ket{\vec{y}},
 \end{split}
 \end{align}
 where $|\inp_S|=\prod_{i\in S}|\inp_{i}|$, and ${\inp_S}=\inp_{i_1}\times \ldots \times \inp_{i_k}$ for $i_k \in S$,
 we have that 
\begin{align}\label{swqn}
  \omega(g_n^{\ell}) &= \frac{1}{|G|}\de{1+\sqrt{|\inp_1|  \ldots |\inp_n|}\sum_{k \in G\setminus \DE{e}}\bra{\alpha^k}\I_{A_1\ldots A_n}\otimes \Phi^S_k \ket{\beta^k}}\nonumber\\
  &\leq  \frac{1}{|G|}\de{1+\sqrt{|\inp_1|  \ldots  |\inp_n|}\sum_{k\in G\setminus \DE{e}}\norm{\I_{A_1\ldots A_n}\otimes \Phi^S_k}}\\
  &=  \frac{1}{|G|}\de{1+\sqrt{|\inp_1|  \ldots  |\inp_n|}\sum_{k\in G\setminus \DE{e}}\norm{\Phi^S_k}}\nonumber.
 \end{align}
 
 By the construction of the proof we see that for all subset $S$ we have a valid upper bound to the quantum value. 
 \end{proof}

In Theorem \ref{thmnormn} each partition $S$ of the set of parties provides an upper bound to the quantum value, the minimum in Eq. \eqref{normmulti}
selects the most restrictive one. 
In Definition \ref{defLGmatrices} we have chosen $S=\DE{1}$ for the 3-player game, but  writing 
the game matrices
with $S=\DE{2}$ or $S=\DE{3}$
can lead to tighter bounds than the 
one derived from Eq. \eqref{phi3}. 

%

\subsection*{The computational complexity of our bound}

Theorem \ref{thmnormn} states an upper bound to the quantum value of $n$-player linear games in terms of the spectral norm of the game matrices $\Phi_k^S$,
whose dimension depends on the number of players and the number of questions per player. 
Given an $n$-player linear game $g_n^{\ell}(G,f,p)$ with $m$ questions per player and $d$ possible outcomes, the game
matrices have dimension $m^n$ (where the number of rows and columns depends on the subset $S$ chosen to construct the matrix).
The singular value decomposition of these matrices has time complexity at most\footnote{In Ref. \cite{SVDcomplex}, an algorithm  
for finding  the singular values of an $k \times l$ matrix, $l \leq k$, in time $T(k,l)=\mathcal{O}(2kl^2+2l^3)$ is presented.
Therefore, for the worst case of $S$ containing $n/2$ elements,  $\Phi_k^S$ is an $m^{n/2}\times m^{n/2}$ matrix, 
and $T(n,m,d)=\mathcal{O}(dm^{\frac{3}{2}n}) $.} 
$T(n,m)=\mathcal{O}(m^{\frac{3}{2}n})$ \cite{SVDcomplex}.
Therefore, for a particular subset $S$, an upper bound to the quantum value of game $g_n^{\ell}(G,f,p)$ can be obtained with time complexity
$T(n,m,d)=\mathcal{O}(dm^{\frac{3}{2}n})$. So we see that the complexity increases exponentially with the number of players. Moreover
if we want to obtain the smallest of the upper bounds we would have to run the algorithm an exponential number of times, since
there are $2^{n-1}$ possible subsets\footnote{Note that $\Phi^S_k=(\Phi^{S^c}_k)^T$ which therefore leads to the same bound. } $S$.
Nevertheless, for particular problems (as for example the $n$-player CHSH-$d$ game that we are going to discuss in the next Section)
the bound may be easily calculated analytically by using the symmetries of the game matrix, without the need to perform an explicit numerical calculation.

\section{\texorpdfstring{$n$}{n}-player CHSH-\texorpdfstring{$d$}{d} game}\label{chsh}

In Section \ref{secapplicationsxord} we considered a $d$-output generalization of the CHSH game, for $d$ prime or power of a prime. 
Here we generalize this game for $n$ players following an expression first introduced by Svetlichny \cite{Svetlichny} 
in the context of detecting genuine multipartite nonlocality. 

\begin{definition}The $d$-input and $d$-output per player, $n$-player CHSH game, the {CHSH${_n}$-$d$} game, for $d$ prime or a power of prime, is a
linear game with the winning condition {given by}
\begin{align}\label{eqwinchshnd}
 a_1 + \ldots + a_n= \sum_{i<j} x_i \cdot x_j
 \end{align}
where addition and multiplication are operations defined over the  {finite} field $\mathbb{F}_d$.
\end{definition}

In order to exemplify, let us consider the CHSH${_3}$-$3$ game, where the inputs and outputs are elements
of $\mathbb{Z}_3$, $a,b,c,x,y,z \in \DE{0,1,2}$, and $+,\cdot$ are sum and multiplication modulo 3. The winning condition \eqref{eqwinchshnd}
reduces to
\begin{align}\label{eqwinchsh33}
 a + b+c= x \cdot y+x\cdot z +y\cdot z.
 \end{align}
The game matrix $\Phi_1$ for the  CHSH${_3}$-$3$ is then given by
\begin{align}
  \Phi_1=\sum_{x,y,z=0}^{2} \frac{1}{27} \zeta^{x\cdot y +x\cdot z+y \cdot z}\ketbra{x}{yz},
\end{align}
which is explicitly written as
\begin{align}
  \Phi_1=\frac{1}{27} \begin{bmatrix} 
1&1&1&1&\zeta&\zeta^2&1&\zeta^2&\zeta \\
1 &\zeta &\zeta^2 &\zeta&1& \zeta^2&\zeta^2&\zeta^2&\zeta^2 \\
1 & \zeta^2 &\zeta&\zeta^2&\zeta^2  &\zeta^2&\zeta&\zeta^2&1 \end{bmatrix}.
\end{align}
where $\zeta=e^{2\pi \mathrm{i}/3}$ is a 3rd root of unity.

Now we use Theorem \ref{thmnormn} to prove an  upper bound on the performance of quantum players in the CHSH${_n}$-$d$ game. 
\begin{theorem}\label{thmchshd}
 The quantum value of the CHSH${_n}$-$d$ game, for $d$  prime or a power of a prime, obeys
 \begin{align}\label{boundchsh}
\omega_q(\text{CHSH$_n$-d})\leq \frac{1}{d}+\frac{d-1}{d\sqrt{d}}  .
 \end{align}
\end{theorem}

\begin{proof}
The proof follows from a direct calculation of $\Phi_k^S{\Phi_k^S}^{\dag}$ with the partition $S=\DE{1}$.
The game matrices associated to partition $S=\DE{1}$ are the following:
\begin{align}
 \Phi_k^{A_1}=\sum_{x_1,\ldots,x_n}\frac{1}{d^n}\chi_k(\sum_{j>i}x_i \cdot x_j)\ketbra{x_1}{x_2 \ldots x_n}.
\end{align}

By making use of the characters 
 relations, Eq. \eqref{eqcharacprop}, we have
 \begin{align}
 \begin{split}
   \Phi_k^{A_1}{\Phi_k^{A_1}}^{\dag}  &=\frac{1}{d^{2n}}\sum_{x_1,\ldots,x_n}\sum_{x'_1,\ldots,x'_n}\chi_k(\sum_{j>i}x_i \cdot x_j)\\
   &\;\;\;\; \;\;\;\;\;\;\;\;\;\;\;\;\;\;\;\times \bar{\chi}_k(\sum_{j>i}x'_i \cdot x'_j)\ket{x_1}\braket{x_2 \ldots x_n}{x'_2 \ldots x'_n}\bra{x'_1}\\ 
           & =\frac{1}{d^{2n}}\sum_{x_1,\ldots,x_n}\sum_{x'_1}\prod_{j>1}\chi_k(x_1 \cdot x_j)\bar{\chi}_k(x'_1 \cdot x_j)\ketbra{x_1}{x'_1}\\ 
           &=\sum_{x_1}\frac{1}{d^{n+1}}\ketbra{x_1}{x_1}\\
           &=\frac{1}{d^{n+1}}\I_d,
           \end{split}
 \end{align}
 where in the second step we have used the fact that
 $\chi_k(x_i \cdot x_j)\bar{\chi}_k(x_i \cdot x_j)=1$. 
 Therefore, $\norm{\Phi_k^{A_1}}=\frac{1}{\sqrt{d^{n+1}}}$ for all $k$. Finally, by applying Theorem \ref{thmnormn} we obtain
 \begin{align}
\omega_q(\text{CHSH$_n$-$d$})\leq \frac{1}{d}+\frac{d-1}{d\sqrt{d}}.
 \end{align}
\end{proof}

Interestingly these bounds are independent of the number of parties showing that by increasing the number of
players the performance is still limited.
Note that in order to derive Theorem \ref{thmchshd} we have used $S=\DE{1}$. For the 3-player case all the other possible partitions of the players
would lead to the same result since the winning condition of the game, Eq. \eqref{eqwinchshnd}, is invariant under the permutation of parties. However 
for games with $n>3$ players it is possible that exploring partitions with more players $S=\DE{1,\ldots, r}$ can lead to a better bound.

\section{No quantum realization of non-trivial multiparty functional boxes}
In Ref. \cite{vanDam} it was shown that the possibility of existence of strong correlations known as PR-boxes \cite{prbox} 
would lead to the trivialization of communication complexity. As Alice and Bob would
be able to compute any distributed Boolean function with only one bit of communication, by sharing a sufficient number of PR-boxes. 
{Therefore, the belief that communication complexity is not trivial (\ie, 
some distributed functions require more than one bit of communication in order to be computed)
is viewed as a principle that should be respected by nature.}

As we discussed in Section \ref{secapplicationsxord}, this result was later generalized \cite{PRd} to the so-called 
functional boxes, 
\ie  a $d$-output generalization of PR-boxes for  
$a+b = f(x,y)$ in $\mathbb{Z}_d$ arithmetic, with d prime and $f$ any non-additively separable function
(\ie $f(x,y)\neq f_1(x)+f_2(y)$).
Any functional box which cannot be simulated classically would also lead to trivialization of 
communication complexity \cite{PRd}. 
{Furthermore, a generalization to multiparty
communication complexity scenarios for binary outcome  was considered in Ref. \cite{BP05}. In the multipartite problem, $n$ parties are each 
given an input $x_i$ and must compute a function $f(\vec{x})$ of their joint inputs with as little communication as possible. 
If the parties share a sufficient number of 
boxes
with input-output
relation satisfying $\bigoplus_{i} a_i = \prod_i x_i$, then any $n$-party communication complexity
problem can be solved with only $n-1$ bits of communication (from $n-1$ parties to the first party who then computes the function) \cite{BP05}.

In Appendix \ref{A-CC}, we analyze  the task of computing a multipartite function with $d$ possible values, and show 
that multipartite functional boxes associated to non-trivial  \xor-$d$ games with  uniformly distributed inputs, for $d$ prime, 
lead to a trivialization of communication complexity
in this scenario.
Here we use Theorem \ref{thmnormn} to show that an $n$-party functional box  which maximally 
saturates these games 
cannot be realized in quantum theory.

%

\begin{theorem}\label{stotalfunc}
 For an $n$-player \xor-$d$ game $g_n^{\oplus_d}$, with $m$ questions per player and uniform input distribution
 $p(\vec{x})=1/m^n$, $\omega_q(g_n^{\oplus_d})=1$ iff $\omega_c(g_n^{\oplus_d})=1$.
\end{theorem}

\begin{proof} 
We start by proving the result for 3 players. We first chose the partition $S=\DE{1}$ to write the game matrices.

The constraint that the input distribution is $p(x,y,z)=1/m^3$ for all $x,y,z$ implies $\norm{\Phi^A_k}\leq 1/\sqrt{m^3}$ since $\Phi_k^A{\Phi_k^{A}}^{\dagger}$ is an
$m \times m$ matrix with the absolute value of all elements $\leq 1/m^4$ (and then $\norm{\Phi^A_k{\Phi_k^{A}}^{\dagger}} \leq 1/m^3$).
For the particular case of 3-player \xor-$d$ games, the bound \eqref{norm3} is given by
\begin{align}
  \omega_q(g_n^{\oplus_d})\leq \frac{1}{d}\De{1+\sqrt{m^3}\sum_{k=1}^{d-1}\norm{\Phi_k^A}},
\end{align}
where
\begin{align}
 \Phi^A_1=\sum_{x,y,z}\frac{1}{m^3} \zeta^{f(x,y,z)} \ketbra{x}{yz}.
\end{align}
So we see that 
\begin{align}  
 \omega_q(g^{\oplus_d})=1 \Rightarrow \norm{\Phi_k^A}= \frac{1}{\sqrt{m^3}}\; \forall k.
\end{align}

Let $\ket{\lambda^A}=(\lambda^A_0,\ldots,\lambda^A_{m-1})$ be the eigenvector corresponding to the maximum eigenvalue 
$1/{m^3}$ of $\Phi_1^A{\Phi_1^A}^{\dagger}$.
Consider $ \lambda^A_j=|\lambda^A_j|\zeta^{\theta^A_j}$ and assume $|\lambda^A_0| \geq |\lambda^A_1| \geq \ldots \geq |\lambda^A_{m-1}|$.
Then we have
\begin{align}
\Phi^A_1 {\Phi_1^A}^{\dagger} =& \frac{1}{m^6} \sum_{x,y,z} \sum_{x',y',z'}\zeta^{f(x,y,z)-f(x',y',z')}\ket{x}\braket{yz}{y'z'}\bra{x'}\\
 =& \frac{1}{m^6} \sum_{x,x',y,z}\zeta^{f(x,y,z)-f(x',y,z)}\ketbra{x}{x'} \nonumber
\end{align}
and
\begin{align}
\Phi^A_1 {\Phi^A_1}^{\dagger}\ket{\lambda^A} =&  \frac{1}{m^6} \sum_{x,x',y,z,i}\zeta^{f(x,y,z)-f(x',y,z)+\theta^A_i}|\lambda^A_i|\ket{x}\braket{x'}{i}\\
 =& \frac{1}{m^6} \sum_{x,y,z,i}\zeta^{f(x,y,z)-f(i,y,z)+\theta^A_i}|\lambda^A_i|\ket{x} \nonumber.
\end{align}

By analyzing the first component of the eigenvalue equation ${\Phi^A_1}{\Phi^A_1}^{\dagger}\ket{\lambda^A}=1/m^3 \ket{\lambda^A}$ we have
\begin{align}
 \De{\Phi^A_1{\Phi^A_1}^{\dagger}\ket{\lambda^A}}_1 =& \frac{1}{m^6} \sum_{y,z,i}\zeta^{f(0,y,z)-f(i,y,z)+\theta^A_i}|\lambda^A_i|=\frac{1}{m^3} |\lambda^A_0| \zeta^{\theta^A_0}.
\end{align}
In order to satisfy this equation we need to have 
\begin{subequations}
\begin{align}
 |\lambda^A_i|&=|\lambda^A_0| \; \; \forall \;\;i\\
\zeta^{f(0,y,z)-f(i,y,z)+\theta^A_i}&= \zeta^{\theta^A_0}  \; \; \forall \;\;i, y, z.
\end{align}
\end{subequations}
The equations for the other components of the eigenvalue equation imply: 
\begin{subequations}
\begin{align}\label{sthetaA}
 f(x,y,z)-f(x',y,z)=\theta^A_{x}-\theta^A_{x'} \; \forall\; y,z,
\end{align}
where the operations are modulo $d$.

We can do the same argument for the 
other partitions, $S=\DE{2}$ and $S=\DE{3}$, and the hypothesis of $\omega_q(g_n^{\oplus_d})=1$ implies
that $\rank (\Phi^S_1)=1$ for all $S$, and 
\begin{align}
f(x,y,z)-f(x,y',z)&=\theta^B_{y}-\theta^B_{y'}\; \forall\; x,z\label{sthetaB}\\
f(x,y,z)-f(x,y,z')&=\theta^C_{z}-\theta^C_{z'}\; \forall\; x,y\label{sthetaC}.
\end{align}
\end{subequations}

By the relations \eqref{sthetaA}, \eqref{sthetaB} and \eqref{sthetaC} we deduce that
\begin{align}
 f(x,y,z)&=(\theta^A_{x}-\theta^A_{0})+f(0,y,z)\nonumber\\
 &=(\theta^A_{x}-\theta^A_{0})+(\theta^B_{y}-\theta^B_{0})+f(0,0,z)\\
 &=(\theta^A_{x}-\theta^A_{0})+(\theta^B_{y}-\theta^B_{0})+(\theta^C_{z}-\theta^C_{0})+f(0,0,0).\nonumber
\end{align}
Now, consider $a_0, b_0, c_0$ such that $a_0\oplus_d b_0 \oplus_d c_0=f(0,0,0)$, the classical strategy
\begin{align}
 a&=a_0+(\theta^A_{x}-\theta^A_{0}) \nonumber\\
  b&=b_0+(\theta^B_{y}-\theta^B_{0})\\
   c&=c_0+(\theta^C_{z}-\theta^C_{0})\nonumber
\end{align}
wins the game with probability 1.\vspace{1em}

The proof for an $n$-player game follows in the same way.
Considering the partition $S=\DE{A_1}$.
The constraint of equally distributed inputs implies $\Phi^S_k{\Phi^S_k}^{\dagger}$ is an
$m \times m$ matrix with the absolute value of all elements equal to $1/m^{n+1}$ (and then $\norm{\Phi^S_k{\Phi^S_k}^{\dagger}} \leq 1/m^n$).
Now considering the bound
\begin{align}
  \omega_q \leq \frac{1}{d}\De{1+\sqrt{m^n}\sum_{k=1}^{d-1}\norm{\Phi^S_k}},
\end{align}
we see that  $\omega_q(g^{\oplus_d})=1$ requires $\norm{\Phi^S_k}= 1/\sqrt{m^n}$ for all $k$.

Following the same argument as  for the 3-player game, we conclude that in order to satisfy $\norm{\Phi^S_k}= 1/\sqrt{m^n}$ all 
the rows of the 
game matrix have to be proportional to each other and then
\begin{align}\label{sthetaA1}
 f(x_1,x_2,\ldots,x_n)-f({x'}_1,x_2,\ldots,x_n)=\theta^{A_1}_{x_1}-\theta^{A_1}_{{x'}_1} \;\; \forall\;\; x_2,\ldots x_n.
\end{align}
Running the analysis over the partitions $S={A_i}$ for $i={2,\ldots,n}$ we can specify a classical strategy,
given by $a_{x_i}=a_{0_i}+(\theta^{A_i}_{x_i}-\theta^{A_i}_{0_i})\; \forall i$, that wins the 
game with probability 1. 
\end{proof}

Note that the no-signaling boxes that win some of these $n$-player {\xor-$d$} games 
with probability one correspond to nontrivial functional boxes, 
hence our bound excludes the possibility of {quantum} realization of functional boxes that trivialize
communication complexity in the multiparty scenario.


\section{Device independent witnesses of genuine tripartite entanglement}\label{secdiew}
We now present a systematic way to derive device-independent witnesses for
genuine tripartite entanglement.

The characterization of entanglement is a very challenging task (see Appendix \ref{secentang}).
For bipartite systems, positive maps which are not completely positive constitute a powerful tool for generating
 simple operational criteria for detecting entanglement in mixed states \cite{HorodeckiCPmaps}.
The most celebrated example is the Peres-Horodecki criterion \cite{PeresPPT} (Proposition \ref{PropPPT}), also known as PPT criterion.
The characterization of multipartite entanglement, however, is much more challenging since inequivalent forms of entanglement appear.
{When we are considering correlations among many parties it can happen that all of the parties share quantum correlations, 
in which case we say that 
we have \emph{genuine multipartite entanglement}, or it can be the case that the state is composed by quantum correlations 
only between subsets of the parties, in which case the 
state of the system is said to be biseparable\footnote{This concept can be refined to many levels which is denoted $k$-separability. 
An $n$-partite state
is  $k$-separable (see \cite{ksep} for more detailed definition) if it can be write as a convex combination of states that are products of $k$ subspaces:
\begin{align*}
 \rho_{k-\text{sep}}=\sum_ip_i\ketbra{\psi^i_{k-\text{sep}}}{\psi^i_{k-\text{sep}}}
\end{align*}
where $\ket{\psi^i_{k-\text{sep}}}=\ket{\psi_{i_1}}\otimes\ket{\psi_{i_2}}\otimes \ldots \otimes \ket{\psi_{i_k}}$.
An $n$-partite state is separable if it is $n$-separable.
}.
Formally, a biseparable state of three parties  is a state that can be decomposed into the form
\begin{align}\label{biseparable}
 \rho_{\bi}=\sum_{i}\de{p_A^i\rho_A^i \otimes \rho_{BC}^i+p_B^i\rho_B^i \otimes \rho_{AC}^i+p_C^i\rho_C^i \otimes \rho_{AB}^i},
\end{align}
where {$p^i_j \geq 0$, $\sum_i (p_A^i+p_B^i+p_C^i)=1$}.
If a tripartite quantum state cannot be written as Eq. \eqref{biseparable}, it is
said to be {genuinely tripartite entangled}. }

For the detection of genuine multipartite entanglement there is no known direct criteria  like the PPT-criterion.
For this reason, the development of device-independent entanglement witnesses (DIEW) for genuine multipartite entanglement brings together
with all the advantage of distinguishing the type of entanglement in a
multipartite system,
the possibility of performing this task in
a scenario where we do not have full trust in our devices, which is of extreme importance for cryptographic tasks.

Device independent witnesses of genuine multipartite entanglement were introduced in Ref. \cite{DIEW}. The idea of a DIEW for 
genuine multipartite entanglement is to find the maximal value that a biseparable quantum
state (Eq. \eqref{biseparable} for 3 parties) can achieve in a Bell expression (which in general is in between the classical and the quantum value
of this expression). Therefore, in a Bell experiment, if a quantum state overtakes the biseparable bound we can assert that this state
is genuinely multipartite entangled.

It is important to note that a DIEW is a weaker condition than the Svetilichny inequalities \cite{Svetlichny}.
Svetlichny inequalities were introduced in multipartite Bell scenarios in order
to detect the existence of genuine multipartite nonlocality (see Section \ref{secNLn}). The violation of a Svetlichny inequality guarantees that even if some
parties are allowed to perform a joint strategy (which includes a global quantum measurement in their systems)
they are not able to simulate the exhibited multipartite correlations.
The violation of a DIEW guarantees that 
the parties share a genuinely multipartite entangled state
but not necessarily that they are able to exhibit genuine multipartite 
nonlocality.

Concerning previous works in the subject of DIEWs of genuine multipartite entanglement: A tripartite 3-input 2-output 
inequality which is able to detect genuine
tripartite entanglement in a noisy three-qubit GHZ state, $\rho(V)=V\ketbra{GHZ}{GHZ}+(1-V) \frac{\I}{8}$, for parameter $V> 2/3$,
was presented in Ref. \cite{DIEW}.
This result was later \cite{DIEWPalVertesi} generalized
for multisetting Bell inequalities that in the limit of infinitely many inputs
are able to detect genuine tripartite entanglement for $\rho(V)$ with parameter as low as $2/\pi$, 
which is the limiting value for which there exist a local model
for the noisy GHZ state for full-correlation Bell inequalities \cite{DIEWPalVertesi,DIEW}. Other examples of DIEWs with binary outcomes can be found 
in Ref. \cite{egDIEW}.


Let us consider that Alice, Bob, and Charlie are playing a 3-player linear game, $g_3^{\ell}(G,f,p)$, and they have access to the 
shared biseparable quantum state:
\begin{align}\label{eqbisepC}
 \ket{\psi_{\bi}}=\ket{\psi_{AB}} \otimes \ket{\psi_{C}}.
\end{align}
In that case Alice and Bob can take advantage of a quantum strategy using their shared entangled state $\ket{\psi_{AB}}$. 
However, the best Charlie can do is to apply a classical (deterministic) strategy, since he shares no resources with Alice and Bob.
For this case, their probability of success will be given by
\begin{align}
  \omega(g_3^{\ell}) \leq  \frac{1}{|G|}\de{ 1 + \sum_{x,y,z}\sum_{k\neq\DE{e}}p(x,y,z)
    \chi_k(f(x,y,z))\bar{\chi}_k(c_z) \bra{\psi_{AB}}A_x^k\otimes B_y^k\ket{\psi_{AB}}}, 
 \end{align}
 where $\DE{c_z}$ represents the deterministic strategy performed by Charlie, where upon receiving input $z$ he outputs $c_z$.

 Now we have essentially a bipartite expression to evaluate, and by making use of the
 theorem \ref{thmLGwqbound} we can bound the performance of the players sharing a state biseparable with respect to the partition $AB|C$.
 We denote by  $\omega_{\bi}^C(g_3^{\ell})$ the maximum value the players can achieve in the game $g_3^{\ell}(G,f,p)$ when they share a 
 state of the form \eqref{eqbisepC}:
   \begin{align}\label{normbic}
  \omega_{\bi}^C(g_3^{\ell}) \leq \max_{\DE{c_z}} \frac{1}{|G|}\de{1+\sqrt{|\inp_A||\inp_B|}\sum_{k \in G \setminus \DE{e}}\norm{\Phi^{\bi}_k(c_z)}},
 \end{align}
where
 \begin{align}
\Phi^{\bi}_k(c_z)=\sum_{x,y} \de{\sum_z p(x,y,z)\chi_k(f(x,y,z)-c_z)}\ketbra{x}{y}.
\end{align}

Let us denote by $\omega_{\bi}(g_3^{\ell})$ the maximum probability of success in the game $g_3^{\ell}(G,f,p)$ that can be achieved
with a biseparable state. In Eq.~\eqref{normbic} we have derived an upper 
bound for $\omega_{\bi}^C$, which is the performance  when the players share a quantum state
biseparable with respect to the partition $AB|C$. 
In the case of Bell expressions which are invariant under the permutation of the parties, it is sufficient to consider Eq.~\eqref{normbic} {in order to bound $\omega_{\bi}$}.
For general Bell inequalities, {an upper bound on }the biseparable bound, that holds for any state of the form given by Eq.~\eqref{biseparable}, 
can be obtained by taking the maximum 
over all  bipartitions{, since
\begin{align}
  \omega_{\bi}=\max_{X} \omega_{\bi}^X,
\end{align}
where $X \in\DE{A,B,C}$.}

 In general we have
 \begin{align}
  \omega_c \leq \omega_{\bi} \leq \omega_q
 \end{align}
 and then, for games where the strict relation $\omega_{\bi} < \omega_q$ holds, by violating the 
 biseparable bound $\omega_{\bi}$ we can certify in a device-independent way that Alice, Bob and Charlie
 share a genuine tripartite entangled quantum state. 

By using our norm bounds, Eq. \eqref{normbic}, we have a simple way to upper bound the biseparable value of a Bell expression. 
Hence, our techniques of bounding the quantum value of linear games open the possibility of exploring higher dimensional
device-independent witnesses of genuine tripartite entanglement.

We now exemplify the method by deriving a DIEW from a 3-input 3-output tripartite Bell expression.

\subsection*{Example:}

Inspired by the Mermin's inequality \cite{Mermin} for the GHZ paradox, we consider the following game that we denote $g_3^{\text{Mermin}}$:
A referee picks questions $x, y, z \in \DE{0,1,2}$ with the promise that $x+ y + z =0$. 
The players are supposed to give answers $a,b,c \in \DE{0,1,2}$ in order to satisfy
\begin{align}\label{eqgameghz}
g_3^{\text{Mermin}}: a + b + c=x\cdot y \cdot z\,, \;\;\; \text{s.t.}\;\;\; x+ y + z=0
\end{align}
where the operations $+$, $\cdot$ are sum and multiplication modulo 3.

Substituting the winning condition of the game, Eq. \eqref{eqgameghz}, into Eq. \eqref{normbic} (note that the constraints of the game
are invariant under the 
permutation of the parties), allows us to derive 
\begin{align}
 \omega_{\bi} \leq 0.896,
\end{align}
with RHS approximated up to the third decimal.

On the other hand, the $GHZ_3$ state defined as
\begin{align}
 \ket{GHZ_3}=\frac{\ket{000}+\ket{111}+\ket{222}}{\sqrt{3}},
\end{align}
can win the game $g_3^{\text{Mermin}}$ with probability $1$. 

Consequently, we have a device-independent witness 
of genuine tripartite entanglement: $\omega_{\bi}(g_3^{\text{Mermin}}) \leq 0.896$.

The explicit measurements that lead the $GHZ_3$ to reach value one in the $g_3^{\text{Mermin}}$ game
are presented in Appendix \ref{A-proofs}. However, an important property of these measurements is that they
define traceless `observables' (by the relation $ A_x^i=\sum_a\bar{\chi}_i(a)M_x^a$):
 \begin{align}
  \Tr\de{A_x^i\otimes B_y^j \otimes C_z^k}=0.
 \end{align}
Therefore, we can easily calculate the success probability one can achieve with a noisy $GHZ_3$ state
\begin{align}
 \tilde{\rho}(V)=V\ketbra{GHZ_3}{GHZ_3}+(1-V) \frac{\I}{27},
\end{align}
when
 these measurements are performed:
 \begin{align}
  \omega(\tilde{\rho}(V))& = \frac{1}{3}\de{1+\sum_{x,y,z}\sum_{k=1}^2 {p(x,y,z)} \zeta^{k\cdot f(x,y,z)}\Tr \de{\tilde{\rho}(V)(A_x^k\otimes B_y^k\otimes C_z^k)}}\nonumber\\
		&= \frac{1}{3}+\frac{V}{3}\sum_{x,y,z}\sum_{k=1}^2 {p(x,y,z)} \zeta^{k\cdot f(x,y,z)} \Tr \de{\ketbra{GHZ_3}{GHZ_3}A_x^k\otimes B_y^k\otimes C_z^k}\nonumber\\
		&=\frac{1-V}{3}+V\omega(\ketbra{GHZ_3}{GHZ_3})\\
		&=\frac{1+2V}{3}. \nonumber
 \end{align}
Hence, by using the optimal measurements for the $GHZ_3$ state we are able to witness genuine 
 multipartite entanglement on noisy $GHZ_3$ state for parameter $V> 0.85$.

This example also stresses the difference between genuine multipartite entanglement and genuine multipartite nonlocality.
Since the Svetlichny bound
\cite{Svetlichny} for this game is $1$, it cannot be used as a witness of genuine 
tripartite nonlocality, despite being a good witness for  genuine tripartite entanglement.

The 3-input 2-output witness presented in Ref. \cite{DIEW} can detect genuine multipartite entanglement in $\tilde{\rho}(V)$ for $V>
0.81$, however the two-outcome DIEW involves the calculation of 18 expected values whereas for the
$g_3^{\text{Mermin}}$ game, only 9 expected values are involved.


\section{Discussion and open problems} 
In this Chapter we have extended the bound presented in Chapter \ref{chaptergamesd} to 
$n$-player linear games. 
As for the bipartite case, the bound derived for multiplayer  games is
not tight in general, nevertheless we could apply it to derive non-trivial results.
It would be interesting to characterize the classes of linear games for which the norm bound is tight. This could lead us to 
 understand 
a bit more about the structure of the set of quantum correlations. In addition, the characterization of multiplayer games with no quantum
advantage is interesting to highlight the limitations of quantum theory in the multipartite scenario.
For multipartite Bell scenarios it is known that the quantum set of correlations contains facets \cite{GYNI}, it would be interesting
to investigate if a multiplayer linear game can constitute a facet of the quantum set.

By using Theorem \ref{thmnormn}, we have proved upper bounds to the quantum value of a multipartite generalization of the CHSH-$d$ game, for 
any number of players. 
Also, we have shown that boxes that trivialize communication complexity in the multipartite scenario cannot be realized in 
quantum theory. An important question for further investigation would be to 
analyze the
relation between the bounds on linear games (bipartite and multipartite)
we presented here and the communication complexity of the associated functions.



Finally, we  {have presented} a systematic way to derive device independent witnesses of genuine tripartite entanglement.
The method is very general and can be applied to any tripartite linear game in order to derive a  DIEW of genuine tripartite entanglement with many 
inputs and outputs.
We  {exhibited} an example were a DIEW involving only $9$ expected values is able to detect genuine tripartite entanglement
in a  {noisy} $GHZ_3$ state.
It remains  {an} open point  {whether these} DIEWs with $d$ outcomes are optimal in terms of the number of inputs 
to detect genuine tripartite
entanglement of $d$-dimensional systems.
The search for optimal witnesses with few inputs per player can lead to feasible  {applications} and experiments.


It is important to stress that the bound derived in Theorem \ref{thmnormn} involves the norm of matrices, which is an object with 
 an intrinsic bipartite structure. 
 A possible future direction would be to  {explore the use of tensors, which have} a natural multipartite 
 structure, in order to describe the games.

%% file: conclu.tex
\chapter*{Final remarks}\label{conclu}
\addcontentsline{toc}{chapter}{Final remarks}
\markboth{Final remarks}{}

In this Thesis we have studied nonlocality focusing on the quantum value of a Bell expression.
In particular, we took the approach of stressing how difficult is the problem of determining the quantum value of a Bell expression 
by situating it in the framework of 
computational complexity. 
In this framework, the main result of this Thesis can be summarized as: \textit{efficiently computable upper bounds to the 
quantum value of linear games (a particular class of Bell inequalities).}

Those bounds are based on the spectral norm of some matrices associated to the games. 
The major advantage of our bounds is that they are easy to compute analytically for games with a 
small number of inputs and they can 
be implemented numerically by efficient algorithms. The drawback, however, is that the bounds are not tight in general
and the quality of approximation
is not known. Besides that, we could explore the bounds  deriving several non-trivial results. 

For the case of {\xor} games, the SDP characterization of the quantum value, due to Tsirelson's theorem, allowed us
to derive necessary and sufficient conditions 
for a  game to have no quantum advantage. This led to the characterization of a new family of graphs for which the 
Shannon capacity is equal to the independence number.

Concerning linear games with more outputs we could easily re-derive a recently discovered upper bound to the quantum value of the
CHSH-$d$ game, and we also extended it to the case of $n$ players, the CHSH$_n$-$d$ game. Also, we have shown that the norm bounds can exclude 
the existence of some boxes that would lead to trivialization of communication complexity in the bipartite and the multipartite case.

Furthermore, we defined an extension of the no-advantage for nonlocal computation
principle, introduced in Ref.~\cite{NLC}, to functions with $d$ prime possible values. And finally, as the main outcome of exploring the norm bounds in the 
multipartite scenario, we presented
a systematic way of deriving device-independent witnesses of genuine multipartite entanglement for tripartite systems.

In conclusion, we had a glance in the intricacy of quantum nonlocality by phrasing the problem of finding the maximal 
quantum violation of a Bell inequality in the framework of computational complexity. At the end of the day, I hope 
I have convinced the reader that finding simple bounds to this problem is a good game to play,
or, at least, 
I hope the reader enjoyed it!

%% file: MQ.tex
\part*{Appendix}
\addcontentsline{toc}{part}{Appendix}
\appendix


\chapter{Quantum Mechanics}\label{Aquantum}
In this Appendix we give an overview of Quantum theory, introducing the main concepts and properties  
discussed along the text.
For an excellent introduction to the quantum theory formalism with an information theoretic approach the reader is referred to the book of Nielsen and Chuang
\cite{N&C}.

\section{A few concepts and definitions}

In the study of nonlocality we usually adopt a minimalistic view of quantum theory: We do not exploit 
particular interactions among systems, or how is the dynamics of the systems and so on, other than that,
we only capture the fundamental features
and consequences of the mathematical formulation of quantum theory. 

\subsection{Concepts and axioms}
Every quantum system has an associated complex Hilbert space $\Hi$ (\ie a complete vector space with inner product). The system is mathematically described by a 
\emph{quantum state}, which is an operator acting on $\Hi$.
The quantum state contains all the information necessary to predict the statistics of the results of measurements performed in the system.

\begin{subdefinition}{definition}\label{defstates}
\begin{definition}[Quantum state]
A quantum state (also called density operator) is an operator $\rho \in D(\Hi)$ with the following properties:
\begin{enumerate}[(i)]
 \item $\rho \geq 0$,
 \item $\Tr \rho =1$. 
\end{enumerate}
We denote the set of positive trace-1 operators acting on $\Hi$ by $D(\Hi)$.
\end{definition}

A very important particular case of quantum states are the \emph{pure states}. 

\begin{definition}[Pure state]
A pure state  is a quantum state of rank 1, which means that they are one-dimensional
projections\footnote{Note that the two vectors $\ket{\psi}$
and $e^{i\phi}\ket{\psi}$ give rise to the same density operator, therefore pure quantum states are defined up
to a global phase.}, $\rho=\ketbra{\psi}{\psi}$.
A pure state can be represented by a vector $\ket{\psi} \in \Hi$.
\end{definition}
\end{subdefinition}

If a state is not pure we call it a \emph{mixed state}.\vspace{1em}

Given that we know the quantum state of a physical system, quantum theory tells us how to predict the statistics of results of any experiment
performed in
the system. An experiment is a question that we are asking about our system and each experiment has its set of expected outcomes 
(\eg when tossing a coin we have two possible outcomes: head or tail). In the quantum theory formalism, the experiments
are described as following:

\begin{subdefinition}{definition}
\begin{definition}[Quantum measurements]
A quantum measurement is described by a `Positive Operator-Valued Measure' POVM, \ie a set of $m$ operators acting on $\Hi$, $\DE{M_i}_{i=1}^{m}$ 
such that
\begin{enumerate}[(i)]
 \item $M_i\geq0\;\;\forall\; i=\DE{1,\ldots,m}$,
 \item $\sum_{i=1}^{m}M_i=\I$,
\end{enumerate}
The $M_i$'s are denoted elements of POVM and $m$ is the number of possible outcomes of the measurement.
\end{definition}

A particular case of POVM are the \emph{Projective measurements}.

\begin{definition}[Projective measurements]
 A projective measurement is the particular case of measurements where the all the POVM elements, $\DE{\Pi_i}_{i=1}^{m}$, are orthogonal projectors 
 \ie 
 \begin{enumerate}
  \item $\Pi_i \Pi_j=\delta_{i,j}\Pi_i$,
  \item $\sum_{i=1}^{m} \Pi_i=\I$.
 \end{enumerate}
\end{definition}
\end{subdefinition}
\vspace{1em}

The recipe for obtaining the probabilities of outcomes of experiments given these mathematical objects is state by Born's rule, introduced
by Max Born in \cite{Bornrule}.

\begin{definition}[Born rule]
If an $m$-outcome measurement, described by the POVM set $\DE{M_i}_{i=1}^{m}$, is performed in a system described by a quantum state $\rho$, then the
probability of obtaining the outcome $k$ is given by
\begin{subequations}
\begin{align}
 P(k)=\Tr \de{M_k \rho}.
\end{align}
And for the particular case of projective measurement  $\DE{\Pi_i}_{i=1}^{m}$ in a pure state $\ket{\psi}$ the probability of 
obtaining the outcome $k$ is given by
\begin{align}
 P(k)= \bra{\psi}\Pi_k\ket{\psi}.
\end{align}
\end{subequations}
\end{definition}

\subsection{Composite systems}

When dealing with more than one system, or systems with many degrees of freedom, we have to establish a way to describe these systems in the 
formalism of quantum theory.

Given two single systems, $A$ and $B$, with respective associated Hilbert spaces $\Hi_A$ and $\Hi_B$, the composite system $AB$ has an 
associated Hilbert space:
\begin{align}\label{eqHiAB}
 \Hi_{AB}=\Hi_A \otimes \Hi_B.
\end{align}

Now, a very interesting phenomenon arises from this mathematical structure: the quantum states associated with system $AB$ are still described 
by Definition \ref{defstates} and hence every positive operator of trace one acting on $\Hi_{AB}$ is an allowed quantum state. But \emph{not} all 
quantum states of $AB$ have the \emph{product} structure:
\begin{align*}
 \rho_{AB}=\sigma_A\otimes \sigma_B \;\; \text{where}\;\; \sigma_A \in D(\Hi_A)\;\; \text{and}\;\; \sigma_B \in D(\Hi_B),
\end{align*}
and hence we have a rich structure arising from multipartite systems.

 A quantum state of a bipartite system $AB$ is a \emph{separable state} if it can be approximated by a state of the form 
 \begin{align} \label{eqsepapp}
  \sigma_{AB}=\sum_i q(i)\sigma_A^i \otimes \sigma_B^i,
 \end{align}
 where $q(i)$ is a probability distribution, and $\sigma_A^i \in D(\Hi_A)$, $\sigma_B^i \in D(\Hi_B)$ are 
 quantum states of single systems. Formally:

\begin{definition}[Separable states]\label{defsep}
 A quantum state $\rho_{AB}$ is separable if for every $\epsilon >0$, there exist $N(\epsilon) \in \N$ such that
  \begin{align} \label{eqsep}
\NORM{ \rho_{AB}-\sum_{i=1}^{N} q(i)\sigma_A^i \otimes \sigma_B^i}_1<\epsilon
 \end{align}
 where $\norm{X}_1=\frac{1}{2}\Tr |X|$ is the trace norm.
 \end{definition}
 
  Given a system with Hilbert space $\Hi_{AB}$, the set of all separable states  is denoted $SEP$, and it is
 a closed convex set. For finite dimensional systems, $\Hi_{AB}=\C^d \otimes \C^d$, Carathéodory's theorem
 guarantees that every separable state can be written as a convex combination of at most $d^2+1$ pure separable states (\ie states of the form
 $\ketbra{\psi_A}{\psi_A}\otimes\ketbra{\psi_B}{\psi_B}$), and therefore every finite dimensional separable state has the for of Eq. \eqref{eqsepapp}.

If a quantum state is not separable, \ie if it  cannot be approximated by a separable decomposition \eqref{eqsep}, it is called an \emph{entangled state}.
Entanglement is in the core of the novelties brought by quantum theory. It gives rise to totally new phenomena with no classical analogue.
As a remarkable example: the existence of entangled states gives rise to the quantum nonlocality.

\subsection*{Reduced state}

When in possession of a multipartite system we might be interested in describing only part of it (or only few degrees of freedom). For that task 
we have the concept of \emph{reduced state}.

\begin{definition}[Reduced state]
Given a bipartite system in state $\rho_{AB}$, the reduced state of system $A$ is given by
 \begin{align}\label{eqredstate}
 \rho_A=\Tr_B \rho_{AB},
\end{align}
and analogously for the reduced state of system $B$. $\Tr_B$ denotes partial trace\footnote{The partial trace is a linear map, 
$\Tr_B:D(\Hi_{AB})\longrightarrow D(\Hi_A)$, defined by
\begin{align*}
 \Tr_B (\ketbra{a_1}{a_2}\otimes \ketbra{b_1}{b_2})= \ketbra{a_1}{a_2}\Tr (\ketbra{b_1}{b_2}),
\end{align*}
and extended by linearity.} with respect to subsystem $B$.
$\rho_A$ is a positive operator with trace one acting on $\Hi_A$, and hence is a quantum state of system $A$.
\end{definition}

The reduced state is the best description we can give for a subsystem, in case we ignore completely what is happening with the other subsystems.
It is sufficient to describe the statistics of all local measurements (measurements performed only in the subsystem). If we have a composite
system $AB$ and we want to perform a POVM $\DE{{M_A}_i}_{i=1}^m$ in the system $A$, this is equivalent to performing the 
measurement $\DE{{M_A}_i\otimes \I}_{i=1}^m$ in the system $AB$, and the probability of getting an outcome $k$ is given by
\begin{align}
 p(k)=\Tr_{AB} \De{\de{{M_A}_k\otimes \I}\rho_{AB}}=\Tr_{A} \de{{M_A}_k \rho_A}.
\end{align}

A very interesting fact is that the complete knowledge of the state of the parties do not allow us to recover the global state of the system, as in general
\begin{align}
 \rho_{AB} \neq \rho_A \otimes \rho_B.
\end{align}
For a separable state of the form $\rho_{AB}=\sum_i q(i)\sigma_A^i \otimes \sigma_B^i$ we have that
\begin{align}\label{eqpartialsep}
 \rho_A= \Tr_B \rho_AB= \sum_i q(i)\sigma_A^i,
\end{align}
but the reduced state \eqref{eqpartialsep} does not uniquely recover the global state $\rho_{AB}=\sum_i q(i)\sigma_A^i \otimes \sigma_B^i$, and
it could even correspond to the reduced state of a global entangled state.

\section{Entanglement theory}\label{secentang}

We have presented the definition of an entangled state as the one which is not separable. However given a quantum state, 
checking whether or not it
satisfies Eq. \eqref{eqsep} is not an easy problem. 
In this section we briefly discuss entanglement detection and quantification.

\subsection{Entanglement criteria}

In general, given a finite dimensional quantum state, it is hard to conclude whether or not it can be written as
 Eq. \eqref{eqsepapp}.
For pure states, however, the situation is much simpler, and we just have to look at the rank of the reduced state.
\begin{proposition}\label{proppureentang}
 A pure state $\ket{\psi}_{AB}$ is separable iff $\;\rank (\rho_A) =1$.
\end{proposition}

Proposition \ref{proppureentang} follows from the fact that $\rho_{AB}=\ketbra{\psi_{AB}}{\psi_{AB}}$ is a $\rank$ one operator, 
and if it is separable it has to be written
as the tensor product of $\rank$ one operators in $D(\Hi_A)$ and $D(\Hi_B)$.

For mixed states the situation is way harder, since a reduced state with $\rank$ greater than one leads us to no
conclusion. Nevertheless it is possible to derive simple criteria that gives sufficient conditions
for a state to be entangled. One of the most remarkable of these criteria was introduced by Peres in Ref. \cite{PeresPPT}. Let us consider the 
\emph{partial transposition} map, which is a linear map 
\begin{align}
 T\otimes I: D(\Hi_{AB})&\longrightarrow D(\Hi_{AB})\\
 \rho_{AB}&\longmapsto \rho_{AB}^{\Gamma}, \nonumber
\end{align}
where $\rho_{AB}^{\Gamma}$ denotes the partial transposition of state $\rho_{AB}$.
The action of $T\otimes I$ is to transpose the matrix $\rho_{AB}$ with respect to subsystem $A$
($T\otimes I(\ketbra{ab}{a'b'})=\ketbra{a'b}{ab'}$ and extended by linearity).

Let us look at the application of this map to a separable state \eqref{eqsepapp}:
\begin{align}
 \rho_{AB}^{\Gamma}=\sum_i q(i)(\sigma_A^i)^T \otimes \sigma_B^i,
\end{align}
where $T$ denotes transposition. Since transposition is a trace preserve positive map we have that $(\sigma_A^i)^T$ are quantum states, 
hence $ \rho_{AB}^{\Gamma}$ is also a quantum state (\ie it is a positive operator with trace one). 

So we have  seen that the application of the 
partial transposition map into a separable state results in a positive operator. But that is not true for every quantum state! Let us consider
the maximally entangled 2-qubit state $\Phi=\ketbra{\Phi}{\Phi}$, $\ket{\Phi}=1/\sqrt{2}(\ket{00}+\ket{11})$. A direct calculation shows that
\begin{align}
 \Phi^{\Gamma} \ngeq 0.
\end{align}
And we have just demonstrated our first entanglement criteria:

\begin{proposition}[The PPT criteria]\label{PropPPT}
If $\rho_{AB}^{\Gamma} \ngeq 0$ then $\rho_{AB}$ is entangled. 
\end{proposition}

In Ref. \cite{HorodeckiCPmaps} the Horodeccy showed that the PPT criteria (also know as Peres-Horodecki criteria) is
a necessary and sufficient condition for systems with Hilbert space
$\C^2 \otimes \C^2$ or $\C^2 \otimes \C^3$, but only a sufficient condition for higher dimensional systems.

An interesting feature of the transposition is that it is a positive map that is not completely positive\footnote{A 
completely positive map is a positive map $\Lambda$ such that any trivial extension $\Lambda \otimes \I$ is also a positive map.}.
In Ref. \cite{HorodeckiCPmaps} it was shown that this is the main property  that leads to an entanglement criteria. 
The same argument as above, applied to a generic  map, implies that every separable state should remain positive under the application of a 
generic positive but not completely positive linear map. Hence each non-completely positive map can give rise to an entanglement criteria (a 
sufficient condition for a state to be entangled). But the breakthrough of Ref. \cite{HorodeckiCPmaps} comes
with the proof that actually the formalism of positive but not 
completely positive maps can fully characterize the set of bipartite entangled states.

\begin{theorem}[\cite{HorodeckiCPmaps}]
A quantum state $\rho_{AB} \in D(\Hi_A \otimes \Hi_B)$ is separable iff
\begin{align}
 \Lambda \otimes \I (\rho_{AB}) \geq 0
\end{align}
for any positive map $\Lambda: D(\Hi_A)\longrightarrow D(\Hi_B)$. 
\end{theorem}

\subsection{Entanglement quantification}

Let us start by the two most common operational ways to quantify entanglement: \emph{distillable entanglement}, $E_D$, and \emph{entanglement cost}, $E_C$.
In the operational paradigm the unit of bipartite entanglement is the maximally entangled two qubit state,  $\ket{\psi^-}=\frac{1}{2}(\ket{01}-\ket{10})$.
And the operations that are considered free are the \emph{local operations and classical communication} (LOCC).
Hence the $E_D$ and $E_C$ are defined in terms of how many resources, $\ket{\psi^-}$ states, one need or one can obtain out of a state $\rho$, when only
LOCC operations are applied.

The paradigm of LOCC as free operations in entanglement theory is justified by the fact that entanglement cannot be created if only local operations and 
classical communication are available. LOCC operations can only generate a separable states \eqref{eqsep}. Conversely, any
separable state can be created using LOCC.

Let us represent a general protocol (\ie a set of maps/operations that take quantum states into quantum states) by $\Lambda$.

\begin{definition}[Distillable entanglement]
A distillation protocol is an LOCC
 map that takes a certain number $n$ of copies of a state $\rho$ and turn it into another number, $r_dn$, of copies of $\ket{\psi^-}$:
 \begin{align}
  \rho^{\otimes n} \xrightarrow{\Lambda_{LOCC}} \Lambda_{LOCC}(\rho^{\otimes n})\approx \ketbra{\psi^-}{\psi^-}^{\otimes r_dn}
 \end{align}
where $\approx$ means that the states are asymptotically close in trace distance:
\begin{align}
\norm{\Lambda_{LOCC}(\rho^{\otimes n})-\ketbra{\psi^-}{\psi^-}^{\otimes r_dn}}\xrightarrow{n \longrightarrow \infty}0.
\end{align}
The distillable entanglement is defined as the best rate at which one can distil singlets out of the state $\rho$:
 \begin{align}
  E_D=\sup_{\Lambda_{LOCC}} r_d .
 \end{align}
\end{definition}

In a similar way we have the cost of entanglement.
\begin{definition}[Entanglement cost]
A creation protocol is a
 map that takes a certain number $r_cn$ of copies of the singlet state $\ket{\psi^-}$ and turn in into  $n$ copies of the state $\rho$:
 \begin{align}
  \ketbra{\psi^-}{\psi^-}^{\otimes r_cn} \xrightarrow{\Lambda_{LOCC}} \Lambda_{LOCC}(\ketbra{\psi^-}{\psi^-}^{\otimes r_cn})\approx \rho^{\otimes n}
 \end{align}
where $\approx$ means that the states are asymptotically close in trace distance:
\begin{align}
\norm{\Lambda_{LOCC}(\ketbra{\psi^-}{\psi^-}^{\otimes r_cn})-\rho^{\otimes n}}\xrightarrow{n \longrightarrow \infty}0.
\end{align}
The entanglement cost is defined as the best rate at which one can generate the state $\rho$:
 \begin{align}
  E_C=\inf_{\Lambda_{LOCC}} r_c .
 \end{align}
\end{definition}

In general $E_D \leq E_C$, with equality carrying a meaning of reversibility. For pure states $E_D = E_C$, but for mixed states the strict 
relation can hold $E_D < E_C$. The entanglement cost is always strictly positive for an entangled state, however the distillable entanglement 
can be zero even though
the state is entangled. There exist entangled states which cannot be distillable \cite{PPTbound}, these are called \emph{bound entangled} states. 
In particular it is known that states which do not violate the PPT criteria (Proposition \ref{PropPPT}), \ie states which are positive under partial transposition
(that we call PPT states), cannot be distilled \cite{PPTbound}. Whether these are the only non-distillable states, or if there exists NPT states (states which 
are not positive under partial transposition) for which $E_D=0$, is 
one of the big open problems in quantum information theory.
\vspace{1em}

\subsection{Multipartite entanglement}\label{Amultipartiteentang}

Multipartite systems have a much richer structure them the bipartite ones. 
Already in the tripartite case we have example of two inequivalent maximally entangled states. The states:
\begin{align}
 \ket{W}&=\frac{1}{\sqrt{3}}(\ket{001}+\ket{010}+\ket{100})\\
 \ket{GHZ}&=\frac{1}{\sqrt{2}}(\ket{000}+\ket{111})
\end{align}
cannot be taken one into the other by any SLOCC protocol (\ie an LOCC protocol that 
has a probability of success)\cite{WGHZ}, whereas in the bipartite case, all the maximally entangled states are equivalent.

The generalization of separability for multipartite systems is straightforward. A finite dimensional\footnote{Separability of
infinite dimensional systems are defined analogously to Eq. \eqref{eqsep}.} $n$-partite state is \emph{separable} if it can be written as:
\begin{align}
 \rho=\sum_i p_i \rho_{A_1}^{(i)}\otimes \rho_{A_2}^{(i)}\otimes \ldots \otimes \rho_{A_n}^{(i)}\;\;\st \;\;\;\; p_i\geq 0\;,\; \sum_ip_i=1.
\end{align}
But when it concerns entanglement in a multipartite system the situation is not so simple. It can be the case 
that all the parties share quantum correlations, in which case the state is said to be \emph{genuinely
multipartite entangled} (GME), or else, we can have a combination of states where only  subsets of the parties share entanglement, which is
called a \emph{biseparable state}.

\begin{definition}[Biseparable state]
Let $S\subset\DE{1,\ldots, n}$. An $n$-partite state $\rho_{\bi}$ is biseparable if it can be decomposed into the form:
\begin{align}\label{eqbiseparable}
 \rho_{\bi}=\sum_{S\neq \emptyset}\sum_{i_S} p_{i_S} \rho_{A_S}^{(i_S)}\otimes \rho_{A_{S^C}}^{(i_S)},
\end{align}
where $\sum_{S\neq \emptyset}\sum_{i_S} p_{i_S}=1$, $p_{i_S}\geq 0$, and $S^C$ is the complement of $S$.
\end{definition}
A state which is not biseparable is genuinely multipartite entangled.
One can also define intermediate classes between biseparable and fully separable states, 
which are basically characterized by the number of subsystems that share entanglement.

In the tripartite case we have three possible bipartitions of the set of parties:
$A|BC$, $B|CA$ and $C|AB$, and a general biseparable state is written as:
\begin{align}\label{Aeqbiseparable}
 \rho_{\bi}=\sum_{i_A} p_{i_A} \rho_{A}^{(i_A)}\otimes \rho_{BC}^{(i_A)}+ \sum_{i_B} p_{i_B} \rho_{B}^{(i_B)}\otimes \rho_{CA}^{(i_B)}+\sum_{i_C} p_{i_C} \rho_{C}^{(i_C)}\otimes \rho_{AB}^{(i_C)}.
\end{align}

In order to explore a multipartite scenario in its full extent, one is usually interested in having a GME state. 
Detection of genuine multipartite entanglement is a fruitful area of research, but, for this task,
there is no such direct criteria like the PPT-criteria. A connexion between 
positive maps and witnesses of genuine multipartite entanglement was established in Ref. \cite{gmePositMaps}. There, a framework 
to construct witnesses of genuine multipartite entanglement from positive maps was derived. Other criteria 
to detect genuine multipartite entanglement were proposed in Refs. \cite{Marcus1,Marcus2}. And a device-independent
witness of GME, based on Bell inequalities
was proposed in \cite{DIEW}.
A systematic method to derive device-independent witnesses of genuine tripartite entanglement
was presented in Chapter \ref{chapternplayer}.


%% file: Absstatebound.tex
\chapter{State dependent bounds}\label{chapterstatebound}

In this Appendix we briefly describe the main results of Ref. \cite{GMstatebound}:
\begin{center}
\begin{minipage}{12.5cm}
\item \textit{Bounds on quantum nonlocality via partial transposition}\\
   K. Horodecki and \textbf{G. Murta}\\
 \href{http://link.aps.org/doi/10.1103/PhysRevA.92.010301}{\textbf{Phys. Rev. A, \textbf{92}, 010301, (2015)}}.
\end{minipage}
\end{center}

The quantitative study of quantum nonlocality has two opposite approaches: One is to ask, for a fixed Bell scenario, 
what is the highest violation one can obtain 
optimizing over all possible quantum resources 
(states and measurements);
Another is to fix the quantum state, or a class of states, and ask what is the best one can achieve using this state as a nonlocal resource, 
\ie optimizing over all Bell scenarios. 
In this thesis we have focused on the first approach, deriving bounds on the quantum value of a particular classes of Bell inequalities: the linear 
games.
In Ref. \cite{GMstatebound}   we have taken the opposite direction and asked 
\textit{`How much nonlocality one can extract from a particular quantum state?'}.


\section{Bound on single copy nonlocality}

In this section we consider a standard Bell scenario where Alice and Bob share a single copy
of a quantum state  $\rho_{AB}$ and perform local measurements on it, and we want to bound, for an  arbitrary Bell inequality, 
the violation Alice and Bob can achieve by using  $\rho_{AB}$.

Some previous results in this direction include the seminal work of Werner \cite{Werner} showing that some entangled quantum 
states cannot violate any Bell inequality\footnote{The result of Werner \cite{Werner} concerns only projective measurements.
It was extended for POVMs in Ref. \cite{Barrett}.}.
Another general result shows that typically the violation of correlation Bell inequalities by multipartite qudit states is very 
small \cite{Drumond2012}. And also an hierarchy of semidefinite programs that allows one to bound the violation achievable
by a PPT state was developed in Ref. \cite{hierarchyPPT}. 

Our goal is to derive a state dependent bound for the violation of a particular Bell inequality. 
In order to achieve this goal, we explore the link between two concepts: the level of violation of a Bell inequality by a quantum state 
and discrimination between two states by means of a restricted classes of operations.
Note that since only entangled quantum states can exhibit nonlocality, Bell inequalities can be viewed as particular cases of entanglement 
witnesses \cite{Hyllus-Bell-wit,Terhal-Bell-wit}. Moreover, the test of a  Bell inequality can be seen as the application of a 
separable operator to the quantum state.
Therefore, we can say that a Bell inequality is an entanglement witness which only involves a restricted class of operations, the 
separable ones.

\begin{definition}\label{defsepoperator}
A separable operation is a quantum operation that can be written in the form: 
\begin{align}\label{eqsepoperator}
\Lambda_{sep}(\rho)= \sum_i K_i \rho K^{\dagger}_i  \st \;\;K_i={K_A}_i\otimes {K_B}_i.
\end{align}
\end{definition}

The class of separable operations includes the LOCC operations, and they are a subset of a larger class called $PPT$ operations.

\begin{definition}
A PPT operation is a quantum operation that can be written in the form: 
\begin{align}\label{eqPPToperator}
\Lambda_{PPT}(\rho)= \sum_i K_i \rho K^{\dagger}_i  \st \;\;(K^{\dagger}_iK_i)^{\Gamma}\geq 0.
\end{align}
\end{definition}

It is easier to impose the constraint that an operator is PPT compared to imposing separability. For this reason, PPT operations are used many times to upper
bound results concerning separable and LOCC operations\cite{Duan-discrim}.

The first surprising result concerning distinguishability of quantum states by a restricted class of operations 
was obtained in Ref. \cite{Bennett-nlwe}, where the authors
showed the existence of a set of separable orthogonal states which cannot be perfectly distinguishable by any sequence of LOCC operations.
Later, it was shown that there exist pairs of states which are hardly distinguishable from each other by means of LOCC, 
although being almost orthogonal. This gave rise to the \emph{quantum data hiding}, which is the task of hiding classical bits in a
quantum state \cite{hiding-prl,hiding-ieee,WernerHide}. 
In Ref. \cite{karol-PhD} it is shown that there exist even entangled states containing a bit of private key, which are almost indistinguishable 
by LOCC operations from some separable 
(insecure) states. This fact has been shown recently to rule them out as a potential resource for swapping of a private key, 
in the so called quantum key repeaters \cite{BCHW-swapping}. \vspace{1em}

Given a Bell expression $\mathcal{S}$, we denote by  $\mathcal{S}_c$, $\mathcal{S}_q$ and $\mathcal{S}_{NS}$ respectively the classical, quantum
and no-signaling value of this expression. 
For a particular bipartite  state $\rho_{AB}$ and local POVMs $\DE{M_x^a}$ and $\{M_y^b\}$, we represent the corresponding box
by $\boxp \equiv\DE{\Tr (M_x^a\otimes M_y^b\, \rho_{AB})}$. The value of the Bell expression $\mathcal{S}$ for these particular POVMs and 
state is 
denoted 
\begin{align}
\mathbf{S}(\rho_{AB})= \Tr \mathbf{S}\rho_{AB},
\end{align}
where
\begin{align}
 \mathbf{S} = \sum_{a,b,x,y} s^{a,b}_{x,y} M_x^a\otimes M_y^b
\end{align}
is the Bell operator for POVMs $\DE{M_x^a}$ and $\{M_y^b\}$.
Note that the Bell operator $\mathbf{S}$ is a separable operator
(Definition \ref{defsepoperator}), so, intuitively, we expect that if a given state is hardly distinguishable from some separable one by means of separable operations,
it cannot exhibit
large violation in any Bell scenario, or else, one could use the  
procedure of checking the violation of a Bell inequality to discriminate between 
these two states. We now make this idea quantitative and state our first result, relating the value of a Bell expression 
on two bipartite states to the their
distinguishability by means of $PPT$ operations.

\begin{lemma} \label{thm:2-state-disc}
Given two states $\rho,\sigma \in D(\C^{d}\otimes\C^{d})$, a Bell expression $\mathcal{S}=\DE{s^{a,b}_{x,y}}$
and a set of  POVMs $\DE{ M_x^a},\{M_y^b\}$, it holds that:
\begin{align}
|\mathbf{S}(\rho) - \mathbf{S}(\sigma)| \leq \norm{\mathbf{S}^{\Gamma}}_{\infty} \norm{\rho^{\Gamma} - \sigma^{\Gamma}}_1.
\end{align}
where $\norm{\cdot}_1$ denotes the trace norm, $\norm{X}_{\infty}$ is the largest eigenvalue
in modulus of operator $X$ (which is equivalent to the spectral norm), and  $\Gamma$ denotes partial transposition.
\end{lemma}

Note that $\mathbf{S}^{\Gamma}$ is
also a Bell operator, since partial transposition maps the POVMs $\{ M_x^a\otimes M_y^b\}$ into another set of allowed measurements 
$\{ M_x^a\otimes (M_y^b)^T\}$, 
and  $\norm{\mathbf{S}^{\Gamma}}_{\infty}$ is nothing but the largest quantum value of the Bell expression $\mathbf{S}$ given the particular 
measurements $\{ M_x^a\otimes (M_y^b)^T\}$.
The second term on the RHS represents the distinguishability
of these two states by  means of PPT operations. Calculating distinguishability by separable operations is a hard problem since 
there is no simple description for the set of separable operations. A relaxation of this problem is to consider the set 
of PPT operations which is much simpler to characterize \cite{RainsPPToperations}.
A weaker form of Lemma \ref{thm:2-state-disc}, concerning distinguishability under
general measurements, was similarly derived in Ref. \cite{BrunnerVertesiTh1}.

\begin{proof} \begin{align}
|\mathbf{S}(\rho) - \mathbf{S}(\sigma)| = & |\sum_{a,b,x,y} \Tr s^{a,b}_{x,y}  M_x^a\otimes M_y^b (\rho - \sigma)|  \nonumber \\
 =&|\sum_{a,b,x,y} \Tr s^{a,b}_{x,y} M_{x}^a\otimes (M_{y}^b)^T (\rho - \sigma)^{\Gamma}|  \nonumber\\
=& | \Tr\, \mathbf{S}^{\Gamma}(\rho^{\Gamma} -\sigma^{\Gamma}) | \\
\leq & \Tr\, |\mathbf{S}^{\Gamma}(\rho^{\Gamma} -\sigma^{\Gamma}) | \nonumber \\
\leq &\norm{\mathbf{S}^{\Gamma}}_{\infty}\norm{\rho^{\Gamma} - \sigma^{\Gamma}}_1.\nonumber
\end{align}
In the first and second equality we used the linearity of the trace, then the identity $\Tr XY = \Tr X^{\Gamma}Y^{\Gamma}$. 
In the fourth step we used the triangle inequality, and the last step follows from
H\"{o}lder's inequality for $p-$norms, which states that $\norm{XY}_1 \leq \norm{X}_{\infty} \norm{Y}_1$.
\end{proof}

As we have discussed in Chapter \ref{chapterNL}, separable states can only generate 
local boxes. Therefore, we have that
\begin{align}
\mathbf{S}(\sigma_{AB}) \leq \mathcal{S}_c \;\; \forall \; \sigma_{AB} \in SEP.
\end{align}
Given this observation, we can derive the main result of this Section.

\begin{theorem} \label{thm:C-Q-bound}
  For any bipartite Bell expression $\mathcal{S}$, and state $\rho$, it holds that:
\begin{align}\label{eq:C-Q-bound}
\mathcal{S}(\rho) \leq \mathcal{S}_c + \mathcal{S}_q\inf_{\sigma \in SEP} \norm{\rho^{\Gamma}  - \sigma}_1 .
\end{align}
where
\begin{align}
\mathcal{S}(\rho) := \sup_{\DE{ M_x^a}, \{M_y^b\}} \sum_{a,b,x,y}s^{a,b}_{x,y}\Tr ( M_x^a\otimes M_y^b\, \rho),
\end{align}
with supremum taken over all POVMs $\DE{ M_x^a}$ and $\{M_y^b\}$.  
\end{theorem}

\begin{proof}
By substituting any separable state  $\sigma$ in Lemma \ref{thm:2-state-disc}, and using the fact that
 $\mathbf{S}(\sigma_{AB}) \leq \mathcal{S}_c \;\; \forall \; \sigma_{AB} \in SEP$
 we have: 
\begin{align}  
 \mathbf{S}(\rho) \leq \mathcal{S}_c + \norm{\mathbf{S}^{\Gamma}}_{\infty}\norm{\rho^{\Gamma}  - \sigma^{\Gamma}}_1.
 \end{align}
Now taking supremum over POVMs $\{ M_x^a\}$, $\{M_y^b\}$
 on both sides, and infimum over all separable states $\sigma$, we have the desired result .
Note that $\mathcal{S}_q = \sup_{\rho}\mathcal{S}(\rho)$, and   
that  $\norm{\mathbf{S}^{\Gamma}}_{\infty}$ is upper
bounded by $\mathcal{S}_q$.
\end{proof}

Theorem \ref{thm:C-Q-bound} shows that, given a Bell inequality $\mathcal{S}$, the best violation one can achieve with a particular quantum state $\rho$ 
 cannot outperform the classical bound by the quantum value of the Bell inequality shrunk 
by a factor reporting the distinguishability of state $\rho$ from the set of separable states by means of $PPT$ operations.
Since the bound \eqref{eq:C-Q-bound} only depends on the distance to the set of separable states, every state which is 
distinguishable from a separable state by the same content $\epsilon$ will exhibit the same limitations in a Bell scenario.
Hence we can 
define the sets 
\begin{align} \label{D-epsilon}
D(\epsilon) : = \{ \rho : \exists_{\sigma \in SEP} \,\, \norm{\rho^{\Gamma} - \sigma}\leq \epsilon\}.
\end{align}
Observe that $D(\epsilon)$ is a convex set, which includes the set of separable states $SEP$ for any $\epsilon >0$.
Due to Theorem \ref{thm:C-Q-bound}, we have the following dependence (see Fig. \ref{fig:d-epsilon}):
\begin{align} 
\sup_{\rho \in D(\epsilon)} \mathcal{S}(\rho) \leq \mathcal{S}_c + \epsilon \mathcal{S}_q.
\end{align}

\begin{figure}[h]
\begin{center}
 \includegraphics[scale=0.4]{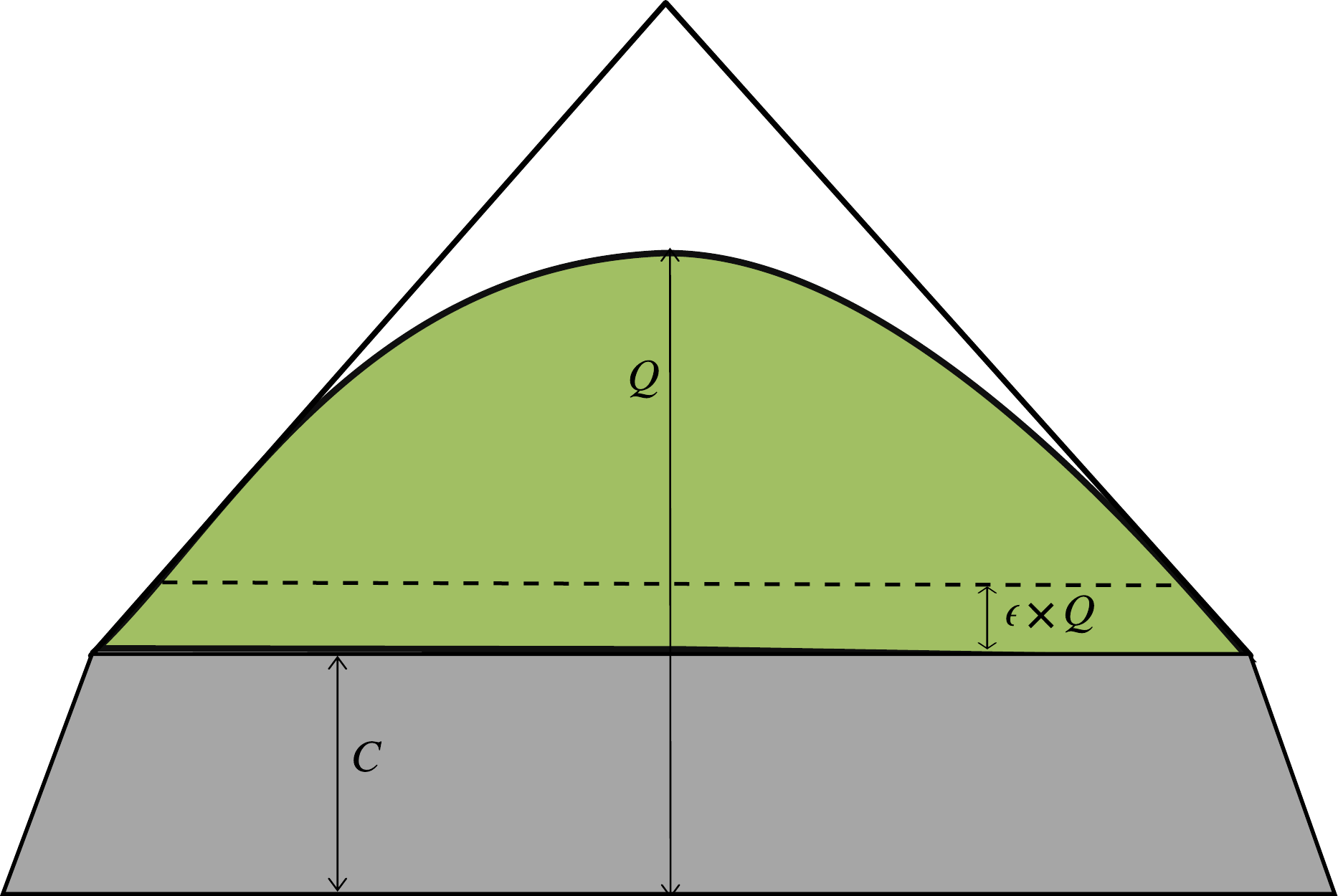}
 \caption{For states in $D(\epsilon)$, Eq. \eqref{D-epsilon}, the 
	violation of a Bell inequality $\mathcal{S}$ is limited by its quantum value $\mathcal{S}_q$ shrunk by $\epsilon$ (dashed line).
		}
		\label{fig:d-epsilon}
 \end{center}
 \end{figure}

\section{Bound on the asymptotic scenario} 

In the previous section we have considered the standard Bell scenario, where a single copy of a bipartite state 
is shared between Alice and Bob. 
However, if the state $\rho_{AB}$ is distillable, as it was noted by Peres \cite{Peres1996}, a pre-processing of many copies of a state by 
LOCC, before the Bell test, could lead to the violation of a Bell inequality, even for states that have local model 
in the standard scenario.
Here we quantify the asymptotic nonlocality by defining the 
asymptotic relative entropy of nonlocality 
and applying methods of Ref. \cite{BCHW-swapping} to bound it. In the first step, we will bound this quantity
by a function of the relative entropy distance under restrictive measurements introduced by Piani in Ref. \cite{Pianirelentropy}.

In Ref. \cite{vDamGrunwaldGill} a  nonlocality quantifier, based on the relative entropy, was introduced. 
The relative entropy between two probability distributions, $P$ and $Q$, is defined as $D(P||Q) \equiv \sum_i P(i) \log \frac{P(i)}{Q(i)}$.
The \emph{statistical strength of nonlocality} defined in Ref. \cite{vDamGrunwaldGill} captures quantitatively how 
``similar''  a given probability
distribution is to a local one: Given
a box $\vec{P}(ab|xy)$, where for fixed $x,y$ we have distribution $P_{xy}(ab|xy)$, its nonlocality is quantified by
\begin{align}
\mathcal{N}(\vec{P}) = \sup_{\{p(x,y)\}} \inf_{P_L \in \mathcal{L}} \sum_{x,y} p(x,y) D(P_{xy}(ab|xy) || P_L(ab|x,y)),
\end{align}
where infimum in the above is taken over all local boxes, $\vec{P}(ab|xy)\in \mathcal{L}$, for the particular scenario and $D(P||Q)$  
is the relative entropy.

Now, we are interested in quantifying  how much nonlocality $\mathcal{N}$ one can obtain from $n$ copies of a given 
state $\rho_{AB}$, per number of copies, 
in the asymptotic limit, after processing it by LOCC. For that we introduce the \emph{asymptotic relative entropy of nonlocality}.

{\definition For a bipartite state $\rho_{AB}$ its asymptotic relative entropy of nonlocality, $R(\rho_{AB})$, is given by:
\begin{align}
R(\rho_{AB}) \equiv \varlimsup_{n \rightarrow \infty}\frac{1}{n} \sup_{\Lambda \in LOCC} \sup_{\DE{M_{xy}}} \mathcal{N}(\{\Tr M_{xy}\Lambda(\rho_{AB}^{\otimes n})\}),
\end{align}
where $\varlimsup$ is the supremum limit, and $\{M_{xy}\}$ denote local POVMs $M_{xy}=M_x^a \otimes M_y^b$.
}

The asymptotic relative entropy of nonlocality captures the idea of an optimization over all Bell scenarios after a pre-processing of many copies
of a quantum state $\rho_{AB}$ by LOCC operations. 

We want to set bounds for the nonlocality attainable in the asymptotic scenario.
In order to state the bound,
we will need a well known entanglement measure called \emph{relative entropy of entanglement} \cite{DefER}:
\begin{align}
E_r(\rho) = \inf_{\sigma \in SEP} S(\rho||\sigma),
\end{align}
where $S(\rho||\sigma) = \Tr\rho\log \rho - \Tr\rho\log \sigma$ is the quantum relative entropy.




Now we are ready to state the main result of this Section.
\begin{theorem}\label{thmboundasymp}
 For any bipartite state it holds that
\begin{align}
R(\rho_{AB}) \leq  E_r(\rho_{AB}). 
\label{eq:RT}
\end{align}
For $\rho_{AB}$ a PPT state, it holds that
\begin{align}
R(\rho_{AB}) \leq \min\DE{E_r(\rho_{AB}), E_r(\rho_{AB}^{\Gamma})}.
\end{align}
\label{thm:R-T}
\end{theorem}

The upper bound $R(\rho_{AB}) \leq E_r(\rho_{AB})$ gives meaningful results only when the state is close to
separable states under global operations. More important is the second bound, which can give 
meaningful results even for some states which are highly distinguishable from separable states by global operations, but cannot be
distinguished if only PPT operations are allowed.
We refer the reader to Ref. \cite{GMstatebound} for the proof of Theorem \ref{thmboundasymp}.

We can also extend Theorem \ref{thmboundasymp} to asymptotic scenarios  where the parties 
can perform a `filtering'
operation (a non-trace-preserving map) before the Bell test: the so called hidden nonlocality scenario, introduced by Popescu in Ref. \cite{hiddenPopescu}.
Popescu showed that by performing a `filtering'
operation  it is possible to obtain a much larger violation of a Bell inequality on the resulting state. 
However, we note that, in order to quantify nonlocality in this scenario, it is 
important to take into account the probability of obtaining the  `filtered' result. 
For this reason, we propose to consider
the \emph{asymptotic relative entropy of hidden-nonlocality}, $R_H(\rho_{AB})$, defined as follows:
\begin{align}
R_H(\rho_{AB})  \equiv  \varlimsup_{n \rightarrow \infty}\frac{1}{n} 
\sup_{\Lambda \in LOCC}  \sup_{\DE{M_{xy}}} \sup_{F_0} p^{F_0} \mathcal{N}(\{\Tr M_{xy}F_0(\Lambda(\rho_{AB}^{\otimes n}))\}).
\end{align}
Where we can see a filtering process $F_0$ as an operation that takes the state $\Lambda(\rho_{AB}^{\otimes n})$ 
to a flag form
\begin{align}
F(\rho)=\sum_i \ketbra{i}{i}\otimes F_{i} \rho F_{i}^{\dagger},
\end{align}
 and later  erasures all other results except the  ``good'' 
one, that leads to the highest violation of the Bell inequality. The probability that the filter results in the desired outcome
is given by
\begin{align}
 p^{F_0}= \Tr F_{0}\Lambda(\rho_{AB}^{\otimes n})F_0^{\dagger}.
\end{align}
We can have the same bound for $R_H$, as for $R$.

\begin{theorem}\label{thmboundasympH}
For any bipartite state $\rho_{AB}$ it holds that
\begin{align}
R_H(\rho_{AB}) \leq  E_r(\rho_{AB}).
\end{align}

For a bipartite PPT state $\rho_{AB}$ it holds that
\begin{align}
R_H(\rho_{AB}) \leq \min  \DE{E_r(\rho_{AB}),E_r(\rho_{AB}^{\Gamma})}.
\end{align}
\label{ThmRH}
\end{theorem}

\section{Examples}

In Ref. \cite{pptkey}, it was shown that the general form of a quantum state from which it is possible to extract one bit of classical private key,
a \emph{private bit}, is given by
\begin{align*}
\gamma_{X}=\frac{1}{2}[\ketbra{00}{00}\otimes \sqrt{XX^{\dagger}}+ \ketbra{00}{11}\otimes X +  \nonumber\\
\ketbra{11}{00}\otimes X^{\dagger} +  \ketbra{11}{11}\otimes\sqrt{X^{\dagger}X} ] ,
\end{align*}
where $X$ is an arbitrary operator with trace norm 1 acting on $\mathbb{C}^{d_s}\otimes \mathbb{C}^{d_s}$. . 

Consider a private state 
defined by $X = \frac{1}{d^2}\sum_{i,j=0}^{d-1} \ketbra{ij}{ji}$ being the (normalized) swap operator. 
Then for the CHSH inequality we have the following bound \cite{karol-PhD}:
\begin{align}
Q_{CHSH}(\gamma_{X}) \leq 2 + \frac{\sqrt{2} + 1}{{2}\sqrt{2}{d}}.
\label{eq:1-example}
\end{align}

In \cite{PrivateNonlocality} it was shown that all perfect private states violate the CHSH inequality.
With our techniques we can bound this violation and also analyze PPT approximate private states.

Consider the following PPT approximate private state:
\begin{align}
\rho_{p} =(1-p) \gamma_{X} + \frac{p}{2}[\ketbra{01}{01}\otimes \sqrt{YY^{\dagger}} + \ketbra{10}{10}\otimes \sqrt{Y^{\dagger}Y}] 
\label{eq:ppt-1example}
\end{align}
with $X= \frac{1}{d_s \sqrt{d_s}} \sum_{i,j=0}^{d_s-1} u_{ij}\ketbra{ij}{ji}$ and $Y = \sqrt{d_s}X^{\Gamma}$, 
$|u_{ij}| = \frac{1}{d_s}$, and $p=\frac{1}{\sqrt{d_s}+1}$. 
By Theorem \ref{thm:C-Q-bound}, we have
\begin{align}
Q_{\mathcal{S}}(\rho_{p}) \leq \mathcal{S}_c + \mathcal{S}_q \frac{1}{\sqrt{d_s}}.
\label{eq:ppt-example}
\end{align}

In Ref. \cite{pptkey}, it was shown that there exists a family of states 
invariant under partial transposition for which the distillable key can be made arbitrarily close to one, $K_D \rightarrow 1$.
By applying Theorem \ref{thm:C-Q-bound} to this family of states we obtain the following proposition.

\begin{proposition}
 There exist bipartite states $\rho \in B(\C^2\otimes \C^2 \otimes (\C^{d^k}\otimes\C^{d^k})^{\otimes m})$ with $d=m^2$, $k=m$ satisfying
$K_D(\rho^{\Gamma}\otimes\rho) \rightarrow 1$ with increasing $m$, such that: 
\begin{align}
Q_{\mathcal{S}}(\rho\otimes\rho^{\Gamma})\leq \mathcal{S}_c +\frac{ \mathcal{S}_q}{2^{m-1}}.
\label{eq:prop-example}
\end{align}
\label{prop:example}
\end{proposition}

Proposition \ref{prop:example} shows that for the class of states in consideration, although the rate of  
distillable key can be made
arbitrarily close to one by increasing the dimension of the systems, the possibility of violating any Bell inequality is bounded by an amount
vanishing  
with the dimension of the system. 
For a bipartite Bell inequality with $n$ inputs and $k$ outputs it holds that
$ \mathcal{S}_q \leq \mathcal{S}_c \times \min \DE{n,k}$,
up to some universal constant independent of the parameters of the scenario \cite{Palazuelos2011}. 
Therefore, Theorem \ref{thm:C-Q-bound} ensures that for any fixed Bell scenario, as we wish to increase the key rate obtained from the exhibited family
of states,  
the possibility of observing a violation of a Bell inequality vanishes.\vspace{1em}

%

An application of Theorems \ref{thmboundasymp} and \ref{thmboundasympH}  follows from the fact 
that the relative entropy of entanglement, $E_r$, is asymptotically continuous\footnote{A function $f:D(\C^d)\longrightarrow \R$ is asymptotically continuous if
\begin{align}
 |f(\rho_1)-f(\rho_2)|\leq K \norm{\rho_1-\rho_2}_1 \log d+g(\norm{\rho_1-\rho_2}_1)
\end{align}
where $K$ is a constant and $g(\epsilon)\xrightarrow{\epsilon \longrightarrow 0}0$.
}\cite{Asympcontinuous}. Generally, for
$\rho_{\epsilon} \in D(\epsilon)$, and $\epsilon < \frac{1}{2}$, we have:
\begin{align*}
R_{(H)}(\rho_{\epsilon}) \leq 4\epsilon \log d + 2 h(\epsilon)
\end{align*}
Hence, if $\epsilon$ decreases with $d$ faster 
than $\frac{1}{\log d}$,  the asymptotic relative entropy of nonlocality vanishes with increasing dimension
The family of states shown in Eq. \eqref{eq:ppt-1example} have this property. 

\section{Discussion and open problems} 

We have presented bounds on the quantum nonlocality of a bipartite state, both, in the single copy case,
as well as in the asymptotic  scenarios. 
Although we use partial transposition techniques for obtaining nontrivial results, our method is based on the concept of
state discrimination via a restricted classes of operations: the separable ones. 

As future directions exploring the bounds: For the single copy scenario, instead of discrimination
from separable states, a refinement would be to consider the distance from states admitting a local hidden-variable model,
which could possibly lead to tight\-er bounds; 
For the asymptotic scenarios, it would be interesting to find nontrivial upper bounds on the asymptotic relative entropy of (hidden) 
nonlocality of NPT states. Note that, for these states, we only have $E_R(\rho)$ as an upper bound, which is small only for the cases
where the state is almost indistinguishable from a separable state under global operations.

In Refs. \cite{HHHLO:unco-pbit, HLLO-PRL} it was shown that one can launch quantum key distribution (QKD) protocols 
based on shared private bits. However, if we are interested in device-independent quantum key distribution (DI QKD) 
protocols we need to have a  violate of some Bell inequality.
A DI QKD protocol
is based on some Bell inequality $\mathcal{S}$,
and admits some level of violation, say $\epsilon_v$, below which it aborts.
Now, due to Eqs. \eqref{eq:ppt-example} and \eqref{eq:prop-example} there are (approximate) private bits, which exhibit violation of inequality
$\mathcal{S}$ only up to $\epsilon' < \epsilon_v$, and hence will be aborted. This rules out such states from usage in this particular
DI QKD protocol. Moreover, every realization of a DI QKD has inevitable errors. In such a case, the level of violation
$\epsilon'$ can be even below the precision of the experiment.
An interesting question for further investigation is the difference in terms of key rates between QKD and 
 DI QKD protocols.

It is worth noting, that our results are strongly related to the so called {\it Peres conjecture} \cite{Peres-conjecture}, 
recently disproved in \cite{Vertesi-Brunner-disproving}. Namely, we have asked a quantitative rephrasing of the 
original question posed by Asher Peres: how much one can violate a Bell inequality with PPT states? We have shown that,
for {\it certain} PPT states, the level of violation, both for single copy as well as in terms of the relative entropy of (hidden) nonlocality in the asymptotic cases, 
can be negligible. 
Notably, as we showed in the examples, even some states containing privacy admit such limited nonlocality content.

%% file: apendix.tex
%
%
%
%
%
%
%
%
%
%
%
%

\chapter{A group of facts about groups}\label{A-groups}

In Chapters \ref{chaptergamesd} and \ref{chapternplayer} we 
have discussed a class of games associated to finite Abelian groups. Here we state some results and properties of finite Abelian groups. For more details
the reader in referred to the book \cite{FourierTerras}.

\section{Some definitions}

\begin{Adefinition}[Group]
A group $(G,\circ)$ is a set of elements $a \in G$  with the associated operation $\circ$ satisfying
\begin{enumerate}[(i)]
 \item $a\circ b\in G$, $\forall a,b \in G$,
 \item $\exists \; e$ s.t. $a\circ e=a ,\, \forall a \in G$,
  \item $\forall \; a, \exists \; a^{-1}$ s.t. $a\circ a^{-1}=a^{-1}\circ a=e$,
   \item Associativity holds: $\forall a,b,c\,\in G$
  \begin{align*}
   a\circ (b\circ c)&=(a\circ b)\circ c\\
     \end{align*}
\end{enumerate}
\end{Adefinition}

\begin{Adefinition}[Abelian group]
An Abelian group $(G,+)$ is a group with associated operation denoted by $+$ that, additionally to properties $(i)$-$(iv)$, also satisfies commutativity:
\begin{align}
a+ b=b+ a\;\;\forall a,b \in G .
\end{align}
\end{Adefinition}

\begin{Adefinition}[Cyclic group]
 A group is called cyclic if it can be generated by a single element $g \in G$, called the \emph{generator} of the group. 
\end{Adefinition}
A finite cyclic group is a group with the property that all the elements
can be obtained by starting with the generator $g$ and applying the group operation to it many times,
\ie $G=\DE{g, (g\circ g), ((g\circ g)\circ g),\ldots}$. This makes it clear that all finite cyclic groups are Abelian groups.


\begin{Adefinition}[Homomorphism]
 An homomorphism of a group $(G,\circ )$ into a group $(H,\circ )$ is a function
 $f: G \longrightarrow H$ such that
 \begin{align}
  f(a\circ b)=f(a)\circ f(b) \;\; \forall \; a,b\in G
 \end{align}
\end{Adefinition}

A remarkable example of these concepts is given by the cyclic group  $\mathbb{Z}_{d}$, for which:
\begin{itemize}
 \item $G=\DE{0,\ldots,d-1}$.
 \item The associated operation $+$ is the sum modulo $d$.
 \item The identity element is $e=0$.
 \item The inverse element of $a$ is $a^{-1}=-a=d-a$.
 \item $g=1$ is a generator\footnote{A cyclic group can have more than one generator. For example, for prime $d$, every  
 element of $\mathbb{Z}_d$ is a generator. } for all $d$.
 \item An homomorphism from $\mathbb{Z}_{d}$ to the group $\mathbb{T}$  of unitary complex numbers with the product operation is given by
 \begin{align}
  a\longmapsto e^{\frac{2\pi \mathrm{i}a}{d}}\;\;\forall\; a\in \DE{0,\ldots,d-1} .
 \end{align}
\end{itemize}

\section{The characters of an Abelian group}

\begin{Adefinition}\label{Adefcharacter}
 A character $\chi_j$ of a finite Abelian group $(G,+)$ is a group homomorphism from $(G,+)$ to 
 the group $\mathbb{T}$  of unitary complex numbers with the product operation:
 \begin{align}
  \chi_j: a \mapsto \chi_j(a).
 \end{align}
\end{Adefinition}

By Definition \ref{Adefcharacter} we have that  the characters of an Abelian group are complex numbers
satisfying the properties:
\begin{align}\label{Aeqcharacprop}
\begin{cases}
\text{Homomorphism:} & \chi_j(a+b)=\chi_j(a)\chi_j(b) \; \; \forall a,b \in G,\\
\text{Reflexivity:} & \bar{\chi}_j(a)=\chi_j(-a),\\
\text{Orthogonality:} & \sum_{a \in G}\chi_i(a)\bar{\chi}_j(a)=|G|\,\delta_{i,j}.
\end{cases}
\end{align}

For the cyclic group $\mathbb{Z}_{d}$, 
the characters are the $d$ roots of unit 
\begin{align} 
\chi_j(a)=\zeta^{ja} \;\; \text{for}\;\;j\in\DE{0,\ldots,d-1},
\end{align}
where $\zeta=e^{2\pi \mathrm{i}/d}$.

A very interesting result is the called \emph{fundamental theorem of finite Abelian groups} (see Ref. \cite{FourierTerras}), 
which states that any finite Abelian group $(G,+)$ can be 
seen as a direct product of cyclic groups
\begin{equation}
(G,+) \cong \mathbb{Z}_{d_1} \times \mathbb{Z}_{d_2} \times \dots \times \mathbb{Z}_{d_r},
\end{equation}
where every element $x \in G$ is as a $r$-tuple $(x_1, \dots, x_r)$ with $x_i \in \mathbb{Z}_{d_i}$ and the operation in $(G,+)$ is 
given by
\begin{align}
 (x_1, \dots, x_r)+(y_1, \dots, y_r)=(x_1+y_1,\ldots,x_r+y_r).
\end{align}

With this characterization, the characters $\chi_j$ of the Abelian group $(G,+)$, can be written as the product of the characters of the 
cyclic groups that compose $(G,+)$:
\begin{equation}
\chi_{k}(a) = \prod_{j=1}^{r} \chi_{k_j}(a_j),
\end{equation} 
where $\chi_{k_j}(a_j) = e^{\frac{2 \pi \mathrm{i} k_j a_j}{d_j}}$ is a $d_j$-th root of unity, and $a_j \in \mathbb{Z}_{d_j}$. 

\section{The Fourier transform over finite Abelian groups}\label{A-Fourier}

A Fourier transform can be seen as a change of basis, 
where we pass from a description of a function $f$ in terms of a set of variables 
to a description in terms of different variables, but keeping all the information about $f$. 
When we have a function of the
elements of a finite Abelian group, the Fourier transform can be expressed in terms of the characters of the group.

Let $f$ be a complex valued function on the finite Abelian group $(G,+)$
\begin{align*}
 f: (G,+) \longrightarrow \C.
\end{align*}
The Fourier transform of 
$f$ is defined as :
\begin{align}
 \hat{f}(\chi_i)=\sum_{a\in G}\bar{\chi}_i(a)f(a),
\end{align}
with the inverse  given by
\begin{align}
 f(a)=\frac{1}{|G|}\sum_{j\in G}\chi_i(a)\hat{f}(\chi_i).
\end{align}

For non-Abelian groups the characters are more subtle, and the Fourier transform uses the irreducible representations of the group.
For the reader interested in a further
reading about Fourier transform on groups we refer to Ref. \cite{FourierTerras}.

\section{Finite Fields}\label{A-finitefield}

Fields are sets which have more structure than an Abelian group. While for an Abelian group $(G,+)$ we only have the associated sum operation $+$, 
a field $\mathbb{F}_d$ has two associated
operations $+$ and $\cdot$ satisfying the properties specified bellow.

\begin{Adefinition}[Finite Field]
A finite field $\mathbb{F}_d$ is a set of $d$ elements with the operations sum $+$ and multiplication $\cdot$ such that
\begin{enumerate}[(i)]
 \item $a+b$, $a \cdot b$ $\in \mathbb{F}_d$, $\forall a,b \in \mathbb{F}_d$,
 \item $\exists \; 0$ s.t. $a+0=a ,\, \forall a \in \mathbb{F}_d$,
  \item $\exists \; 1$ s.t. $a\cdot 1=a ,\, \forall a \in \mathbb{F}_d$,
  \item $\forall \; a, \exists \; -a$ s.t. $a+(-a)=0$,
  \item $\forall \; b \neq 0, \exists \; b^{-1}$ s.t. $b\cdot b^{-1}=1$,
  \item Associativity: $\forall a,b,c\,\in \mathbb{F}_d$
  \begin{align*}
   a+(b+c)&=(a+b)+c\\
   a\cdot (b\cdot c)&= (a \cdot b) \cdot c
  \end{align*}
\item Commutativity: $\forall a,b\,\in \mathbb{F}_d$
\begin{align*}
   a+b&=b+a\\
   a\cdot b&=  b \cdot a
  \end{align*}
 \item Distributivity: $\forall a,b,c\,\in \mathbb{F}_d$
 \begin{align}
  a\cdot (b+c)= a \cdot b + a \cdot c.
 \end{align}
\end{enumerate}
\end{Adefinition}

Finite fields can only be defined for a set of $d$ prime or a power of prime elements.
For $d$ prime all the conditions $(i)$-$(viii)$ can be satisfied by arithmetic modulo $d$.
For $d=p^r$ the arithmetic operations\footnote{See \url{https://en.wikipedia.org/wiki/Finite_field_arithmetic}.} can be defined by addition and multiplication of polynomials of degree $< r$ over $\mathbb{Z}_p$, $\mathbb{Z}_p\De{X}$. 
In order to construct a field $\mathbb{F}_{p^r}$ one can start by choosing an irreducible polynomial of degree $r$ over  $\mathbb{Z}_p$, this polynomial will
define the zero of the field  by the so called quotient. 

As an example for the field $\mathbb{F}_d$ with $d=2^3$ we can pick the polynomial $X^3+X+1 \in  \mathbb{Z}_2\De{X}$ from 
which we can obtain the relation 
\begin{align}\label{spirred}
X^3+X+1 =0 \Rightarrow X^3=X+1.
\end{align}
Now the elements of the field can be represented by strings $(a,b,c)$, with $a,b,c \in \DE{0,1}$, and we can associate each string with
the polynomial $aX^2+bX+c$. Given that, addition and multiplication will be taken as addition and multiplication of polynomials reduced by  the 
relation \eqref{spirred}. 


\chapter{Functional boxes and multipartite communication complexity}\label{A-CC}
\chaptermark{Communication complexity}

Here we discuss functional boxes and the communication complexity task in the multipartite scenario. These results were also
presented in Ref. \cite{GMmultiplayer}.

We start by defining the $PR_n$-$d$ boxes, a generalization of $PR$ boxes \cite{prbox} for $n$ parties and prime $d$ outputs.

\begin{Adefinition}
 \begin{align}
  PR_n\text{-}d(\vec{a}|\vec{x})=\begin{cases}
                                          {\frac{1}{d^{n-1}}}, \; \text{if} \; a_1+ \ldots + a_n=x_1 \cdot \ldots \cdot x_n\\
                                          0, \; \text{otherwise}
                                         \end{cases}
 \end{align}
 where $d$ is prime and sum $+$ and multiplication $\cdot$ are operations modulo $d$.
\end{Adefinition}

A multipartite communication complexity scenario consists of $n$ parties, denoted $A_1, \ldots, A_n$, where  each party $A_i$
receives an input $x_i$, and their goal is to exchange the least number of classical messages in order to compute the value of a global function
of their inputs $f(x_1,\ldots,x_n)$.

In what follows we prove that if the $n$ parties have access to a sufficient large amount of $PR_n$-$d$ boxes, they can compute any
function $f(x_1,\ldots,x_n)$ with only $n-1$ dits of communication: where each player $A_{i\neq 1}$ communicate one dit to player $A_1$ which
them computes the function.

\begin{Atheorem}\label{sPRnd}
 If $n$ parties are allowed to share an arbitrary number of $PR_n$-$d$ boxes, any $n$-partite communication
 complexity problem can be solved with only $n-1$ dits of communication. 
\end{Atheorem}

\begin{proof}
Our proof is a straightforward generalization of Ref. \cite{PRd}.
We prove for the case $n=3$. The proof for general $n$ follows directly.

Let us consider that Alice receives input string  $\vec{x}\in\mathbb{Z}_d^{m_1}$, Bob
receives $\vec{y}\in \mathbb{Z}_d^{m_2}$ and Charlie receives $\vec{z}\in  \mathbb{Z}_d^{m_3}$.
We start by observing that any function 
$f(\vec{x},\vec{y},\vec{z})$, $f: \mathbb{Z}_d^{m_1} \times  \mathbb{Z}_d^{m_2} \times  \mathbb{Z}_d^{m_3}\rightarrow \mathbb{Z}_d$,
can be written as a multivariate polynomial with degree at most $d-1$ in each variable $x_i, y_j$ and $z_k$
\begin{align}
f(\vec{x},\vec{y},\vec{z})=\sum_{\vec{\alpha},\vec{\beta},\vec{\gamma}}\mu_{\vec{\alpha},\vec{\beta},\vec{\gamma}}\; {\vec{x}}^{\,\vec{\alpha}} \, \vec{y}^{\,\vec{\beta}} \, \vec{z}^{\,\vec{\gamma}} , 
\end{align}
where ${\vec{x}}^{\,\vec{\alpha}}=\prod_{i=1}^{m_1}x_i^{\alpha_i}$, ${\vec{y}}^{\,\vec{\beta}}=\prod_{j=1}^{m_2}y_j^{\beta_j}$, 
${\vec{z}}^{\,\vec{\gamma}}= \prod_{k=1}^{m_3}z_k^{\gamma_k}$, $\mu_{\vec{\alpha},\vec{\beta},\vec{\gamma}} \in \mathbb{Z}_d$, and
$\vec{\alpha}=(\alpha_1,\ldots,\alpha_{m_1}) \in \mathbb{Z}_d^{m_1}$. And analogously for $\vec{\beta}$ and $\vec{\gamma}$.
Now if the players have access to $r=d^{m_1   m_2  m_3}$ $PR_3$-$d$ boxes they can execute the following protocol in
order to compute $f(\vec{x},\vec{y},\vec{z})$:
\begin{enumerate}
 \item For each $(\vec{\alpha},\vec{\beta},\vec{\gamma})$ the players picks one $PR_3$-$d$,
 \item Alice inputs ${\vec{x}}^{\,\vec{\alpha}}=\prod_{i=1}^{m_1}x_i^{\alpha_i}$. Bob inputs ${\vec{y}}^{\,\vec{\beta}}=\prod_{j=1}^{m_2}y_j^{\beta_j}$. 
 Charlie inputs
 ${\vec{z}}^{\,\vec{\gamma}}=\prod_{k=1}^{m_3}z_k^{\gamma_k}$. And they get respectively the outputs $a_{\vec{\alpha}}$, $b_{\vec{\beta}}$ and $c_{\vec{\gamma}}$.
 \item Bob sets $b=\sum_{\vec{\alpha}}\sum_{\vec{\beta}}\sum_{\vec{\gamma}}\mu_{\vec{\alpha},\vec{\beta},\vec{\gamma}}b_{\vec{\beta}}$ and send $b$ to Alice.
 Charlie sets $c=\sum_{\vec{\alpha}}\sum_{\vec{\beta}}\sum_{\vec{\gamma}}\mu_{\vec{\alpha},\vec{\beta},\vec{\gamma}}c_{\vec{\gamma}}$ and send $c$ to Alice.
  \item Alice sets $a=\sum_{\vec{\alpha}}\sum_{\vec{\beta}}\sum_{\vec{\gamma}}\mu_{\vec{\alpha},\vec{\beta},\vec{\gamma}}a_{\vec{\alpha}}$ and she
 computes $f(\vec{x},\vec{y},\vec{z})=a+b+c$.
\end{enumerate}
where only 2 dits were communicated in order to compute the function.

The generalization for $n$-party function follows in analogous way, where any function $f:\mathbb{Z}_d^{m_1} \times \ldots \times  \mathbb{Z}_d^{m_n}\rightarrow \mathbb{Z}_d$
will be written as a multivariate polynomial with degree at most $d-1$ in each variable and using an analogous protocol, with $n-1$ parties communicating
only one dit to the first party, the computation of function $f$ will be performed.
\end{proof}

Now we consider a natural generalization of bipartite functional boxes, introduced in Ref. \cite{PRd}, to the multipartite case:

\begin{Adefinition}
 For any function $f: \mathbb{Z}_d \times \ldots \times \mathbb{Z}_d \rightarrow \mathbb{Z}_d$, the multipartite functional box corresponding 
 to $f$ is defined as
 \begin{align}
  P_n^f(\vec{a}|\vec{x})=\begin{cases}
                                          {\frac{1}{d^{n-1}}}, \; \text{if} \; a_1+ \ldots + a_n= f(x_1,\ldots,x_n)\\
                                          0, \; \text{otherwise}
                                         \end{cases}
 \end{align}
\end{Adefinition}

Now we argue that all $n$-partite functional box with $f$ non-additively separable ($f$ is additively separable if 
$f(x_1,\ldots,x_n)=f_1(x_1)+f_2(x_2)+\ldots + f_n(x_n)$) would lead to some kind of trivialization of communication complexity.

\begin{Atheorem} 
 All $P_3^f$ with $f(x,y,z)$ such that there exists a partial derivative of some order equals to $\lambda \cdot x\cdot y\cdot z+g(x)+h(y)+s(z)$ can be 
 used to simulate a $PR_3$-$d$, and then can be used to solve any $3$-partite communication
 complexity problem  with only $2$ dits of communication.
\end{Atheorem}

\begin{proof}
First let us consider 
\begin{align}\label{sfPR3d}
 f(x,y,z)=\lambda \cdot x\cdot y\cdot z+g(x)+h(y)+s(z).
\end{align}
So by using box $P_3^f$ Alice, Bob and Charlie can input $x,y$ and $z$ respectively
and get outputs $a$, $b$ and $c$. Now following Ref. \cite{PRd} Alice sets $a'=\lambda^{-1}(a-g(x))$, Bob sets $b'=\lambda^{-1}(b-h(y))$ and Charlie sets $c'=\lambda^{-1}(c-s(z))$, 
so that we have
\begin{align}
 a'+b'+c'=x\cdot y \cdot z\,.
\end{align}
In order to randomize the results they can randomly chose $k \in \mathbb{Z}_d$ and output $a_f=a'+k$, $b_f=b'+k$ and $c_f=c'-2k$, so that they perfectly
simulate a $PR_3$-$d$ box.

Now for other functions $f$ we can use the method of Ref. \cite{PRd} of applying partial derivatives to the function. 
The partial derivative of $f$ with respect to $x$ is defined as
\begin{align}
 f_x(x,y,z)\equiv f(x+1,y,z)-f(x,y,z)
\end{align}
and it generates a polynomial with the degree in $x$ reduced by 1, while the degree in $y$ and $z$ remains the same or is smaller. And note that 
with two boxes $P_3^f$ we can simulate the box $ f_x(x,y,z)$.

Then if by partial derivatives we can reduce function $f$ to the form \eqref{sfPR3d} we have a protocol using a finite number of boxes $P_3^f$ to simulate
$PR_3$-$d$. By the result of Theorem \ref{sPRnd}, with an arbitrary finite number of $P_3^f$, we can solve  any $3$-partite communication
 complexity problem with only $2$ dits of communication.
 \end{proof}

If a function $f(x,y,z)$ is not additively separable it will contain at least one term involving product of two variables, for example $x^ry^s$ and this box 
can be reduced, by derivatives, into a box of the form $\lambda \cdot x \cdot y+g(x)+h(y)+s(z)$.
Now using the results for the bipartite case \cite{PRd}, if Charlie always 
inputs $z=1$, with only 2 dits of communication they can compute any function of two variables $f(x,y)$.

\chapter{Proofs of some results}\label{A-proofs}

\section{Some proofs on {\xor} games}\label{A-xorgames}
In this Section we prove some of the results stated in Chapters \ref{chaptergames} and \ref{chapterxor}.

%
%
%

\begin{Atheorem}\label{Athmboundxor}
Given a bipartite \xor-game, with $m$ inputs for Alice, $m$ inputs for Bob, and an associated game matrix $\Phi$, the quantum value is upper bounded by
\begin{align}
 \omega_{q}^{\oplus}\leq \frac{1}{2}\de{1+m|| \Phi||}.
\end{align}
\end{Atheorem}

\begin{proof}
 We now prove this result using  Lagrange duality. Let us consider the SDP characterization given by Theorem \ref{thmxorSDP}:
 \begin{align}\label{AeqxorSDP}
\epsilon_q =\begin{cases}         
 \max& \Tr{\Phi}_{s} \mathcal{X}	\\
\text{s.t.} &\text{diag}(\mathcal{X}) = |\1 \rangle, \\
         &\mathcal{X} \geq 0.
         \end{cases}
\end{align}
By weak duality (Theorem \ref{thmweakdual}), every feasible solution to the dual Lagrange problem provides an upper bound to the 
 quantum bias $\epsilon_q$. So let us consider the dual problem of \eqref{AeqxorSDP}:
 \begin{align}\label{Aeqxordual}
(\mathcal{D})\begin{cases}         
 \min \;& \sum_{i=1}^{2m}y_i \\
\text{s.t.} &\text{Diag}(y) \geq  {\Phi}_s.
         \end{cases}
\end{align}
Note that $y'_i=\norm{\Phi}$ for all $i$ is a feasible solution to problem $(\mathcal{D})$, since $\norm{{\Phi}_s}={\norm{\Phi}}/{2}$ and then
\begin{align}
 \text{Diag}(y')=\frac{\norm{\Phi}}{2}\I_{2m}\geq {\Phi}_s.
\end{align}
Therefore we have that
\begin{align}
 \epsilon_q=\sum_{i=1}^{2m}y'_i=m\norm{\Phi}
\end{align}

Now
\begin{align}
 \omega_q\leq\frac{1}{2}(1+\epsilon_q)=\frac{1}{2}\de{1+m\norm{\Phi}},
\end{align}
which ends the proof.
\end{proof}

For the case where $|\inp_A|=m_A$ and $|\inp_B|=m_B$, the same analysis results in the upper bound
\begin{align}
 \omega_q\leq\frac{1}{2}\de{1+\frac{m_A+m_B}{2}\norm{\Phi}},
\end{align}
which is worse then the bound given by Theorem \ref{thmboundxor}.
\vspace{1em}

\begin{theorem}\label{AthmAeigen}
The adjacency matrix of an {\xor}-game graph $\mathcal{G}(\Phi)$ 
\begin{align} \label{Aeqadjmatrix}
\begin{split}
\mathcal{A}(\mathcal{G}(\Phi)) =& \I_m \otimes(\proj{\1}-\I_m)\otimes \sigma_X+\half\proj{\mathbf{1}}\otimes\I_m\otimes(\I_2+\sigma_X)	\\
&-\half [D(\proj{\1}\otimes\I_m)D]\otimes(\I_2-\sigma_X)
\end{split}
\end{align}
has the following spectrum
and corresponding degeneracies:
\begin{equation}
\label{Aadj-spec}
\text{spec}(\mathcal{A}(\mathcal{G}(\Phi)))=\begin{cases}
2m-1 & \times  1 	\\
m-1 & \times 2m-2	\\
-1 & \times (m-1)^2 \\
1-m\pm \lambda_z & \times 1 \\
1 &  \times m(m-2) 
\end{cases}.
\end{equation}
where $\lambda_z$ denotes the $m$ singular values of $\tilde{\Phi}$. 
 \end{theorem}

 \begin{proof} 
 We start by nothing that the adjacency matrix, Eq. \eqref{Aeqadjmatrix}, can be written in the form
 \begin{align} \label{Aeqadjdecomp}
\begin{split}
\mathcal{A}(\mathcal{G}(\Phi)) =& \de{\I_m \otimes(\proj{\1}-\I_m)+\proj{\mathbf{1}}\otimes\I_m}\otimes \ketbra{+}{+}	\\
&-\de{ \I_m \otimes(\proj{\1}-\I_m)+[D(\proj{\1}\otimes\I_m)D]}\otimes\ketbra{-}{-},
\end{split}
\end{align}
which allows us to write it as a direct sum
 \begin{align}
\mathcal{A}(\mathcal{G}(\Phi)) = \mathcal{A}_1\oplus \mathcal{A}_2
\end{align}
where
\begin{align}
 \begin{split}
  \mathcal{A}_1&= \de{\I_m \otimes(\proj{\1}-\I_m)+\proj{\mathbf{1}}\otimes\I_m}\\
  \mathcal{A}_2&=\de{ -\I_m \otimes(\proj{\1}-\I_m)-[D(\proj{\1}\otimes\I_m)D]}.
 \end{split}
\end{align}
Therefore we can proceed to diagonalize $\mathcal{A}_1$ and $\mathcal{A}_2$ separately.

Let us start with the eigenvalues of $\mathcal{A}_1$. Note that $\mathcal{A}_1$ contains only the identity and the all 1's matrix $\proj{\1}$, therefore an eigenbasis
for $\mathcal{A}_1$ is formed by the $m^2$ Fourier vectors $\DE{\ket{v_i}\otimes \ket{v_j}}$, where
\begin{align}
|v_{j} \rangle = \left(1, \zeta^j, \zeta^{2j}, \dots, \zeta^{(m-1)j}\right)
\end{align}
with $j \in \{0, \dots, m-1\}$, where $\zeta = \exp{(2\pi I/m)}$. Since $\proj{\1}\ket{v_i}=m \delta_{i,0} \ket{v_i}$ we have that
\begin{itemize}
 \item $\ket{v_0}\otimes \ket{v_0}$ is an eigenvector of $\mathcal{A}_1$ with eigenvalue $2m-1$,
 \item $\ket{v_0}\otimes \ket{v_{i\neq0}}$ and  $\ket{v_{i\neq0}}\otimes \ket{v_0}$  are eigenvectors with eigenvalue $m-1$,
 \item $\ket{v_{i\neq 0}}\otimes \ket{v_{j\neq 0}}$ are eigenvectors  with eigenvalue -1.
\end{itemize}
which completes the diagonalization of $\mathcal{A}_1$.

Now let us proceed to find the eigenvalues of $\mathcal{A}_2$. First, note that,
if $\ket{\lambda^A_z}$ and $\ket{\lambda^B_z}$ are the singular vectors of the game matrix $\Phi$, corresponding to the singular value $\lambda_z$, 
then, by using the relation $\bra{\lambda^A_z}\Phi\ket{\lambda^B_z}=\lambda_z$, we can derive that
\begin{align}
\langle\lambda^A_z|\bra{j}D\ket{j}|\lambda^B_{z'}\rangle=\lambda_z\delta_{z,z'}.	\label{eqn:neat_trick}
\end{align}
By using relation \eqref{eqn:neat_trick} one can verify that the following $2m$ vectors 
\begin{align}
\ket{\eta_z^\pm}=\frac{\ket{\lambda^A_z}\ket{j}\pm D\ket{j}\ket{\lambda^B_z}}{\sqrt{2(m\pm\lambda_z)}},
\end{align}
are eigenvectors of $\mathcal{A}_2$, with respective eigenvalues $1-(m \pm \lambda_z)$.

The remaining eigenvalues of $\mathcal{A}_2$ are all equal to 1. This can be shown by the fact that subtracting 
an appropriate amount of the projector into the eigenvectors $\DE{\ket{\eta_z^\pm}}$ leaves us with identity, \ie the following
equality holds
\begin{align}
 \mathcal{A}_2+\sum_{z=1}^{m} (m+\lambda_z)\ketbra{\eta_z^+}{\eta_z^+}+\sum_{z=1}^{m} (m-\lambda_z)\ketbra{\eta_z^-}{\eta_z^-}=\I,
\end{align}
This completes the proof.
 \end{proof}

\section{On DIEWs}

Bellow we list the projective measurements  $\{M_x^a\}=\DE{ \ketbra{A_x^a}{A_x^a}}$, $\{M_y^b\}=\DE{\ketbra{B_y^b}{B_y^b}}$, 
$\{M_z^c\}=\DE{ \ketbra{C_z^c}{C_z^c}}$ 
that allow the players to win the generalized Mermin game  discussed in Section \ref{secdiew}, Eq. \eqref{eqgameghz},
with probability 1 when they share the $GHZ_3$ state:

 \begin{subequations}
 \label{sproj}
\begin{align}
\ket{A_0^0}=\ket{B_0^0}=\frac{1}{\sqrt{3}}\de{\,1\,,\zeta^{4/3},\,1\,}\;\;&,\;\;\ket{C_0^0}=\frac{1}{\sqrt{3}}\de{\,1\,,\zeta^{1/3},\,1\,}\;, \\
\ket{A_0^1}=\ket{B_0^1}=\frac{1}{\sqrt{3}}\de{\zeta^{2},\zeta^{7/3},\,1\,} \;\;&,\;  \; \ket{C_0^1}=\frac{1}{\sqrt{3}}\de{\zeta^{2},\zeta^{4/3},\,1\,} \;,\\
\ket{A_0^2}=\ket{B_0^2}=\frac{1}{\sqrt{3}}\de{\zeta,\zeta^{1/3},\,1\,}\;\;&,\;\; \ket{C_0^2}=\frac{1}{\sqrt{3}}\de{\zeta,\zeta^{7/3},\,1\,}
\end{align}
\end{subequations}
\begin{subequations}
\begin{align}
\ket{A_1^0}=\ket{B_1^0}=\frac{1}{\sqrt{3}}\de{\zeta^{1/3},\,1\,,\,1\,}\;\;&,\; \;\ket{C_1^0}=\frac{1}{\sqrt{3}}\de{\zeta^{1/3},\zeta^2,,\,1\,}\;,   \\
\ket{A_1^1}=\ket{B_1^1}=\frac{1}{\sqrt{3}}\de{\zeta^{7/3},\zeta,\,1\,} \;\;&,\; \; \ket{C_1^1}=\frac{1}{\sqrt{3}}\de{\zeta^{7/3},\,1\,,\,1\,} \;,   \\
\ket{A_1^2}=\ket{B_1^2}=\frac{1}{\sqrt{3}}\de{\zeta^{4/3},\zeta^{2},\,1\,}\;\;&,\;\; \ket{C_1^2}=\frac{1}{\sqrt{3}}\de{\zeta^{4/3},\zeta,\,1\,}
\end{align}
\end{subequations}
\begin{subequations}
\begin{align}
\ket{A_2^0}=\ket{B_2^0}=\frac{1}{\sqrt{3}}\de{\zeta^{8/3},\zeta^{8/3},1}\;\;&,\; \;\ket{C_2^0}=\frac{1}{\sqrt{3}}\de{\zeta^{8/3},\zeta^{5/3},1}\;,  \\
\ket{A_2^1}=\ket{B_2^1}=\frac{1}{\sqrt{3}}\de{\zeta^{5/3},\zeta^{2/3},1} \;\;&,\; \; \ket{C_2^1}=\frac{1}{\sqrt{3}}\de{\zeta^{5/3},\zeta^{8/3},1} \;,  \\
\ket{A_2^2}=\ket{B_2^2}=\frac{1}{\sqrt{3}}\de{\zeta^{2/3},\zeta^{5/3},1}\;\;&,\;\; \ket{C_2^2}=\frac{1}{\sqrt{3}}\de{\zeta^{2/3},\zeta^{2/3},1}
 \end{align}
 \end{subequations}
where $\zeta=e^{2\pi i/3}$.